\documentclass{article}[13pt]
\usepackage[backend=biber, style=alphabetic]{biblatex}

\usepackage{tikz}

\DeclareCiteCommand{\citeyear}
    {}
    {\bibhyperref{\printdate}}
    {\multicitedelim}
    {}

\newcommand{\citet}[1]{\citeauthor{#1}~(\citeyear{#1})}

\usepackage[margin=3cm]{geometry}
\usepackage{graphicx}
\usepackage{tabularx}
    \newcolumntype{L}{>{\raggedright\arraybackslash}X}

\usepackage{mathrsfs}

\usepackage{caption}

\usepackage[linkbordercolor=blue]{hyperref}

\usepackage{amsmath,amsfonts,amssymb,amsthm,bm} 
\usepackage{xcolor}
\usepackage{url}
\usepackage[export]{adjustbox}
\numberwithin{equation}{section}
\usepackage{xltabular}
\usepackage{booktabs}
\usepackage[font={footnotesize}]{caption}
\usepackage[en-GB,showdow=false]{datetime2}
\usepackage{bbm}
\usepackage{physics}
\usepackage{algorithm} 
\usepackage{algpseudocode} 
\usepackage{accents}
\usepackage{setspace}

\font\myfont=cmr12 at 18pt

\newtheorem{theorem}{Theorem}[section] 
\newtheorem{lemma}{Lemma}[section]
\newtheorem{corollary}{Corollary}[section]

\counterwithout{theorem}{section}

\renewcommand\thefootnote{} 

\begin{document}

\title{\myfont{No-Regret Learning for Stackelberg Equilibrium Computation in Newsvendor Pricing Games}}

\date{} 

\author{
    Larkin Liu\footnote{Technical University of Munich; \texttt{larkin.liu@tum.de}}  \\
    \and
    Yuming Rong\footnote{Google; \texttt{ryrong@google.com}} \\
}
\maketitle
\vspace{-1.4cm}

\begin{abstract}

We introduce the application of online learning in a Stackelberg game pertaining to a system with two learning agents in a dyadic exchange network, consisting of a supplier and retailer, specifically where the parameters of the demand function are unknown. In this game, the supplier is the first-moving leader, and must determine the optimal wholesale price of the product. Subsequently, the retailer who is the follower, must determine both the optimal procurement amount and selling price of the product. In the perfect information setting, this is known as the classical \textit{price-setting Newsvendor} problem, and we prove the existence of a unique Stackelberg equilibrium when extending this to a two-player pricing game. In the framework of online learning, the parameters of the reward function for both the follower and leader must be learned, under the assumption that the follower will best respond with optimism under uncertainty. A novel algorithm based on contextual linear bandits with a measurable uncertainty set is used to provide a confidence bound on the parameters of the stochastic demand. Consequently, optimal finite time regret bounds on the Stackelberg regret, along with convergence guarantees to an approximate Stackelberg equilibrium, are provided. 

\renewcommand\thefootnote{} 
\footnotetext{A extended abstract was accepted at the \textit{8th International Conference on Algorithmic Decision Theory} (ADT 2024).}
\renewcommand\thefootnote{\arabic{footnote}} 






\end{abstract}



\section{Introduction}

We study the problem of dynamic pricing in a two-player general sum Stackelberg game. An economic interaction between learning agents in a dyadic exchange network is modelled as a form of economic competition between two firms in a connected supply chain. Economic competition can arise between self-interested firms in the intermediate stages of a supply chain, especially when they are non-centralized or not fully coordinated, operating as individual entities \cite{perakis:2007price} \cite{cachon:1999competitive-supply-chain}. 

From a single agent perspective, online learning for dynamic pricing has been studied extensively, and known algorithms exist guaranteeing sub-linear regret \cite{goyal:2021mnl} \cite{bu:2022context-dyn-pric} \cite{babaioff:2015dyp-lim-supply}. \citet{kleinberg:2003value} first presented the problem of demand learning, where the firm must optimally learn the parameters of demand in an online learning setting. From a single-agent perspective, this has been previously studied in \citet{bertsimas:2006dynamicpricing} and \citet{pk:2001dynamicpricing}, to name a few. When more than one agent is involved, the dynamic pricing environment becomes more complex. \citet{braverman:2018selling-noregret} investigates the problem setting where a seller is selling to a regret-minimizing buyer. Under the assumption that the buyer is running a contextual bandit algorithm, a reward-maximizing algorithm for the seller upper-bounded by a margin above the monopoly revenue is provided. \citet{goyal:2021mnl} applied the use of linear contextual bandits to model multinomial logit demand in the multi-assortment dynamic pricing problem. \citet{bu:2022context-dyn-pric} investigates online learning for dynamic pricing with a partially linear demand function. Multi-agent dynamic pricing has also been previously studied in \citet{dasgupta:2000-dyp} and \citet{kutschinski:2003-ma-dyp}, to name a few. Nevertheless, these aforementioned works do not consider the context of agents in a supply chain setting, where the inventory policy has an impact on the expected profits.

Dynamic pricing can also be considered in scenarios with limited supply \cite{babaioff:2015dyp-lim-supply}, thereby limiting the retailer from profiting from the theoretical revenue maximizing price for the product. In situations where the retailer must order in advance the product to be sold, and bear the risk of limited supply, or supply surplus, this is known as a \textit{Newsvendor} firm \cite{arrow:1951newsboy} \cite{mills:1959uncertainty} \cite{hadley:1963-analysis-inven}. Specifically, learning and optimization within the Newsvendor framework have also been studied in \citet{ban:2019big-data}. When the Newsvendor must also dynamically price the item, this is referred to as the \textit{price-setting Newsvendor} \cite{jammernegg:2013pricesettingnewsvendor} \cite{petruzzi:1999newsv} \cite{mills:1959uncertainty}, and methods to compute such optimal solutions under perfect information were first provided in \citet{mills:1959uncertainty}. We further extend the well-known problem of online learning in a single agent dynamic pricing environment to a two agent competitive learning environment. One example of such competition is the rivalry between independent self-interested suppliers and retailers within a supply chain. Research by \citet{perakis:2007price} and \citet{ma:2022-poascm} has demonstrated that this type of competition can lead to a degradation in the general social welfare. 

In our case we model the competition in a dyadic exchange network as a repeated Stackelberg game, which finds numerous applications in market economics \cite{defraja:1990-game-theory-stack} \cite{anderson:1992-stackelberg-v-cournot} and supply chain analysis \cite{noh:2019-stack-supp-chain}. In previous work within the domain of supply chain management, the elements of parameter learning and regret minimization remain under-investigated. Specifically, in scenarios where suppliers and retailers engage in myopic supply chain competition aimed at profit maximization through the exchange of demand function information, as in \citet{zhang:2013-infoshare}, equilibrium is often established straightforwardly under a perfect information setting. Most saliently, the ramifications of worst-case scenarios and suboptimal deviations in the face of uncertainty remain wholly unquantified. Therefore, addressing this challenge with mathematical rigor constitutes one of the focal points of our work.

In the framing of online learning for a supply chain Stackelberg game, \citet{cesa:2022-supp-chain-games} considered Newsvendor competition between supplier and retailer, under perfect information and limited information settings. Utilizing bandit learning techniques, such as \textit{explore-then-commit} (ETC) and \textit{follow-the-leader} (FTL), they are able to produce bounds on regret under mild assumptions \cite{garivier:2016-etc}. 

Our supply chain game consists of a leader (supplier) and a follower (retailer). The leader dynamically prices the item to sell to the follower, and the follower must thereafter dynamically price the item for the market. 

The objective of this repeated game is to determine the optimal policy to maximize the expected profit in the face of stochastic demand for the leader, and determine the best response of the follower (under the assumptions in Section \ref{sec:assumptions}), leading to an approximate Stackelberg equilibrium. The rationale behind this is that, often times, decisions between firms competing with one another do not occur simultaneously, but rather turn-based. We define this \textit{Newsvendor pricing game} more concretely in Section \ref{sec:newsvendor-games}.

In most cases, given the known parameters of the model, the best response of the follower can be computed for any of the leader's first action. Assuming rationality of both agents, it follows from the best response of the follower that the optimal policy of the leader can be computed. However, the sample-efficient learning of the parameters of the joint reward function presents a genuine problem for Stackelberg gamess, as the leader is learning both the joint reward function of the system and computing the best response of the follower. To increase the sample efficiency of parameter learning in Stackelberg games, bandit algorithms have been proposed to address this problem \cite{bai:2021-stackel} \cite{fiez:2019-stack-learn} \cite{fiez:2020-stack-implicit-learn} \cite{balcan:2015-stack-learn-sec}. Our paper proposes online learning algorithms for Stackelberg games in a dyadic exchange network utilizing linear contextual bandits \cite{abbasi:2011ofu} \cite{chu:2011-linucb}, borrowing bounds on misspecification from these bandit methods to achieve theoretical guarantees on the behaviour of the algorithm in terms of regret and optimality, later discussed in Section \ref{sec:learning-algorithms}.

\subsection{Key Contributions}
    
    Our study integrates established economic theory with online learning theory within a practical game theoretic framework. Specifically, we employ a linear contextual bandit algorithm within a Stackelberg game to address a well-defined economic scenario. The primary technical challenge lies in optimizing for both contextual and Stackelberg regret, particularly when facing stochastic environmental parameters, as well as uncertainty in agent strategies. Incorporating economic theory into online learning theory is inherently challenging due to the unique characteristics of reward functions and best response functions, as expounded in Section \ref{sec:econ-newsv-theory}. In some cases, closed-form solutions are unattainable, requiring innovative approaches for optimization and bounding functions, as expounded in Section \ref{sec:leader_stackelberg_regret}, in order to derive theoretical worst-case bounds on regret.

    To summarize, we provide technical contributions in the following areas:
    
    \begin{enumerate}
        \item We prove the existence of a unique Stackelberg equilibrium under perfect information for a general sum game within a \textit{Newsvendor pricing game} (described in Section \ref{sec:newsvendor-games}).        
        \item An online learning algorithm for the Newsvendor pricing game, leveraging stochastic linear contextual bandits, is outlined, which is integrated with established economic theory.
        \item The convergence properties of an online learning algorithm to approximate Stackelberg equilibrium, as well theoretical guarantees for bounds on finite-time regret, are provided.
        \item We demonstrate our theoretical results via economic simulations, to show that our learning algorithm outperforms baseline algorithms in terms of finite-time cumulative regret.
    \end{enumerate}

\section{Model Formulation} \label{sec:model-formulation}

\textbf{Definitions:} We define two players in a repeated Stackelberg game. Player A is the leader, she acts first with action $\mathbf{a}$. Player B is the follower, he acts second with action $\mathbf{b}$. The follower acts in response to the leaders action, and both players earn a joint payoff as a function of the reward function. In the abstract sense, $\mathbf{a}$ and $\mathbf{b}$ denote the actions of the leader and follower in action spaces $\mathcal{A}$ and $\mathcal{B}$ respectively. To be precise, given an action space with dimension $\mathbbm{R}^d$, $\mathbf{a}$ and $\mathbf{b}$ are vectors of different dimensions (we later illustrate that $\mathcal{A} \subseteq \mathbbm{R}^1$ and $\mathcal{B} \subseteq \mathbbm{R}^2$ in Section \ref{sec:rules_of_npg}). This is due to the fact that the leader's and follower's actions do not necessarily follow the same form.


\subsection{Repeated Stackelberg Games}


In a repeated Stackelberg game, the leader takes actions $\mathbf{a} \in \mathcal{A}$, and the follower takes actions $\mathbf{b} \in \mathcal{B}$ in reaction to the leader. The joint action is represented as $(\mathbf{a}^t, \mathbf{b}^t)$ for each time step $t$ from $1$ to $T$. The joint utility (reward) functions for the leader and follower are denoted as $\mathcal{G}_A(\mathbf{a}^t, \mathbf{b}^t)$ and $\mathcal{G}_B(\mathbf{a}^t, \mathbf{b}^t)$, respectively. From an economic standpoint, the application of pure strategies aligns with the prevalent strategic behaviours observed in many firms, emphasizing practicality and realism \cite{porter:1980techniques_competition} \cite{tirole:1988theory_industrial_org}.

\textbf{Strategy Definition:} Let $\pi_A$ and  $\pi_B$ denote the pure strategies of both agents represented as standard maps,
  
\begin{align}
    \pi_A(t) : \mathcal{F}_t \to \mathcal{A}, \quad \pi_B(\mathbf{a}) : \mathcal{A} \to \mathcal{B}, \quad t \in \mathbb{Z} \cap \{ 1 \leq t \leq T \} \label{eq:strategy_def_agents}
\end{align}

The strategy of the leader at time  $t$ is denoted by $\pi_A(t)$, which maps from the filtration $\mathcal{F}_t$ to the set of possible actions $\mathcal{A}$. Where the filtration $\mathcal{F}_t$ consists of all possible $\sigma$-algebras of observations up until time $t$, i.e., $\mathcal{F}_t \equiv \{(\mathbf{a}^1, \mathbf{b}^1), \ldots, (\mathbf{a}^t, \mathbf{b}^t)\}$. To clarify, $t$ belongs to the set of integers $\mathbb{Z}$ such that $1 \leq t \leq T$, where $T$ is a fixed finite-time endpoint. The action taken by the leader at time $t$, denoted as $\mathbf{a}^t \in \mathcal{A}$, is determined by the leader's strategy $\pi_A(t)$, which depends on the observed actions up to time $t$. Therefore, $\pi_A$ is a morphism from the set of filtrations to the set of actions in discrete time, $\pi_A : \{\mathcal{F}_1, \ldots, \mathcal{F}_T\} \mapsto \{\mathbf{a}^1, \ldots, \mathbf{a}^T\}$. On the other hand, the strategy of the follower is formulated as a response to the leader's action; therefore, it is a morphism, $\pi_B : \{ \mathbf{a}^1, \dots, \mathbf{a}^T \} \mapsto \{ \pi_B(\mathbf{a}^1), \dots, \pi_B(\mathbf{a}^T) \}$.


 \textbf{Best Response Function:} To be specific, $\mathcal{G}_B(\mathbf{a}^t, \mathbf{b}^t)$ is expressed as a parametric function with exogenous sub-Gaussian noise (the exact parametric form for $\mathcal{G}_B(\mathbf{a}^t, \mathbf{b}^t)$ is provided in Eq. \eqref{eq:profit-cases}). Conversely, for the leader, $\mathcal{G}_A(\mathbf{a}^t, \mathbf{b}^t)$ is a deterministic function solely driven by the direct action of the leader $\mathbf{a}^t$ followed by the reaction of the follower $\mathbf{b}^t$. At each discrete time interval, the best response of the follower, $\mathfrak{B}(\mathbf{a}^t)$, is defined as follows:


\begin{equation}
    \mathfrak{B}(\mathbf{a}^t) = \underset{\mathbf{b} \in \mathcal{B}}{\mathrm{argmax}} \ \mathbbm{E} [\mathcal{G}_B(\mathbf{a}^t, \mathbf{b})] \label{eq:br-def}
\end{equation}

Eq. \eqref{eq:br-def} makes the assumption that there exists a unique best response to $\mathbf{a}^t$. We demonstrate the uniqueness of $\mathfrak{B}(\mathbf{a}^t)$ later in Theorem \ref{thm:optimal-order-b}. $\mathfrak{B}(\mathbf{a})$ is the best response function of $\mathbf{b}$ to $\mathbf{a}$, as defined in Equation \eqref{eq:br-def} under perfect information. Given $\mathfrak{B}(\mathbf{a})$ is unique, the objective for the leader is to maximize her reward based on the expectation of the follower's best response. In a Stackelberg game, the leader takes actions, from time $\{ \mathbf{a}^1, \dots, \mathbf{a}^T \}$, and receives rewards $\{ \mathcal{G}_A(\mathbf{a}^1, \pi_B(\mathbf{a}^1)), \dots, \mathcal{G}_A(\mathbf{a}^T, \pi_B(\mathbf{a}^T)) \}$, as the follower is always reacting to the leader with his perfect or approximate best response, he receives rewards $\{ \mathcal{G}_B(\mathbf{a}^1, \pi_B(\mathbf{a}^1)), \dots, \mathcal{G}_B(\mathbf{a}^T, \pi_B(\mathbf{a}^T)) \}$. We shall use the notation $(\pi_A, \pi_B)$ to represent the joint strategy of both players.


\textbf{Contextual Regret (Follower):} In the bandit setting with no state transitions, the leader's strategy $\pi_A$ represents a sequence consisting of actions, or contexts, $\mathbf{a}^t$, $\forall t \in \{1, \dots, T \}$ for the follower. The contextual regret is defined as, 



\begin{align}
    R_B^T(\pi_A, \pi_B) &= \mathbbm{E}_\theta \Big[ \sum^T_{t = 1} \mathcal{G}_B(\mathbf{a}^t, \mathfrak{B}(\mathbf{a}^t)) - \sum^T_{t = 1} \mathcal{G}_B(\mathbf{a}^t, \mathbf{b}^t) \Big]  \label{eq:follower-regret}
\end{align}

Contextual regret, defined as $R_B^T(\cdot)$ in Eq. \eqref{eq:follower-regret}, is the difference between cumulative follower rewards from the best-responding follower and any other strategy adhered to by the follower in response to a series of actions form the leader, $\{ \mathbf{a}^1, \dots, \mathbf{a}^T \}$. Given context $\mathbf{a}^t$ and environmental parameters $\theta$, the follower attempts to maximize his rewards via some no-regret learner (we choose an optimistic stochastic linear contextual bandit algorithm as outlined in Section \ref{sec:learning-algorithms}). $\mathbbm{E}_{\theta}[\mathcal{G}_B(\cdot)]$ indicates that expectation over uncertainty is with respect to the stochastic demand environment, governed by parameter $\theta$.



\textbf{Stackelberg Regret (Leader):} The Stackelberg regret, defined as $R_A^T(\cdot)$ for the leader, is the difference in cumulative rewards when players optimize their strategies in hindsight, versus via any joint strategy $(\pi_A, \pi_B)$.  \textit{Stackelberg regret}, also known as the \textit{external regret} \cite{haghtalab:2022-stack-non-myo} \cite{zhao:2023-online} is defined from the leader's perspective in Eq. \eqref{eq:leader-regret-def} as the cumulative difference in leader's rewards between the maximization of $\mathcal{G}_A(\cdot)$ under a best responding follower, versus any other joint strategy adhered to by both agents. This definition assumes that the rational leader must commit to a strategy $\pi_A$ at the beginning of the game, which she believes is optimal on expectation. The leader maximizes their reward taking into account their anticipated uncertainty of the follower's response $\mathbf{b}^t$.  When computing the expectation of $\mathcal{G}_A(\cdot)$, the randomness of $\mathcal{G}_A(\cdot)$ is fully strategic, that is, it arises from the follower's action according to the strategy $\mathbf{b}^t \sim \pi_B$. Our algorithm seeks to minimize the Stackelberg regret, constituting a no-regret learner from the leader's perspective.


\begin{align}
    R_A^T(\pi_A, \pi_B) &= \mathbbm{E}_{\pi_B} \Big[\sum_{t = 1}^{T} \underset{\mathbf{a} \in \mathcal{A} }{\mathrm{max}} \ \mathcal{G}_A(\mathbf{a}, \mathfrak{B}(\mathbf{a}^t)) - \sum_{t = 1}^{T} \, \mathcal{G}_A(\mathbf{a}^t, \mathbf{b}^t)\Big] \label{eq:leader-regret-def}
\end{align}

For convenience of notation, moving forward, we will establish the equivalences for $\mathbbm{E}_{\pi_B}[\mathcal{G}_A(\cdot)] \equiv \mathbbm{E}[\mathcal{G}_A(\cdot)]$ and $\mathbbm{E}_\theta[\mathcal{G}_B(\cdot)] \equiv \mathbbm{E}[\mathcal{G}_B(\cdot)]$.

\subsection{Literature Review on Learning in Stackelberg Games}

\begin{table*}\centering 
    \small
    \begin{tabular}{p{1.0in}p{1.0in}p{0.35in}p{1.1in}p{0.6in}p{1.0in}} \toprule
        \textbf{Source} & \textbf{Prob. Setting} & \textbf{Info.} & \textbf{Methodology} & \textbf{Theorem} & \textbf{Regret} \\ \midrule
       \citet{zhao:2023-online}  & Noisy omniscient follower. with noise& Partial Info. & UCB with expert information. & Thm. 3.3 & $\mathcal{O}(\hat{\epsilon} T + \sqrt{ N(\hat{\epsilon}) T})$ \\ \midrule
       \citet{haghtalab:2022-stack-non-myo}  & Stackelberg Security Game.  & Partial Info. & Batched binary search, with $\gamma$ discounting, and delayed feedback. & Thm. 3.13 & $\mathcal{O}(\log(T) + T_\gamma \log^2{(T_\gamma)})$  \\ \midrule
       \citet{haghtalab:2022-stack-non-myo}  & Dynamic pricing.  & Partial Info. & Successive elimination with upper-confidence bound. & Thm. 4.7 & $\mathcal{O}(\sqrt{T \log(T)} + T_\gamma \log^2(T_\gamma T) )$ or $\mathcal{O}(\log(T) + T)$  \\ \midrule
       \citet{flaxman:2004-online} & Convex joint reward function* & Partial Info.  & Gradient approximation. & Thm. 2 & $\mathcal{O}(P(T) + \log^{3/4}{(T)})$\\ \midrule
       \citet{chen:2019-grinding} & Spam Classification & Full Info.  & Adaptive polytope partitioning. & Thm. 4.1 & $\mathcal{O}( \sqrt{T \log{(T)}})$\\ \midrule
       \citet{balcan:2015-stack-learn-sec} & Stackelberg Security Game & Full Info.  & Follow the leader. & Thm. 5.1 & $\mathcal{O}(\sqrt{T} )$\\ \bottomrule
    \end{tabular} \caption{Worst case regret comparison among learning algorithms for Stackelberg games.} \label{tab:regret_summary}
\end{table*}

We classify game theoretic learning as consisting of two key aspects. The first is the learning of the opposing player's strategy. In particular, the strategy of the follower can be optimal, in the case of an \textit{omniscient follower}, or $\epsilon$-optimal \cite{chen:2019-grinding} \cite{balcan:2015-stack-learn-sec} \cite{gan:2023-robust-optimal}, in the case of an imperfect best response stemming from hindered information sharing or suboptimal follower policies. In this setting, the parameters of the reward function $\theta^*$ are known, and the players are learning an approximate or perfect equilibrium strategy. Notably, \citet{balcan:2015-stack-learn-sec} demonstrated that there always exists an algorithm that will produce a worst case leader regret bound of $\mathcal{O}(\log{(T)})$. In another example, \citet{chen:2019-grinding} presents a method where polytopes constructed within the joint action space of the players can be adaptively partitioned to produce a learning algorithm achieving sublinear regret under perfect information, demonstrated in a spam classification setting.

The second aspect of online learning in SG's pertains to the learning of the parameters of the joint reward function $\theta^*$ which corresponds to $\mathcal{G}_B(\cdot)$ and influences $\mathcal{G}_A(\cdot)$. In this environment, $\theta^*$ is unknown or only partially known, and must be learned over time in a bandit setting. Several approaches address this problem setting. For example, a worst case regret bound of $\mathcal{O}(\hat{\epsilon} T + \sqrt{ N(\hat{\epsilon}) T})$ was proven in \citet{zhao:2023-online}, where the authors use a $\hat{\epsilon}$-covering method to achieve the bound on regret. Where $N(\hat{\epsilon})$ is the $\hat{\epsilon}$-covering over the joint action space. \citet{haghtalab:2022-stack-non-myo} produce a worst-case regret bound of $\mathcal{O}(\log(T) + T_\gamma \log^2{(T_\gamma)})$ by leveraging batched binary search in a demand learning problem setting. Furthermore, they investigate the problem of a non-myopic follower, where the anticipation of future rewards ameliorated by a discount factor $\gamma$, induces the follower to suboptimally best respond. $T_\gamma = 1 / (1 - \gamma)$ represents a parameter proportional to the delay in the revelation of the follower's actions by a third party. Under this mechanism, the regret is proportional to both $T$ as well as $\gamma$, where for very small values of $\gamma$ (close to non-discounting), the $T_\gamma$ could potentially dominate the regret. \citet{flaxman:2004-online} provide a bound for a convex joint reward function using approximate gradient descent, where $P(T)$ represents a polynomial function with respect to $T$. Note that in the partial information setting, the bound on worst-case leader regret is typically superlinear.



\subsection{Approximate Stackelberg Equilibrium} \label{sec:stack-eq-defn}

The \textit{Stackelberg Equilibrium} of this game is defined as any joint strategy which,

\begin{align}
    \mathbf{a}^* &= \underset{\mathbf{a} \in \mathcal{A} }{\mathrm{argmax}} \ \mathop{\mathbb{E}}[\mathcal{G}_A(\mathbf{a}, \mathfrak{B}(\mathbf{a}))]
\end{align}

Where $\mathfrak{B}(\mathbf{a})$ is the best response of the follower, as defined in Eq. \eqref{eq:br-def}. Let $\Delta_B$ denote a probability simplex over possible strategies $\pi_B$, thus the expectation on $\mathbbm{E}_{\pi_B}[\cdot]$ is with respect to $\pi_B$. From the follower's perspective, suppose he adopts an approximate best response strategy $\hat{\pi}_B \equiv \hat{\pi}_B(\mathbf{a})$. This implies that there exists some $\epsilon_B \geq 0$, such that,

\begin{align}
     \mathbbm{E}_\theta [ \mathcal{G}_B(\mathbf{a}, \hat{\pi}_B) ] \geq \mathbbm{E}_\theta [\mathcal{G}_B(\mathbf{a}, \mathfrak{B}(\mathbf{a}))] - \epsilon_B, \quad \forall \mathbf{a} \in \mathcal{A}, \quad \forall \pi_B(\cdot) \in \Delta_B
\end{align}

This formalization allows for an $\epsilon$ deviation, potentially making the problem computationally tractable, particularly when optimization over certain functions lacks simple closed-form solutions or is computationally intensive (as we explore in detail in Section \ref{sec:econ-newsv-theory}). It's worth noting that the expectation for the leader reflects uncertainty about the follower's strategy, while the expectation for the follower relates to the randomness of the environment. Acknowledging the presence of the error of best response function resulting from the possibility of the follower to act imperfectly, or strategically with purpose influence the leader's strategy, we propose an \textit{$\epsilon$-approximate equilibrium} with full support, \cite{roughgarden:2010-algorithmic} \cite{daskalakis:2009complexity}.

\textbf{Upper-Bounding Approximation Function:} Let $\mathcal{U}_B(\mathbf{a}, \mathbf{b})$ represent a Lipschitz function, $\mathcal{U}_B : \mathbf{a} \times \mathbf{b} \to \mathbb{R}^+$ such that it is consistently greater than or equal to $\mathbbm{E}[\mathcal{G}_B(\mathbf{a}, \mathbf{b})]$. The selection of this upper-bounding function is appropriate within our problem framework for two key reasons. Firstly, it leverages on established theory on economic pricing, allowing us to derive a rational best response for the follower (see Section \ref{sec:econ-newsv-theory}). Secondly, its adoption aligns with our implementation of an optimistic bandit algorithm, ensuring an inherent tendency to overestimate reward potential in the face of uncertainty.


\begin{align} 
    \mathcal{U}_B(\mathbf{a}, \mathbf{b}) &\geq \mathbbm{E} [\mathcal{G}_B(\mathbf{a}, \mathbf{b})], \ \forall \mathbf{a} \in \mathcal{A}, \ \mathbf{b} \in \mathcal{B} \label{eq:u_def}
\end{align}

Therefore, we define the approximate follower best response as,

\begin{align} \label{eq:br_approx_def}
    \widetilde{\mathfrak{B}}(\mathbf{a}) &= \underset{b \in \mathcal{B}}{\mathrm{argmax}} \ \mathcal{U}_B(\mathbf{a}, \mathbf{b})
\end{align}

In accordance with Eq. \eqref{eq:u_def}, we note that the follower reward utilizing the approximate best response is less than that of the optimal best response as a consequence, where $\mathop{\mathbb{E}} [\mathcal{G}_B(a, \widetilde{\mathfrak{B}}(a))] \leq \mathop{\mathbb{E}} [\mathcal{G}_B(a, \mathfrak{B}(a))]$. Furthermore, in our problem definition, the follower approximates his best response with $\widetilde{\mathfrak{B}}(\mathbf{a})$ such that, across its domain of $\mathbb{E}[\mathcal{G}_B(\cdot)]$, the image of the approximating function, $\mathcal{U}(\cdot)$, is greater or equal than the image of $\mathop{\mathbb{E}} [\mathcal{G}_B(\cdot)$]. The motivation behind this is that, $\mathcal{U}(\cdot)$ can represent a simple parametric function with a known analytical solution, whereas $\mathop{\mathbb{E}}[\mathcal{G}_B(\cdot)]$ could represent an more elaborate function, with no known exact analytical solution. Therefore it is possible that, $\mathcal{U}(\cdot)$ and $\mathbb{E}[\mathcal{G}_B(\cdot)]$ may not necessarily belong to the same parametric family, nor adhere to the same functional form. 

\textbf{Approximate SE:} As a consequence, we can therefore infer that there exists some additive $\epsilon_B \geq 0$ which upper-bounds this approximation error in the image of $\mathbb{E}[\mathcal{G}_B(\cdot)]$. We define an approximate $\epsilon$-\textit{Stackelberg equilibrium} with $\epsilon_B$ approximation error the follower's perspective as follows,

\begin{equation} \label{eq:approx_se_defn_b}
      \mathop{\mathbb{E}} [\mathcal{G}_B(\mathbf{a}, \widetilde{\mathfrak{B}}(\mathbf{a}))] \leq \mathop{\mathbb{E}} [\mathcal{G}_B(\mathbf{a}, \mathfrak{B}(\mathbf{a}))] \leq \mathop{\mathbb{E}} [\mathcal{G}_B(\mathbf{a}, \widetilde{\mathfrak{B}}(\mathbf{a}))] + \epsilon_B , \quad \forall \mathbf{a} \in \mathcal{A}
\end{equation}

Effectively, $\epsilon_B$ defines an upper bound of the range on where $\mathop{\mathbb{E}} [\mathcal{G}_B(\mathbf{a}, \mathfrak{B}(\mathbf{a}))]$ could lie, given the approximate best response under perfect information, defined as $\widetilde{\mathfrak{B}}(\mathbf{a})$. $\epsilon_B$ represents the worst case additive approximation error of $\mathbbm{E}[\mathcal{G}_B(\mathbf{a}, \mathfrak{B}(\mathbf{a}))]$, where $\epsilon_B > 0$. 

\begin{align} 
    \epsilon_B &= \underset{\mathbf{a} \in \mathcal{A}}{\mathrm{sup}} \ \Big( \mathbbm{E}[\mathcal{G}_B(\mathbf{a}, \mathfrak{B}(\mathbf{a}))] - \mathbbm{E}[\mathcal{G}_B(\mathbf{a}, \widetilde{\mathfrak{B}}(\mathbf{a}))] \Big), \quad \epsilon_B > 0 \label{eq:br_approx_def_eps_b}
\end{align}


The suboptimality of $\mathop{\mathbb{E}} [\mathcal{G}_B(\mathbf{a}, \widetilde{\mathfrak{B}}(\mathbf{a}))] $ is attributed to the error in the approximation of the theoretically optimal response $\mathfrak{B}(\mathbf{a})$ with $\widetilde{\mathfrak{B}}(\mathbf{a})$, in response to the leader's action $\mathbf{a}$ (the economic definitions of $\mathfrak{B}(\cdot)$, $\mathcal{U}_B(\cdot)$ and $\mathcal{G}_B(\cdot)$ are later illustrated in Section \ref{sec:newsvendor-games}). In our learning scenario, the follower has access to $\mathcal{U}(\cdot)$, which shares the same parameters but differs from the functional form of $\mathbb{E}[\mathcal{G}_B(\cdot)]$. This allows us devise a learning algorithm which converges to $\epsilon_B$. Thus, any algorithm which learns an $\epsilon$-approximate Stackelberg equilibrium, is formally defined as an algorithm that converges to an equilibrium solution with $\epsilon_B$ approximation error, as defined in Eq. \eqref{eq:br_approx_def_eps_b}.

\subsection{Multi-Agent Economic Setting - the Newsvendor Pricing Game (NPG)} \label{sec:newsvendor-games}


We model the two learning agents in a \textit{Newsvendor pricing game}, involving a supplier $A$ and a retailer $B$. The leader, a supplier, is learning to dynamically price the product for the follower, a retailer, aiming to maximize her reward. To achieve this, the follower adheres to classical Newsvendor theory, which involves finding the optimal order quantity given a known demand distribution before the realization of the demand. We provide further economic details in Section \ref{sec:econ-newsv-theory}.

Differing from the closely related work of \citet{cesa:2022-supp-chain-games}, we do make some relaxations to the problem in their setting. Firstly, we assume no production cost for the item, which provides a non-consequential change to the economic theory. Secondly, we assume full observability of the realized demand of both the supplier and retailer. This is commonplace in settings where a disclosure of the realized sales must be disclosed to an upstream supplier, which could occur as a stipulation as a part of the supplier contract, and/or when the sales ecosystem is part of an integrated eCommerce platform.

Nevertheless, our problem setting further sophisticates the solution from \citet{cesa:2022-supp-chain-games} in multiple ways. Firstly, in addition to determining the optimal order amount (known as the \textit{Newsvendor} problem), the follower must also dynamically determine the optimal retail price for the product in the face of demand uncertainty (known as the \textit{price-setting Newsvendor}). Thereby, we introduce the problem of demand learning coupled with inventory risk. Secondly, in our problem setting, although the linear relation of price to demand $p \sim d_\theta(p)$ is assumed, the distribution of neither demand nor market price is explicitly given. Third, we apply the optimism under uncertainty linear bandit (OFUL) \cite{abbasi:2011ofu} to the Stackelberg game, whereas in \citet{cesa:2022-supp-chain-games} the \textit{Piyavsky-Schubert} \cite{bouttier:2020PV} algorithm is applied. Lastly, we provide Stackelberg regret guarantees on the cumulative regret, whereas \citet{cesa:2022-supp-chain-games} provides guarantees on the simple regret.

\subsubsection{Rules of the Newsvendor Pricing Game} \label{sec:rules_of_npg}

Moving forward, and away from abstraction, we explicitly denote $a \equiv \mathbf{a} \in \mathbbm{R}^1$, and $\mathbf{b} \equiv [b, p]^\intercal \in \mathbbm{R}^2$. Where $a$ denotes wholesale price from the supplier firm, $p$ and $b$ denote the retail price and order amount of the retail firm.

\begin{enumerate}
    \setlength\itemsep{-0.1em}
    \item The supplier selects wholesale price $a$, and provides it to the retailer. \label{item:supplier-contract-price}
    \item Given wholesale cost $a$, the retailer reacts with his best response $[b, p]^\intercal$, consisting of retail price $p$, and order amount $b$.
    \item As the retailer determines the optimal order amount $b$, he pays $\mathcal{G}_A(a, b) = ab$ to the supplier.
    \item At time $t$, nature draws demand $d^t \sim d_\theta(p)$, and it is revealed to the retailer.
    \item The retailer makes a profit of $\mathcal{G}_B(a, b) = p \ \min \{ d^t, b \} - ab$. \label{item:retail-profit}
    \item Steps \ref{item:supplier-contract-price} to \ref{item:retail-profit} are repeated for $t \in 1 ... T$ iterations.
\end{enumerate}

\subsubsection{Economic Game Assumptions} \label{sec:assumptions}

From the \textit{price setting Newsvendor}, the majority of the assumptions from \citet{mills:1959uncertainty} are inherited as standard practice. In our current problem setting, we explicitly state the assumptions of the \textit{Newsvendor pricing game} as follows: 

\begin{enumerate} 
    \setlength\itemsep{-0.1em}
    \item Expected demand, $\Gamma_\theta(p)$, is an additive linear model, as illustrated later in Eq. \eqref{eq:demand-theta-func}.
    \item Inventory is perishable, and we assume no goodwill costs and/or salvage value.
    \item Market demand, $d^t$, and action history, $(a^t, b^t$, $p^t)$, is fully observable to both retailer and supplier.
    \item The follower is rational and greedy, he maximizes for his immediate reward under the available information and best response computation.
    \item No deviation, stochastic or otherwise, exists in $b^t$, denoting the amount of goods ordered and goods received, from the supplier by the retailer.
    \item Decisions are made during fixed periodic discrete time intervals.
\end{enumerate}

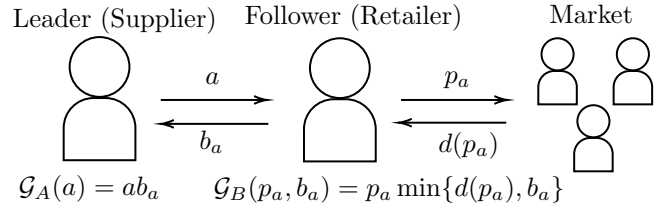
\begin{figure}[ht!]
    \centering
    \tikzset{every picture/.style={line width=0.75pt}} 
    
        \begin{tikzpicture}[x=0.75pt,y=0.75pt,yscale=-0.6,xscale=0.6]
        
        \draw   (60,214) .. controls (60,200.19) and (71.19,189) .. (85,189) -- (95,189) .. controls (108.81,189) and (120,200.19) .. (120,214) -- (120,241) .. controls (120,241) and (120,241) .. (120,241) -- (60,241) .. controls (60,241) and (60,241) .. (60,241) -- cycle ;
        \draw   (65,163) .. controls (65,149.19) and (76.19,138) .. (90,138) .. controls (103.81,138) and (115,149.19) .. (115,163) .. controls (115,176.81) and (103.81,188) .. (90,188) .. controls (76.19,188) and (65,176.81) .. (65,163) -- cycle ;
        \draw    (141.2,190.6) -- (228.2,190.6) ;
        \draw [shift={(230.2,190.6)}, rotate = 180] [color={rgb, 255:red, 0; green, 0; blue, 0 }  ][line width=0.75]    (10.93,-3.29) .. controls (6.95,-1.4) and (3.31,-0.3) .. (0,0) .. controls (3.31,0.3) and (6.95,1.4) .. (10.93,3.29)   ;
        \draw   (261,215) .. controls (261,201.19) and (272.19,190) .. (286,190) -- (296,190) .. controls (309.81,190) and (321,201.19) .. (321,215) -- (321,242) .. controls (321,242) and (321,242) .. (321,242) -- (261,242) .. controls (261,242) and (261,242) .. (261,242) -- cycle ;
        \draw   (266,164) .. controls (266,150.19) and (277.19,139) .. (291,139) .. controls (304.81,139) and (316,150.19) .. (316,164) .. controls (316,177.81) and (304.81,189) .. (291,189) .. controls (277.19,189) and (266,177.81) .. (266,164) -- cycle ;
        \draw    (231.2,210.6) -- (212.4,210.6) -- (143.2,210.6) ;
        \draw [shift={(141.2,210.6)}, rotate = 360] [color={rgb, 255:red, 0; green, 0; blue, 0 }  ][line width=0.75]    (10.93,-3.29) .. controls (6.95,-1.4) and (3.31,-0.3) .. (0,0) .. controls (3.31,0.3) and (6.95,1.4) .. (10.93,3.29)   ;
        \draw    (342.2,190.6) -- (429.2,190.6) ;
        \draw [shift={(431.2,190.6)}, rotate = 180] [color={rgb, 255:red, 0; green, 0; blue, 0 }  ][line width=0.75]    (10.93,-3.29) .. controls (6.95,-1.4) and (3.31,-0.3) .. (0,0) .. controls (3.31,0.3) and (6.95,1.4) .. (10.93,3.29)   ;
        \draw   (457,179.79) .. controls (457,172.38) and (463.01,166.37) .. (470.42,166.37) -- (475.78,166.37) .. controls (483.19,166.37) and (489.2,172.38) .. (489.2,179.79) -- (489.2,194.28) .. controls (489.2,194.28) and (489.2,194.28) .. (489.2,194.28) -- (457,194.28) .. controls (457,194.28) and (457,194.28) .. (457,194.28) -- cycle ;
        \draw   (459.68,152.42) .. controls (459.68,145.01) and (465.69,139) .. (473.1,139) .. controls (480.51,139) and (486.52,145.01) .. (486.52,152.42) .. controls (486.52,159.83) and (480.51,165.83) .. (473.1,165.83) .. controls (465.69,165.83) and (459.68,159.83) .. (459.68,152.42) -- cycle ;
        \draw   (487,235.79) .. controls (487,228.38) and (493.01,222.37) .. (500.42,222.37) -- (505.78,222.37) .. controls (513.19,222.37) and (519.2,228.38) .. (519.2,235.79) -- (519.2,250.28) .. controls (519.2,250.28) and (519.2,250.28) .. (519.2,250.28) -- (487,250.28) .. controls (487,250.28) and (487,250.28) .. (487,250.28) -- cycle ;
        \draw   (489.68,208.42) .. controls (489.68,201.01) and (495.69,195) .. (503.1,195) .. controls (510.51,195) and (516.52,201.01) .. (516.52,208.42) .. controls (516.52,215.83) and (510.51,221.83) .. (503.1,221.83) .. controls (495.69,221.83) and (489.68,215.83) .. (489.68,208.42) -- cycle ;
        \draw   (520,179.79) .. controls (520,172.38) and (526.01,166.37) .. (533.42,166.37) -- (538.78,166.37) .. controls (546.19,166.37) and (552.2,172.38) .. (552.2,179.79) -- (552.2,194.28) .. controls (552.2,194.28) and (552.2,194.28) .. (552.2,194.28) -- (520,194.28) .. controls (520,194.28) and (520,194.28) .. (520,194.28) -- cycle ;
        \draw   (522.68,152.42) .. controls (522.68,145.01) and (528.69,139) .. (536.1,139) .. controls (543.51,139) and (549.52,145.01) .. (549.52,152.42) .. controls (549.52,159.83) and (543.51,165.83) .. (536.1,165.83) .. controls (528.69,165.83) and (522.68,159.83) .. (522.68,152.42) -- cycle ;
        \draw    (430.2,209.6) -- (342.2,209.6) ;
        \draw [shift={(340.2,209.6)}, rotate = 360] [color={rgb, 255:red, 0; green, 0; blue, 0 }  ][line width=0.75]    (10.93,-3.29) .. controls (6.95,-1.4) and (3.31,-0.3) .. (0,0) .. controls (3.31,0.3) and (6.95,1.4) .. (10.93,3.29)   ;
        
        \draw (176,165) node [anchor=north west][inner sep=0.75pt]   [align=left] {$a$};
        \draw (171,214) node [anchor=north west][inner sep=0.75pt]   [align=left] {$b_a$};
        \draw (376,164) node [anchor=north west][inner sep=0.75pt]   [align=left] {$p_a$};
        \draw (371,215) node [anchor=north west][inner sep=0.75pt]   [align=left] {$d(p_a)$};
        \draw (183,252) node [anchor=north west][inner sep=0.75pt]   [align=left] {$\mathcal{G}_B(p_a, b_a) = p_a \min \{ d(p_a), b_a \}$};
        \draw (20,251) node [anchor=north west][inner sep=0.75pt]   [align=left] {$\mathcal{G}_A(a) = a b_a$};
        \draw (15,110) node [anchor=north west][inner sep=0.75pt]   [align=left] {Leader (Supplier)};
        \draw (209,108) node [anchor=north west][inner sep=0.75pt]   [align=left] {Follower (Retailer)};
        \draw (464,108) node [anchor=north west][inner sep=0.75pt]   [align=left] {Market};
    \end{tikzpicture}
    \caption{\textbf{The Newsvendor Pricing Game.} In this Stackelberg game, there a logistics network between a supplier (leader) and retailer (follower), where utility functions are not necessarily supermodular, the supplier issues a wholesale price $a$, and the retailer issues a purchase quantity $b$, and a retail price $p$ in response.} \label{fig:supplier-retailer-game}
\end{figure}

At each time period the leader dynamically selects a wholesale price, $a$. The follower must make two decisions in response to $a$, the order amount, $b$, and the market price $p$, (which we show later by consequence of Eq. \ref{eq:newsvendor-max-int} are unique bijections). The uniqueness of the optimal $b_a^*$, given $a$, is proven in Theorem \ref{thm:optimal-order-b}. This game can be viewed as an ultimatum game, where the retailer is constrained to accept or reject the the terms dictated by the supplier. However, the inclusion of inventory risk and demand parameter uncertainty in our model complicates this dynamic. The retailer must weigh his ordering decision against his risk assessment in the face of these uncertainties. Thus the problem setting can be interpreted as a variant of the ultimatum game, where a learning agent operates under market uncertainty.

\subsubsection{Economic Theory of the Newsvendor Pricing Game} \label{sec:econ-newsv-theory}

Stochastic demand is represented in Eq. \ref{eq:additive-exp-demand}, which is governed by a linear additive demand function $\Gamma_\theta(p)$ representing the expected demand, $\mathbbm{E}[d(p)]$, as a function of $p$ in Eq. \ref{eq:additive-exp-demand}. The demand function is governed by parameters $\theta$.

\begin{align}
    \Gamma_\theta(p) &= \max \{0, \theta_0 - \theta_1 p \}, \quad \theta_0 \geq 0, \ \theta_1 \geq 0 \label{eq:demand-theta-func} \\
    d_\theta(p) &= \Gamma_\theta(p) + \epsilon, \quad \epsilon \in \mathcal{N}(0, \sigma) \label{eq:additive-exp-demand} 
\end{align}

We stipulate that noise $\epsilon \in \mathop{\mathbb{R}}$ is zero mean and $\sigma$-sub-Gaussian. The retailer profit, $\mathcal{G}_B(p, a)$ is restricted by the order quantity $b$, is therefore equal to, given wholesale price $a$,

\begin{align} 
    \mathcal{G}_A(a, b) &= a b \\
  \mathcal{G}_B(p, b| a) &= 
    \begin{cases} 
      d_\theta(p)p -ab , \quad &0 \leq d_\theta(p) \leq b \\ 
      b (p-a)  , \quad  &b < d_\theta(p) \\ 
    \end{cases} \label{eq:profit-cases}
\end{align} 

The retailer's profit should be greater than the supplier's wholesale price multiplied by the order quantity, denoted as $ab$. Let $f_\theta(\cdot)$ be the demand probability function given expected demand $\Gamma_\theta(p)$, and $F_\theta(\cdot)$ be the cumulative demand function. As the noise with respect to $d_\theta(p)$ is $\sigma$-sub-Gaussian, consequently by the additive property so is the noise with respect to the expected retailer profit given supplier cost $a$ is $\mathbbm{E}[\mathcal{G}_B(p, b| a)]$. Given any price $p$, we can compute expected profit of the retailer $\mathcal{G}_B$ as,

\begin{align}
     \mathbbm{E}[\mathcal{G}_B(p, b| a)] = p \int_0^b x f_\theta(x) dx + p b (1 - F_\theta(b)) - ab \label{eq:newsvendor-max-int}
\end{align}

Given $\Gamma_\theta(p)$ and $\epsilon \sim \mathcal{N}(0, \sigma)$ we can obtain a probability distribution $f_\theta(x)$, which has a unique solution \cite{arrow:1951newsboy}. The retailer can select $p$ generating $d \sim \Gamma_\theta(p)$, and retailer must decide on $b$ to order from the supplier. Suppose we express demand as $d \sim \Gamma_\theta(p)$, then the retailer must order an optimal amount $b_0(p, a)$ as determined by Eq. \eqref{eq:newsvendor-cf}. The general solution to the Newsvendor order quantity is to maximize profit via $\mathfrak{B}(a)$, the unique closed form solution is presented in Eq. \eqref{eq:newsvendor-cf}. In classical Newsvendor theory, the optimal solution is a function of the value $\gamma$, referred to as the \textit{critical fractile}, where $\gamma = (p-a)/p$.

\begin{align}
    b_0(p, a) = F_\theta^{-1}( \gamma ) = F_\theta^{-1}\Big(\frac{p-a}{p}\Big)  \label{eq:newsvendor-cf}
\end{align}

Unlike typical dynamic pricing, in a \textit{price-setting Newsvendor} problem setting, the reward maximizing price $p^*$ is not necessarily the same as the profit maximizing price $p_0$ under no inventory risk, also known as the \textit{riskless price}. From \citet{mills:1959uncertainty} we know that the optimal price to set in the \textit{price-setting Newsvendor} with linear additive demand is,

\begin{align}
    p^* &= p_0 - \frac{\Psi_\theta(b)}{2 \theta_1}, \quad \text{where,} \quad \Psi_\theta(b) = \int_b^{\infty} (x - b)f_{\theta, p}(x) dx  \label{eq:mills-opt-price}
\end{align}

Where $f_{\theta,p}(x)$ is the probability distribution of demand, such that $\mathbbm{E}[d_\theta(p)] = x$.

\begin{lemma}\label{thm:invariance-of-z}
    \textbf{Invariance of $\Psi_\theta(b)$:} Assuming symmetric behaviour of $f_{\theta, p}(\cdot)$ with respect to $\Gamma_\theta(p)$ and constant standard deviation $\sigma$, the term representing supply shortage probability defined as $\Psi_\theta(b)$, is invariant of the estimate of $\theta$. (Proof in Appendix \ref{prf:invariance-of-z}.)
\end{lemma}

The optimal price $p^*$, has the relation $p^* \leq p_0$. $f_{\theta,p}(x)$ expresses the probability of demand equal to $x$, under demand parameters $\theta$, and price $p$. This defines the unique relation of $b$ to $p(b)$ as evidenced by Eq. \eqref{eq:mills-opt-price}, and is a unique function. Knowing this relation, it is also logical to see that the riskless price, $p_0$, of the retailer must be always above the wholesale price, $a$, as introduced in the constraint, thus, $a < p^* \leq p_0 \xrightarrow[]{} a < p_0$. 

\begin{theorem} \label{thm:optimal-order-b}
    \textbf{Uniqueness of Best Response under Perfect Information:} In the perfect information NPG (from Sec. \ref{sec:newsvendor-games}), with a known demand function, under the assumptions listed in Sec. \ref{sec:assumptions}, the optimal order amount $b^*_a$ which maximizes expected profit of the retailer $\mathbbm{E}[\mathcal{G}_B(\cdot)]$ is $\mathfrak{B}(a)$ and is a surjection from $a \xrightarrow[]{} \mathfrak{B}(a)$, and a bijection when $a < \bar{p}$. Furthermore, within the range $a < \bar{p}$, as $a$ monotonically increases, $b_a^*$ is monotonically decreasing. Where $\bar{p}$ denotes the optimistic estimate of the riskless price $p_0$. (Proof in Appendix \ref{prf:optimal-order-b}.)
\end{theorem}


\begin{theorem} \label{thm:unique-se}
    \textbf{Unique Stackelberg Equilibrium:} A unique pure strategy Stackelberg equilibrium exists in the full-information Newsvendor pricing game. (Proof in Appendix \ref{prf:unique-se}.)
\end{theorem} 

\textbf{Sketch of Proof (Theorem \ref{thm:optimal-order-b} and \ref{thm:unique-se}):} Under perfect information, the unique solution to $a \xrightarrow[]{} \mathfrak{B}(a)$ in Eq. \eqref{eq:newsvendor-cf} can be obtained, so long as $\mathfrak{B}(a) > 0$. If each leader action $a$ elicits a unique action $\mathfrak{B}(a)$, then there must exist at least one optimal action for the leader in $a \in \mathcal{A}$. Note that, under perfect information, the follower would always act in accordance to their best approximation of $\mathcal{B}(a)$, which is unique under the assumptions in Section \ref{sec:assumptions}. This determines a unique best response. (Mathematical details are given in Appendix \ref{prf:optimal-order-b} and \ref{prf:unique-se}).

Evident from the previous discourse, $\mathbbm{E}[\mathcal{G}_B(p^*,b^*|a)]$, representing the optimal expected reward (profit) function for the follower under perfect information, does not lend itself to a straightforward computation. Instead, it necessitates inverse look up of a probability density function, from Eq. \eqref{eq:newsvendor-cf}, which has no parametric representation nor closed form solution. In addition, the solution to an integral equation from Eq. \eqref{eq:mills-opt-price} is required, which also lacks a closed-form solution. Consequently, this motivates the introduction of an upper-bounding approximation function, such as $\mathcal{U}_B(p|a)$ (represented in Eq. \eqref{eq:riskless-price-lin-func}), which both bounds $\mathbbm{E}[\mathcal{G}_B(p^*,b,a)]$, and is characterized by a straightforward parametric structure, $(\theta_0 - \theta_1p)(p-a)$. Observing the absence of the variable $b$ from the argument is noteworthy, albeit without consequence, as $\mathcal{U}_B(p|a)$ effectively provides an upper bound for $\mathcal{U}_B(p|a)$ across all admissible values of $b$ in $\mathbbm{E}[\mathcal{G}_B(p^*,b^*|a)]$. Subsequently, we can then leverage this approximation within an online learning framework.

\begin{align}
    \mathbbm{E}[\mathcal{G}_B(p, b|a)] \leq \mathcal{U}_B(p|a) &= (\theta_0 - \theta_1p)(p-a) \label{eq:riskless-price-lin-func} 
\end{align}

Note that the linear parametric form is preserved under this representation, where $\mathcal{U}_B(p|a) = \theta X $, with $\theta \equiv [\theta_0, \theta_1]$ and $X \equiv [(p-a), -(p^2 - ap)]^\intercal$.

\textbf{Economic Interpretation of $\mathcal{U}_B(p|a)$:} $\mathcal{U}_B(p|a)$ in Eq. \eqref{eq:riskless-price-lin-func} represents the theoretical reward the follower could achieve if all realized demand was met, and no cost of inventory surplus or shortage is considered. This is simply the unit profit margin, denoted as $p-a$, multiplied by the expected demand generated by the price, denoted as $\theta_0 - \theta_1p$, assuming all inventory is sold. Therefore, this \textit{riskless profit},  $\mathcal{U}_B(p|a)$, is effectively a concave polynomial with respect to $p$, and is always greater or equal than the expected reward $\mathbbm{E}[\mathcal{G}_B(p, b|a)]$ with inventory risk (a key result from \cite{mills:1959uncertainty} \cite{petruzzi:1999newsv}). 

\section{Online Learning Algorithms} \label{sec:learning-algorithms}

We present a learning algorithm for the \textit{Newsvendor Pricing Game} (NPG) involving two learning agents, namely two firms representing a supplier (leader) and retailer (follower), in a repeated Stackelberg game. In this online learning setting, both agents must learn the parameters of a reward function that is derived from a linear demand function. This presents additional challenge for both agents. For the leader, she bears the opportunity costs of suboptimal pricing given the best response of the retailer. For the retailer, he bears the cost of additional surplus or shortage when inaccurately estimating the demand function. These factors create additional complexity in comparison with a pure dynamic pricing problem setting. 

More specifically, we adapt a contextual linear bandit learning algorithm \cite{abbasi:2011ofu} \cite{chu:2011-linucb}, which already enjoys robust theoretical guarantees in a single-agent setting, to a multi-agent setting. Furthermore, we extend its application to the economic domain by integrating classical economic theory from the \textit{price-setting Newsvendor}, as discussed in Section \ref{sec:econ-newsv-theory}, with contextual linear bandits. In this game, the follower optimizes to maximize $\mathcal{U}_B(p|a)$, an optimistic upper-bounding concave function on his perspective reward function, while the leader aims to learn an optimal strategy by leveraging information about the follower's policy. We present technically novel analytical techniques for quantifying bounds on estimated parameters and best response functions. For example, we introduce an innovative argument via analysis in Euclidean space in Theorem \ref{thm:bound-del-ht} and approximation methods for non-elementary functions in Theorem \ref{thm:leader-regret}. Given this, we successfully extend the robust theoretical guarantees of single-agent contextual linear bandits to an economic multi-agent setting.


\subsection{Optimism in the Face of Uncertainty}
 
We propose an adaptation of the \textit{optimism-under-uncertainty} (OFUL) algorithm from \citet{abbasi:2011ofu} (illustrated in Appendix \ref{sec:alg-abbasi} for reference) to learn the parameters of demand from Eq. \eqref{eq:demand-theta-func}. The OFUL algorithm is a contextual linear bandit algorithm, also previously studied in \citet{dani:2008-ball} and \citet{rusmevichientong:2010linearly}, which constructs a confidence ball $\mathcal{C}^t$ where the parameters of an estimated reward function fall within. 

\begin{align}
    \norm{\theta^*_t - \theta_t}_{\bar{V}_t} \in \mathcal{O}(\sqrt{\log{t}}) \label{eq:self-norm-abs-bound}
\end{align}

 Given context $a$, we have two parameters to estimate $\theta_0$ and $\theta_1$, leading us to learning of the optimal riskless price optimization. From Eq. \eqref{eq:self-norm-abs-bound} we see that $||\theta^*_t - \theta_t||_{\bar{V}_t}$ grows logarithmically, matching the covariance of $\theta_t$, where $|| \cdot ||_{\bar{V}_t}$ is defined as the self-normalizing norm \cite{delapena:2004_self_norm}. Lemma \ref{thm:theta-2-norm} extends this relationship to a shrinking finite 2-norm bound, which we later use to derive bounds on regret.

\begin{lemma} \label{thm:theta-2-norm}
     With probability $1-\delta$, and $t \geq e$, the constraint $||\theta - \theta^*||_{\bar{V}_t} \leq \mathcal{O}(\sqrt{\log(t)})$ from \citet{abbasi:2011ofu} implies the finite $2$-norm bound $||\theta - \theta^*||_2 \leq \mathcal{O}(\sqrt{\log(t)/ t})$, and also thereby the existence of $\kappa > 0$ such that $||\theta - \theta^*||_2 \leq \kappa \sqrt{\log(t)/ t}$. (Proof in Appendix \ref{prf:theta-2-norm}.)  
\end{lemma}


\subsection{Online Learning Algorithm for Newsvendor Pricing Game}

Let $\mathcal{H}(\theta) \equiv \theta_0/\theta_1$. We introduce a learning algorithm where both the leader and the follower optimistically estimates $\mathcal{H}(\theta)$, while the follower learns $\mathcal{U}_B(p_0, b_0 | a)$. We denote by $\mathcal{C}^t$, the confidence set, wherein the true parameters lie with high probability $1-\delta$ (i.e. $\theta^* \in \mathcal{C}^t$). Our intuition, supported by Lemma \ref{thm:theta-2-norm}, suggests that the size of $\mathcal{C}^t$ decreases linearly with the sample size $t$. Moreover, the linear structure of $\mathcal{U}_B(p_0, b_0 | a)$ enables Lemma \ref{thm:theta-2-norm} to remain valid during estimation of $\theta$ via linear contextual bandits. To highlight, the follower executes the optimistic linear contextual bandit, while the leader selects her actions maximizing over the follower's optimistic best response.


\begin{algorithm}[h!]
\caption{Learning Algorithm for Newsvendor Pricing Game (LNPG) }\label{alg:se-newsv}
\begin{algorithmic}[1]
    \For {$t \in 1 ... T$}:
        \State Leader and follower estimates $\mathcal{C}^t$ from available data $\mathcal{D}^t$. 
        \State Leader plays action $a$, where $a = \underset{a \in \mathcal{A}, \theta \in \mathcal{C}^t} {\mathrm{argmax}} \ a F^{-1}_{\bar{\theta}_a} \Big( 1 - \frac{2a}{ \mathcal{H}(\theta) + a} \Big)$ from Eq. \eqref{eq:argmax-a-obj}.
        \State Follower sets price $p = (\mathcal{H}(\theta) + a)/2$, referencing the maximization problem in Eq. \eqref{eq:h-max-theta}.
        \State Follower estimates their optimistic parameters $\bar{\theta}_a$, and best response $\bar{b}_a$ from Eq. \eqref{eq:ba_upper_bound_true} and \eqref{eq:h-max-theta} respectively, and plays action $b = \bar{b}_a$.
        \State Leader obtains reward, $\mathcal{G}_A = ab$.
        \State Follower obtains reward, $\mathcal{G}_B = p \min \{ b, d(p) \} $ (Eq. \eqref{eq:additive-exp-demand}).
    \EndFor
\end{algorithmic}
\end{algorithm}


Given leader action $a$, the optimal follower, $\mathcal{G}_B^*(\cdot)$ is unique from Theorem \ref{thm:unique-se}. Thus there there exists a unique solution where $\mathcal{G}_A^*(a)$ is maximized for any estimate of $\theta$. Although Theorem \ref{thm:unique-se} illustrates the uniqueness of the SE under perfect information, there still exists uncertainty under imperfect information, when the parameters of demand are unknown. From the leader's perspective, she is rewarded $\mathcal{G}_A(a, b) = ab$ depending on the reaction of the follower. The follower's economic reaction is determined by $F_{\theta}^{-1}((p_0-a)/p_0)$, which is affected by estimate $\theta$. We present empirical results of Algorithm \ref{alg:se-newsv} in Appendix \ref{sec:experiment-results}.

\subsection{Measuring the Best Response under Optimism} \label{sec:best-response}

To characterize the best response of the follower, we must first establish a bound on the optimistic, yet uncertain, pricing behaviour $p_a$ of the follower given leader action $a$. From Theorem \eqref{thm:optimal-order-b} there exists a direct relationship between $p^t$ and $b_a^t$. As the follower is learning parameters via OFUL, its certainty about the model parameters is bounded by $||\theta - \theta^*||_2 \leq \mathcal{O}(\sqrt{\log(t)/t} )$. For each given action $a$, and a confidence bound on the parameters $\theta^*$, denoted as $\mathcal{C}^t$, we can have an optimistic and pessimistic estimate of the best response, $b_a$, denoted as $(\bar{b}_a, \underline{b}_a)$ respectively. We define the pessimistic estimate of $b_a^*$ as,


\begin{align}
    \underline{b}_a &\equiv \max \Big\{0, F_{\hat{\theta}}^{-1} \Big(\frac{p_0(\theta) - a}{ p_0(\theta) } \Big) \Big\} \geq \max \Big\{ 0, F_{\hat{\theta}}^{-1} \Big(\frac{p^*(\theta) - a}{ p^*(\theta) } \Big) \Big\} \label{eq:ba_lower_bound_true} 
\end{align}

The pessimistic best response $\underline{b}_a$ to $a$ occurs when the optimistic parameter estimate is equal to the maximum likelihood estimate itself, denoted as $\hat{\theta}$. Thus for realism, we set a lower bound of 0 for $b_a$. Additionally, we recognize that the riskless price estimate $p_0$ is always greater than or equal to $p^*$, represented as $a \leq p^* \leq p_0$ in Eq. \eqref{eq:ba_lower_bound_true}. Later, in Lemma \ref{thm:opt-p-star-decrease}, we will demonstrate that when the algorithm converges to $p_0$, the learned parameters enable the follower to calculate $p^*$ and $b_a^*$ respectively.

The upper bound of the optimistic best response, $\bar{b}_a$ (Eq. \eqref{eq:ba_upper_bound_true}), can be determined by estimating the optimistic riskless price $\bar{p}_0$. Let $\mathbb{E}[\mathcal{G}_B(\theta|a)]$ denote the expected profit under parameter $\theta$, give $a$. Both the left and right terms of Eq. \eqref{eq:max-theta} are optimistic estimates of expected profit, $\mathbb{E}[\mathcal{G}_B(\theta|a)]$, using the riskless price $p_0$ as a surrogate for the optimal price. The follower's pricing with $\bar{p}_0$, and ordering with $\bar{b}_a$, under an optimistic estimate of the demand parameters $\theta$ given data $\mathcal{D}^t$, can be summarized from Eq. \eqref{eq:max-theta} to Eq. \eqref{eq:ba_upper_bound_true}. (The derivation from Eq. \eqref{eq:max-theta} to Eq. \eqref{eq:ba_upper_bound_true} is in Appendix \ref{sec:notes-opt-br}.) 

\begin{align}
    \bar{\theta} &= \underset{\theta \in \mathcal{C}^t} {\mathrm{argmax}} \ \mathbb{E}[\mathcal{G}_B(\theta|a)], \quad \mathbb{E}[\mathcal{G}_B(\theta|a)] \leq \underset{\theta \in \mathcal{C}^t} {\mathrm{max}} \ \Big( \frac{\mathcal{H}(\theta) + a}{2} \Big)  F_{\theta}^{-1} \Big(1 - \frac{2a}{\mathcal{H}(\theta) + a} \Big) \label{eq:max-theta} \\
    \bar{b}_a &= F^{-1}_{\bar{\theta}_a} \Big( \frac{p_0 - a}{p_0} \Big), \quad p_0 = a + \frac{\bar{p}_0 - a}{2} = \frac{\mathcal{H}(\bar{\theta}) + a}{2}, \quad \mathfrak{B}(a) \in [\underline{b}_a, \bar{b}_a], \quad \underline{b}_a \geq 0  \label{eq:ba_upper_bound_true}
\end{align}

\begin{theorem} \label{thm:bound-del-ht}
    \textbf{Bounding the Optimistic Follower Pricing:} Let $\Delta_\mathcal{H}(t) = \mathcal{H}^*(\theta) -  \theta_0^*/\theta_1^*$ represent the deviation in pricing action of the follower (retailer) at time $t$. Then $\Delta_\mathcal{H}(t) \leq  \kappa_H \sqrt{ \log(t)/t } $ for some $\kappa_H > 0$ and $t \geq e$. (Proof in Appendix \ref{prf:bound-del-ht}.)
\end{theorem}

\begin{subequations}
    \begin{alignat}{2}
        &\!\max_{\theta}   &\qquad \qquad \mathcal{H}(\theta) &= \ \frac{\theta_0}{\theta_1}  \label{eq:h-max-theta} \\
        &\text{subject to} &\qquad       \sqrt{ (\theta_0 - \theta^*_0)^2 + (\theta_1 - \theta^*_1)^2 } &\leq \kappa \sqrt{\log(t)/t}  \label{eq:theta_constr}
    \end{alignat}
\end{subequations}

\textbf{Sketch of Proof:} As illustrated in Fig. \ref{fig:conf-ball}, we specify the optimization problem, where with probability $1- \delta$, the 2-norm error of the parameter estimates for $\theta$ lie within the confidence ball characterized by a radius less than $\kappa \sqrt{\log(t)/t}$. The general maximization problem is expressed in Eq. \eqref{eq:h-max-theta} and \eqref{eq:theta_constr}. The optimistic estimate of the $\theta$, denoted as $\mathcal{H}^*(\theta)$ (which we can denote in short as $\mathcal{H}^*$) lies in the extreme point between a ray $\mathcal{H}(\theta)$ and the boundary of $||\theta - \theta^*||_2 \leq \kappa \sqrt{\log(t)} / \sqrt{t}$.

\begin{align}
        \mathcal{H}^* \theta_1 &= \theta_0 \label{eq:theta_line} \\
        \frac{| \mathcal{H}^* \theta_1 - \theta_0 |}{\sqrt{ {\mathcal{H}^*}^2+1 }} &= \kappa\sqrt{\log(t)/t} \label{eq:point-to-line}    
\end{align}

Leveraging Eq. \eqref{eq:point-to-line}, we can construct a bound on $\Delta_\mathcal{H}(t)$. Theorem \ref{thm:bound-del-ht} characterizes the shrinking nature of $\Delta_\mathcal{H}(t)$ with respect to $t$, as more samples are observed. Similar to Lemma \ref{thm:theta-2-norm}, there also exists some $\kappa_H$, where $\kappa_H \sqrt{\log(t)/t}$ bounds $\Delta_\mathcal{H}(T)$ as expressed in Theorem \ref{thm:bound-del-ht}. The establishment of this bound characterizes the pricing deviation of the follower pricing strategy, and is used to prove Theorem \ref{thm:follower-regret}. To obtain a closed form solution to this optimization problem, we note the closed form expression of the Euclidean distance between a point in 2-space $(\theta_0, \theta_1)$ to a line described in Eq. \eqref{eq:theta_line}. This distance measure is expressed  in Eq. \eqref{eq:point-to-line}. %
\begin{figure}[!htb]
\minipage{0.4\textwidth}
    \tikzset{every picture/.style={line width=0.75pt}} \hspace*{-2.4em}
    \begin{tikzpicture}[x=0.75pt,y=0.75pt,yscale=-0.73,xscale=0.73,scale=0.73] 

        \draw    (51,62.64) -- (51.19,330.91) ;
        \draw   (51.19,330.91) -- (509.19,330.91) ;

        \draw [dotted, color={rgb, 255:red, 74; green, 144; blue, 226 },draw opacity=1 ][fill={rgb, 255:red, 74; green, 144; blue, 226 }  ,fill opacity=1 ][line width=1.5]    (51.2,142.16) -- (228.2,330.16) ;
        \draw [dotted, color={rgb, 255:red, 74; green, 144; blue, 226 },draw opacity=1 ][fill={rgb, 255:red, 74; green, 144; blue, 226 }  ,fill opacity=1 ][line width=1.5]    (51.2,80.16) -- (470.2,331.16) ;
        \fill[color={rgb, 255:red, 202; green, 235; blue, 245 }] (51.2,142.16) -- (228.2,330.16) -- (470.2,331.16) -- (51.2,80.16);
        

        \draw [color={rgb, 255:red, 126; green, 211; blue, 33 }  ,draw opacity=1 ][line width=1.5]     (51.2,112.16) -- (360.2,330.16) ;
        \draw [color={rgb, 255:red, 126; green, 211; blue, 33 }  ,draw opacity=1 ]   (382,239) .. controls (329.27,273.83) and (326.03,221.53) .. (271.82,262.38) ;
        \draw [shift={(271,263)}, rotate = 322.63] [color={rgb, 255:red, 126; green, 211; blue, 33 }  ,draw opacity=1 ][line width=0.75]    (10.93,-3.29) .. controls (6.95,-1.4) and (3.31,-0.3) .. (0,0) .. controls (3.31,0.3) and (6.95,1.4) .. (10.93,3.29)   ;

        
        \draw [dashed, color={rgb, 255:red, 255; green, 2; blue, 2 },draw opacity=1 ]   (99,106) -- (99,331.67) ;

        \draw [color={rgb, 255:red, 74; green, 144; blue, 226 }  ,draw opacity=1 ]   (221.43,126.29) .. controls (168.69,161.11) and (201.1,93.96) .. (147.25,134.66) ;
        \draw [shift={(146.43,135.29)}, rotate = 322.63] [color={rgb, 255:red, 74; green, 144; blue, 226 }  ,draw opacity=1 ][line width=0.75]    (10.93,-3.29) .. controls (6.95,-1.4) and (3.31,-0.3) .. (0,0) .. controls (3.31,0.3) and (6.95,1.4) .. (10.93,3.29)   ;
        
        \draw [dashed, color={rgb, 255:red, 74; green, 144; blue, 226 }  ,draw opacity=1 ]   (51,255) -- (335,255);
        \draw [dashed, color={rgb, 255:red, 74; green, 144; blue, 226 }  ,draw opacity=1 ]   (335,255) -- (335,331.67);

        \draw[dashed,color={rgb, 255:red, 126; green, 211; blue, 33 }  ,draw opacity=1 ]    (51,227) -- (214,227) ;
        \draw[dashed, color={rgb, 255:red, 126; green, 211; blue, 33 }  ,draw opacity=1 ]   (214,227) -- (214,331.67) ;

        \draw (330,195) node [anchor=north west][inner sep=0.75pt]   [align=left] {$\Gamma_\theta(p) = \theta_0 - \theta_1 p$};
        
        \draw (93,343) node [anchor=north west][inner sep=0.75pt]   [align=left] {$a$};
        \draw (325,343) node [anchor=north west][inner sep=0.75pt]   [align=left] {$\bar{p}_0$};
        \draw (465,343) node [anchor=north west][inner sep=0.75pt]   [align=left] {$\bar{p}$};
        \draw (-15,176.45) node [anchor=north west][inner sep=0.75pt]   [align=left] {$\mathbbm{E}[d]$};
        \draw (199,343) node [anchor=north west][inner sep=0.75pt]   [align=left] {$p_0$};
        \draw (231,109) node [anchor=north west][inner sep=0.75pt]   [align=left] {$||\theta - \theta^*||_2 \leq \kappa \sqrt{\log(t)/t}$};
        \end{tikzpicture}
    \caption{Linear demand uncertainty.} \label{fig:lin-demand}
\endminipage\hfill
\minipage{0.4\textwidth}
    \tikzset{every picture/.style={line width=0.75pt}} \hspace*{0.6em} 
    \begin{tikzpicture}[x=0.75pt,y=0.75pt,yscale=-0.71,xscale=0.71,scale=0.71]
        
        \draw    (131.6,63.28) -- (131.88,345.05) ;
        \draw   (131.88,345.05) -- (441.14,345.05) ;
        \draw[dashed, color={rgb, 255:red, 74; green, 144; blue, 226 }  ,draw opacity=1]   (233.07,204.74) .. controls (232.81,176.11) and (255.8,152.69) .. (284.43,152.44) .. controls (313.05,152.18) and (336.47,175.17) .. (336.73,203.8) .. controls (336.99,232.42) and (314,255.84) .. (285.37,256.1) .. controls (256.74,256.36) and (233.33,233.36) .. (233.07,204.74) -- cycle ;
        \fill[color={rgb, 255:red, 202; green, 235; blue, 245}]   (233.07,204.74) .. controls (232.81,176.11) and (255.8,152.69) .. (284.43,152.44) .. controls (313.05,152.18) and (336.47,175.17) .. (336.73,203.8) .. controls (336.99,232.42) and (314,255.84) .. (285.37,256.1) .. controls (256.74,256.36) and (233.33,233.36) .. (233.07,204.74) -- cycle ;
        
        \draw [color={rgb, 255:red, 74; green, 144; blue, 226 }  ,draw opacity=1 ]   (131.88,345.05) -- (285.22,108.9) ;
        \draw [shift={(286.31,107.23)}, rotate = 123] [color={rgb, 255:red, 74; green, 144; blue, 226 }  ,draw opacity=1 ][line width=0.75]    (10.93,-3.29) .. controls (6.95,-1.4) and (3.31,-0.3) .. (0,0) .. controls (3.31,0.3) and (6.95,1.4) .. (10.93,3.29)   ;
        \draw [color={rgb, 255:red, 208; green, 2; blue, 27 }  ,draw opacity=1 ]   (131.88,345.05) -- (398.46,202.4) ;
        \draw [shift={(400.22,201.45)}, rotate = 151.85] [color={rgb, 255:red, 208; green, 2; blue, 27 }  ,draw opacity=1 ][line width=0.75]    (10.93,-3.29) .. controls (6.95,-1.4) and (3.31,-0.3) .. (0,0) .. controls (3.31,0.3) and (6.95,1.4) .. (10.93,3.29)   ;
        
        \draw[dashed] (131.74,204.16) -- (284.9,204.27) ;
        \draw[dashed]    (284.9,204.27) -- (285.18,345.89) ;
        \draw[dashed]  (244.11,171.92) -- (284.9,204.27) ;
        

        \draw   (285.95,202.16) .. controls (288.92,198.55) and (288.6,195.27) .. (284.99,192.31) -- (281.43,189.38) .. controls (276.28,185.15) and (275.18,181.23) .. (278.14,177.63) .. controls (275.18,181.23) and (271.12,180.92) .. (265.97,176.69)(268.29,178.59) -- (256.43,168.84) .. controls (252.82,165.88) and (249.54,166.2) .. (246.57,169.81) ;
        \draw    (371.43,155.43) .. controls (318.25,200.21) and (337.77,121.22) .. (282.27,173.63) ;
        \draw [shift={(281.43,174.43)}, rotate = 316.31] [color={rgb, 255:red, 0; green, 0; blue, 0 }  ][line width=0.75]    (10.93,-3.29) .. controls (6.95,-1.4) and (3.31,-0.3) .. (0,0) .. controls (3.31,0.3) and (6.95,1.4) .. (10.93,3.29)   ;
        
        \draw (94,192.05) node [anchor=north west][inner sep=0.75pt]   [align=left] {$\theta_0$};
        \draw (280,355) node [anchor=north west][inner sep=0.75pt]   [align=left] {$\theta_1$};
        \draw (250,71) node [anchor=north west][inner sep=0.9pt]   [align=left] {$\bar{p} = \underset{\theta \in \mathcal{C}^t } { \mathrm{argmax}} \ \theta_0/\theta_1 $};
        \draw (378,128) node [anchor=north west][inner sep=0.75pt]   [align=left] {$\kappa \sqrt{\log(t)/t}$};
        
    \end{tikzpicture}
    \caption{Confidence ball.} \label{fig:conf-ball}
\endminipage\hfill
\caption*{Fig. \ref{fig:lin-demand} represents the uncertainty of the linear demand function given $||\theta - \theta^*||_2 \leq \kappa \sqrt{\log(t)/t}$. We can see the optimistic estimate of the riskless price, $\bar{p}_0$ is proportional to the estimates of $
\bar{p} = \theta_0/\theta_1$. Fig. \ref{fig:conf-ball} visualizes a confidence ball with the radius $\kappa \sqrt{\log(t)/t}$. The extreme values of $\bar{p}$ are represented by the rays extending from the origin and intersecting the maximum radius of the confidence ball.} 
\end{figure}


From Lemma \ref{thm:cum-del-HT} we can see that the cumulative optimistic pricing deviation of the follower is $\Delta_\mathcal{H}(T) \in \mathcal{O}(\sqrt{T \log(T)})$. This will have further ramifications when we compute the Stackelberg regret in Theorem \ref{thm:leader-regret}. For interest to the reader, the current bound in Lemman \ref{thm:cum-del-HT} would naturally extend to non-linear settings such as exponential class General Linear Models (GLM's) (we provide a description of this Appendix \ref{sec:non_lin_demand_exp}).

\begin{lemma} \label{thm:cum-del-HT}
    \textbf{Bounding the Cumulation of $\Delta_\mathcal{H}(t)$}: Suppose we wish to measure the cumulative value of $\Delta_\mathcal{H}(t)$ from $1$ to $T$, denoted as $\Delta_\mathcal{H}(T)$. Then  $\Delta_\mathcal{H}(T) \leq 2 \kappa_H  \sqrt{T \log(T)} $. (Proof in Appendix \ref{prf:cum-del-HT}.)
\end{lemma}

\begin{lemma} \label{thm:opt-p-star-decrease}
    \textbf{Convergence  of $\bar{p}^*$ to $p^*$:}  Suppose the solution to Eq. \eqref{eq:mills-opt-price} is computable. For any symmetric demand distribution, the optimistic estimate of the optimal market price from Eq. \eqref{eq:mills-opt-price}, denoted as $\bar{p}^*$ which is the solution to $p^*$ under optimistic parameters $\bar{\theta}$, is decreasing with respect to $T$ such that $\bar{p}^*(\theta_{T+1}) \leq \bar{p}^*(\theta_T)$, where the  difference $\bar{p}^*(\theta_T) - p^*$ is bounded by $\mathcal{O}(\log(T)/T)$. (Proof in Appendix \ref{prf:opt-p-star-decrease} )
\end{lemma} 

In Algorithm \ref{alg:se-newsv} we propose that the follower prices the product at the optimistic riskless price $\bar{p}_0$, and not necessarily the optimistic optimal price $\bar{p}^*$. The profits derived from pricing with any of these suboptimal prices would result in suboptimal follower rewards (we characterize this regret later in Theorem \ref{thm:follower-regret}).

\subsection{Analysis from the Follower's Perspective} \label{sec:oful-npg-follower-reg}

For any approximately best responding follower, we can always expect a worst case error of $\epsilon_B$, defined in Eq. \eqref{eq:br_approx_def_eps_b}, for the follower at any discrete time interval. Thus, by intuition alone, we can presume the regret to be at least $\mathcal{O}(T)$. Nevertheless, it is worth noting that there exists these two distinct components to the regret, one caused specifically by the approximation of the best response, and one caused by the parameter uncertainty of the learning algorithm itself. The problem ultimately reduces to a contextual linear bandit problem.

Given our bounds on the optimistic pricing behaviour of the follower under context $a$ and uncertainty $\theta \in \mathcal{C}^t$, for any sequence of actions, $\{ a^1, ...  , a^T \}$, the expected regret of the follower is expressed in Eq. \eqref{eq:follower-regret}.



\begin{theorem} \label{thm:follower-regret}
    \textbf{Follower Regret in Online Learning for NPG:} Given pure leader strategy $\pi_A$, the worst-case regret of the follower, $R_B^T(\pi_A)$, as defined in Eq. \eqref{eq:follower-regret} when adopting the LNPG algorithm, is bounded by $R_B^T(\pi_A) \leq \aleph_B + \epsilon_B T - \theta_1 \kappa_H^2 \log^2(T)$, for sufficiently large values of $T$, where $\aleph_B, \epsilon_B, \theta_1, \kappa_H$ are positive real constants in $\mathbb{R}^+$. (Proof in Appendix \ref{prf:follower-regret}.)
\end{theorem}

\textbf{Sketch of Proof:} As the follower implements the OFUL algorithm (see Appendix \ref{sec:alg-abbasi}) to estimate the demand parameters $\theta$, we propose $\mathcal{U}_B(p_0|a)$ as the theoretically maximum achievable expected reward. We quantify the variation in $\mathcal{U}_B(\bar{p}_0|a)$, represented by $\Delta_\mathcal{U}(t)$, across a finite number of samples $T$. Furthermore, we provide an upper bound on the instantaneous regret using the larger upper-bounding discrepancy $\mathcal{U}_B(\bar{p}_0|a) - \mathbb{E}[\mathcal{G}_B(p^*|a)]$. It is worth noting that although $\mathcal{U}_B(p|a)$ is a concave function, $\mathbb{E}[\mathcal{G}_B(p|a)]$ may not be. Under perfect information, the follower may incur an approximation error $\tilde{\epsilon}$ derived from this discrepancy, rendering it $\tilde{\epsilon}$-suboptimal in the worst case scenario. Nonetheless, this error is independent of $t$, and the learning error for $\theta$. Next, we establish a bound for the linear function estimation to compute optimal riskless pricing, for this we can apply the results of Theorem \ref{thm:bound-del-ht}, and establish the bound of $\Delta_\mathcal{U}(t) \leq \theta_1 \kappa_H^2 \log(t)/t$. We then partition the learning time period into two phases: Case 2 denotes the period where $\mathbbm{E}[\mathcal{G}_B(p^*|a)] > \mathcal{U}_B(\bar{p}_0|a)$, occurring within the time interval $1 \leq t \leq T'$, and Case 1 denotes when $\mathbbm{E}[\mathcal{G}_B(p^*|a)] \leq \mathcal{U}_B(\bar{p}_0|a)$, occurring within the interval $T' \leq t \leq T$ (assuming sufficiently large values of $T$). We note the differences in regret behaviour within these regimes and focus on the asymptotic regime as $T$ approaches $\infty$, where the cumulative regret is bounded between the worst-case approximation error for $\mathbbm{E}[\mathcal{G}_B(p^*|a)]$, denoted as $\epsilon_B$, and the logarithmic term $\Delta_\mathcal{U}(t)$ pertaining to linear stochastic bandit learning, presented as $\aleph_B + \epsilon_B T - \theta_1 \kappa_H^2 \log^2(T)$. (Complete details are provided in Appendix \ref{prf:follower-regret}.)

To provide further insight, Theorem \ref{thm:follower-regret} suggests that the regret is bounded by the linear worst case approximation error $\epsilon_B T$, subtracted by the sublinear term, $\theta_1 \kappa_H^2 \log^2(T)$, for sufficiently large values of $T$. As the follower forms an approximate best response to the leader's action $a$ based on optimistic riskless pricing, we can expect linear regret in the worst-case. Nevertheless, knowing the characteristic behaviour of this regret is useful when we can quantify the learning error caused by the approximation of $\mathbb{E}[\mathcal{G}_B(\cdot)]$ with $\mathcal{U}_B(\cdot)$.


\subsection{Stackelberg Regret in the Newsvendor Pricing Game} \label{sec:leader_stackelberg_regret}

We focus on providing a theoretical worst-case guarantee of the regret from the leader's perspective as defined in Eq. \eqref{eq:leader-regret-def}, also known as Stackelberg regret under an approximate best responding follower. We see that given the equation Eq. \eqref{eq:newsvendor-cf}, the critical fractile used to compute $b$ is decreasing with respect to $a$. This makes $F^{-1}_\theta$ also strictly decreasing with respect to $a$, so $b$ is strictly decreasing with respect to $a$. Therefore there exists a unique $a$, under the assumptions provided in Section \ref{sec:assumptions}, which maximizes the upper bound on the leader profit $\mathcal{G}_A^*(\cdot)$, given the best response of the follower, $b$, for any leader action $a$.

\begin{align}
    \mathcal{G}_A^*(a, \mathfrak{B}(a)) \leq  \underset{a \in \mathcal{A} }{\mathrm{max}} \ a F_{\theta}^{-1} \Big(\frac{p_0 - a}{p_0} \Big) = \underset{a \in \mathcal{A} }{\mathrm{max}} \ a F_{\theta}^{-1} \Big(  1 - \frac{ 2 a }{ \mathcal{H}(\theta) + a} \Big) \label{eq:argmax-a-obj}
\end{align}

To bound the Stackelberg regret, we must effectively characterize the best response of the follower given the assumptions outlined in Sec. \ref{sec:assumptions}. The key idea is that the follower will 1) act optimistically in the face of uncertainty, and rationally and be greedy under his best approximation of the best response. Meaning he will not try to deceive or misplay its own actions in order to influence the leader's future actions.


 If the follower's strategy is known to the leader, and that the information available to both agents are the same, then in theory there exists no regret from the leader's perspective, as the leader has perfect premonition over the actions of the follower. In our case, regret comes from the existence of information asymmetry, for example when the follower has received more observations about the demand than the leader, and therefore tends to be more confident in their estimate of $\Gamma_\theta(p)$, and less optimistic in their ordering and pricing strategies. Algorithm \ref{alg:se-newsv} (LNPG) seeks to minimize the Stackelberg regret, as defined in Eq. \eqref{eq:leader-regret-def}, constituting a no-regret learner from the over the follower's strategy from the leader's perspective. Where $a^*$ denotes the optimal action from the leader's perspective. 
 
\begin{align}
    a^* &= \underset{a \in \mathcal{A}} { \mathrm{argmax}} \ \mathcal{G}_A(a, \mathfrak{B}(a)) \label{eq:a-max-leader}
\end{align}


 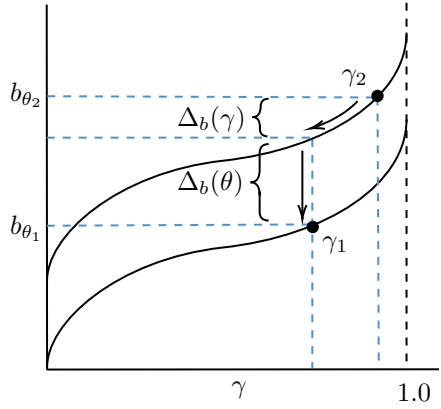
\begin{figure}
    \centering
    \tikzset{every picture/.style={line width=0.75pt}} 
        \begin{tikzpicture}[x=0.75pt,y=0.75pt,yscale=-0.65,xscale=0.65]
            
            \draw    (145.6,89.28) -- (145.88,371.05) ;
            \draw    (145.88,371.05) -- (448.43,370.86) ;
            \draw    (343,202) -- (343,253) ;
            \draw [shift={(343,255)}, rotate = 270] [color={rgb, 255:red, 0; green, 0; blue, 0 }  ][line width=0.75]    (10.93,-3.29) .. controls (6.95,-1.4) and (3.31,-0.3) .. (0,0) .. controls (3.31,0.3) and (6.95,1.4) .. (10.93,3.29)   ;
            \draw    (145.88,371.05) .. controls (144.43,334.86) and (200.43,285.86) .. (281.43,276.86) .. controls (362.43,267.86) and (423.43,227.86) .. (423.43,177.86) ;
            \draw[dashed]    (423.43,87.86) -- (423.43,370.86) ;
            \draw [dashed, color={rgb, 255:red, 74; green, 144; blue, 226 }  ,draw opacity=1 ]   (350.67,193) -- (350.67,370) ;
            \draw    (145.88,304.05) .. controls (144.43,267.86) and (200.43,218.86) .. (281.43,209.86) .. controls (362.43,200.86) and (423.43,160.86) .. (423.43,110.86) ;
            \draw    (385.6,163.8) .. controls (377.05,173.3) and (362.18,180.99) .. (350.77,184.48) ;
            \draw [shift={(349,185)}, rotate = 344.58] [color={rgb, 255:red, 0; green, 0; blue, 0 }  ][line width=0.75]    (10.93,-3.29) .. controls (6.95,-1.4) and (3.31,-0.3) .. (0,0) .. controls (3.31,0.3) and (6.95,1.4) .. (10.93,3.29)   ;
            \draw [dashed, color={rgb, 255:red, 74; green, 144; blue, 226 }  ,draw opacity=1 ]   (145.67,192) -- (350,192) ;
            \draw [dashed, color={rgb, 255:red, 74; green, 144; blue, 226 }  ,draw opacity=1 ]   (145.67,260) -- (350.67,259) ;
            \draw [dashed, color={rgb, 255:red, 74; green, 144; blue, 226 }  ,draw opacity=1 ]   (401,161) -- (402,370) ;
            \draw [dashed, color={rgb, 255:red, 74; green, 144; blue, 226 }  ,draw opacity=1 ]   (146,160) -- (401,161) ;
            \draw   (317,162) .. controls (313.02,161.86) and (310.96,163.78) .. (310.82,167.77) -- (310.82,167.77) .. controls (310.63,173.45) and (308.54,176.22) .. (304.56,176.09) .. controls (308.54,176.22) and (310.43,179.13) .. (310.24,184.82)(310.33,182.27) -- (310.24,184.82) .. controls (310.1,188.8) and (312.02,190.86) .. (316,191) ;
            \draw   (317,197) .. controls (312.33,197) and (310,199.33) .. (310,204) -- (310,216.5) .. controls (310,223.17) and (307.67,226.5) .. (303,226.5) .. controls (307.67,226.5) and (310,229.83) .. (310,236.5)(310,233.5) -- (310,249) .. controls (310,253.67) and (312.33,256) .. (317,256) ;
            
            \draw (414,379.91) node [anchor=north west][inner sep=0.75pt]   [align=left] {1.0};
            \draw (115,250) node [anchor=north west][inner sep=0.75pt]   [align=left] {$b_{\theta_1}$};
            \draw (115,147) node [anchor=north west][inner sep=0.75pt]   [align=left] {$b_{\theta_2}$};
            \draw (286,377.91) node [anchor=north west][inner sep=0.75pt]   [align=left] {$\gamma$};
            \draw (244,165) node [anchor=north west][inner sep=0.75pt]   [align=left] {$\Delta_b(\gamma)$};
            \draw (244,215) node [anchor=north west][inner sep=0.75pt]   [align=left] {$\Delta_b(\theta)$};

            \filldraw[black] (401,159.8) circle (3pt) node[anchor=south east]{$\gamma_2$};
            \filldraw[black] (351,261) circle (3pt) node[anchor=north west]{$\gamma_1$};

        \end{tikzpicture}
    \caption{Denotes the change in the optimistic order amount $b_a$ as improved estimates of the demand function are obtained.}
    \label{fig:cum-dist-shift}
\end{figure} 


\begin{theorem} \label{thm:leader-regret}
    \textbf{Stackelberg Regret:} Given the pure strategy best response of the follower defined $\mathfrak{B}(a) \in [\underline{b}_a, \bar{b}_a]$ from Eq. \eqref{eq:ba_lower_bound_true} and \eqref{eq:ba_upper_bound_true} respectively, the worst-case regret, $R_A^T(\pi_A)$, from the leader's perspective, as defined in Eq. \eqref{eq:leader-regret-def} when adopting the LNPG algorithm, is bounded by $\mathcal{O}(\sqrt{T \log(T)})$. (Proof in Appendix \ref{prf:leader-regret}.)
\end{theorem}

\textbf{Sketch of Proof:} The technicalities in the proof for Theorem \ref{thm:leader-regret} are best illustrated in Fig. \ref{fig:cum-dist-shift}. The change in the best response of the follower are affected by two components. $\Delta_b(\gamma)$ reflects the change in $b_a$ due to changes in the critical fractile $\gamma$, and $\Delta_b(\theta)$ reflects the change in $b_a$ due to the overall shift in the expected demand due to the increasing confidence on $||\theta^* - \hat{\theta}||_2$. We construct linear function approximator, $f(\gamma) \simeq F_\theta^{-1}(\gamma)$ with a bounded estimation error. One existing challenge is that $F_\theta^{-1}(\gamma)$ has no closed form solution in the form of natural expression. We utilize the function $f(\gamma)$ to provide an estimation for the value of $a^*$, given our confidence set $\theta^* \in \mathcal{C}^t$. The core technical challenges lie in securing a sufficiently accurate approximation for $F_\theta^{-1}(\gamma)$ while imposing upper bounds on the error arising from such a function approximation. To address this, we employ a Taylor series expansion to approximate $F_\theta^{-1}(\gamma)$, subsequently applying the Wasserstein approximation theorem for polynomials to confine the approximation error. Following this, we measure this error within the function's domain, exploiting certain inherent properties such as the concavity of the upper-bounding function (refer to details in Lemma \ref{thm:shrink_x_eps} in Appendix). Consequently, the approximation of $\bar{b}^*_a(\theta)$ establishes of a bound on $\Delta_b(\gamma)$, while $\Delta_b(\theta)$ can be bounded using the outcome derived from Theorem \ref{thm:bound-del-ht} (the mathematical details for the proof of Theorem \ref{thm:bound-del-ht} are outlined in Appendix \ref{prf:leader-regret}).


\begin{corollary} \label{cor:alt_br}
     Suppose the follower were to respond with an alternative best-response function $\tilde{\mathfrak{B}}(a)$. As $a^t \to a^*$, $R_A^T(\pi_A)$, where $a^*$ maximizes $a \mathfrak{B}(a)$ and $\tilde{a}^*$ maximizes $a \tilde{\mathfrak{B}}(a)$, the Stackelberg regret $R_A^T(\pi_A)$ from Eq. \eqref{eq:leader-regret} would suffer an additive penalty term of $\hat{\epsilon}T$, where $\hat{\epsilon} = \tilde{a}^* \tilde{\mathfrak{B}}(\tilde{a}^*) - a^* \tilde{\mathfrak{B}}(a^*)$. (Proof in Appendix \ref{prf:alt_br}.)
\end{corollary}

As one of our key assumptions in Section \ref{sec:assumptions} was that the follower acts via risk-free pricing due to the lack of a closed-form solution for the \textit{price-setting Newsvendor} problem, a prudent reader might question what would happen if an alternative approximation for the \textit{price-setting Newsvendor} were devised and applied in this Stackelberg game setting, thereby violating one of the assumptions in Section \ref{sec:assumptions}. Corollary \ref{cor:alt_br} states that, in this case, a penalty term of $\hat{\epsilon}$, which can be computed in advance, would be added to the Stackelberg regret from Theorem \ref{thm:leader-regret}. This notion is relatively intuitive, as the two objective functions are converging to different solutions for the leader, $\tilde{a}^*$ and $a^*$, while maintaining the alternate best response $\tilde{\mathfrak{B}}(a)$. This is an interesting notion that relaxes the constraint of assuming a specific unique best response for the leader. Nevertheless, we stress that this serves only as an extension of the theoretical framework, and all Stackelberg games require some form of assumptions on the structure of the follower's best response.

\subsection{Convergence to Stackelberg Equilibrium} \label{sec:conv-se}

%


\begin{theorem} \label{thm:se-convergence}
    The LNPG algorithm (Algorithm \ref{alg:se-newsv}) converges to an approximate Stackelberg equilibrium, as defined in Section \ref{sec:stack-eq-defn}. (Proof in Appendix \ref{prf:se-convergence}.)
\end{theorem}

\textbf{Sketch of Proof:} Given the regret properties of the leader in Theorem \ref{thm:leader-regret} and follower in Theorem \ref{thm:follower-regret}, as the asymptotic performance of instantaneous regret approaches 0 as $t \xrightarrow[]{} \infty$, the algorithm achieves an approximate Stackelberg equilibrium. 





\section{Empirical Results}

We compare the performance of our proposed algorithm, LNPG, with a well-known and widely used algorithm in the field of multi-armed bandit problems: Upper Confidence Bound (UCB) \cite{szepesvari:2010algorithmsRL}. UCB follows a straightforward exploration-exploitation trade-off strategy by selecting the action with the highest upper confidence bound. It requires discretization over the action space, and we use this as a baseline measure when computing empirical Stackelberg regret. The results are presented in Fig. \ref{fig:stack_regret}, and additional results and configurations are presented in Appendix \ref{sec:experiment-results}.


\begin{figure}[t]
\minipage{0.5\textwidth}
  \includegraphics[width=\linewidth]{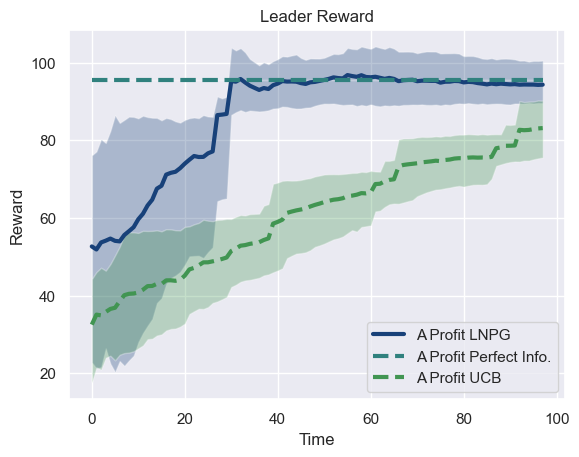}
\endminipage\hfill
\minipage{0.5\textwidth}
  \includegraphics[width=\linewidth]{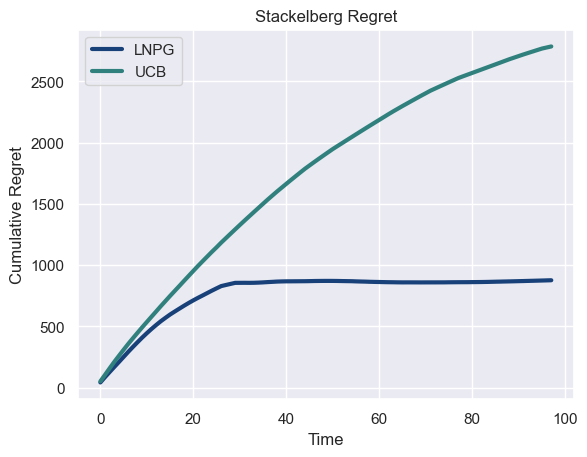}
\endminipage\hfill
\caption{Episodic reward and Stackelberg (leader) regret performance across an Newsvendor pricing game, comparing Algorithm \ref{alg:se-newsv} (LNPG) against UCB. All shaded areas, denoting confidence intervals, are within a quarter quantile. (Parameters used $\kappa=3, \theta_0 = 18, \theta_1 = 7$.)  } \label{fig:stack_regret}
\end{figure}

\section{Conclusion} \label{sec:conclusion}

We demonstrated the effective use of an optimistic online learning algorithm in a dyadic exchange network. We define a Newsvendor pricing game (NPG) as an economic game between two learning agents, representing firms in a supply chain. Both agents are learning the parameters of an additive demand function. The follower, or retailer, is a Newsvendor firm, meaning it must bear the cost of surplus and shortages when subject to stochastic demand. The follower must also determine the optimal price for the product, constituting the classical \textit{price-setting Newsvendor} problem. We establish that, under perfect information, the NPG contains a unique follower best response and that a Stackelberg equilibrium must exist. In the context of online learning, we combine established economic theory of the \textit{price-setting Newsvendor} with that of linear contextual bandits to learn the parameters of reward function. We demonstrate that from the leader's perspective, the this algorithm achieves a Stackelberg regret $\mathcal{O}(\sqrt{T \log(T)})$ for $T$ finite samples, and is  guaranteed to converge to an approximate Stackelberg equilibrium.  One limitation of this work is the structure of the demand function. Currently, the demand function is a linear function of price. To better represent reality, we propose in future work to expand the structure of the demand function to accommodate, in addition to price, multiple features (such as customer demographic features, product features etc.). This implies the extension of the Lemma \ref{thm:theta-2-norm} to p-norm rather than 2-norm, and a restructuring of Theorem \ref{thm:bound-del-ht} to accommodate a more general solution to the maximization problem from Eq. \eqref{eq:h-max-theta}.







\printbibliography

\clearpage


\appendix

\section{Economic Theory}

\subsection{Proof of Theorem \ref{thm:optimal-order-b}} \label{prf:optimal-order-b}

\textbf{Uniqueness of Best Response under Perfect Information:} In the perfect information NPG (from Sec. \ref{sec:newsvendor-games}) with a known demand function, under the assumptions listed in Sec. \ref{sec:assumptions}, the optimal order amount $b^*_a$ which maximizes expected profit of the retailer $\mathbbm{E}[\mathcal{G}_B(\cdot)]$ is $\mathfrak{B}(a)$ and is a surjection from $a \xrightarrow[]{} \mathfrak{B}(a)$, and a bijection when $a < \bar{p}$. Furthermore, within the range $a < \bar{p}$, as $a$ monotonically increases, $b_a^*$ is monotonically decreasing. Where $\bar{p}$ denotes the optimistic estimate of the riskless price $p_0$.

\begin{proof}
    Under the assumptions in Section \ref{sec:assumptions}, given any $a$, there exists a unique maximizing solution to Eq. \eqref{eq:newsvendor-max-int}. From \cite{mills:1959uncertainty} we know that the optimal price to set in the price-setting Newsvendor, with linear additive demand, is of the form,
    
    \begin{align}
        p^*(z) &= p_0 - \frac{\int_z^{\infty} (x - z)f_\theta(x) dx  }{2 \theta_1}
    \end{align}
    
    Where $p_0$ is the \textit{riskless price}, the price which maximizes the linear expected demand. And therefore, the optimal price $p^*$, has the relation $p^* \leq p_0$. $f_\theta(x)$ expresses the probability of demand equal to $x$, under demand parameters $\theta$, and price $p$. This defines the unique relation of $z$ to $p(z)$ from \cite{mills:1959uncertainty} and as evidenced by Eq. \eqref{eq:mills-opt-price} is a unique function. Knowing this relation, it is also logical to see that the riskless price of the retailer must be always above the wholesale price of the product which, as introduced in the constraint in Eq. \eqref{eq:p<pa}.
    
    \begin{align}
         a < p^* &\leq p_0 \xrightarrow[]{} a < p_0 \label{eq:p<pa}
    \end{align}

    Given the unique relation from $z$ to $p(z)$, $p^*$ represents the optimal pricing for ordering $z$ quantity above the expected demand from $\Gamma_\theta(p)$. Having this unique relation, we are ultimately faced with the optimization problem in Eq. \eqref{eq:pz_optim}.
    
    \begin{align}
        \mathcal{G}_B^*(a, p) &= \underset{z \in [-\mu, \infty) }{\mathrm{argmax}} \ \mathbbm{E}[\mathcal{G}_B(\Gamma(p) + z, p^*(z))] \label{eq:pz_optim}
    \end{align}

    $\mathcal{G}_B^*(p)$ is unique under certain conditions outlined in \citet{petruzzi:1999newsv} and \citet{mills:1959uncertainty}, and there exists a unique optimal price $p^*$ for each problem defined by $\theta$. Examining the behaviour of the critical fractile $\gamma^* = (p^* - a)/p^*$, and we see that $\gamma^*$ is monotonically decreasing with respect to $a$. 
    
    \begin{align} 
      \mathfrak{B}(a) = F_\theta^{-1} \Big( \frac{p^* - a}{p^*} \Big), \quad & p^* > a \label{eq:cf-babr}
    \end{align} 
    
    We see from Eq. \eqref{eq:cf-babr} that $\mathfrak{B}(a)$ is unique function of $a$, as $a$ is strictly monotonic increasing, as $b_a^*$ is strictly monotonic decreasing. 
    
    
\end{proof}

\subsection{Proof of Lemma \ref{thm:invariance-of-z}} \label{prf:invariance-of-z}

\textbf{Invariance of $\Psi_\theta(b)$:} Assuming symmetric behaviour of $f_{\theta, p}(\cdot)$ with respect to $\Gamma_\theta(p)$ and constant standard deviation $\sigma$, the term representing supply shortage probability defined as $\Psi_\theta(b)$, is invariant of the estimate of $\theta$. 

\begin{proof}
    First we define the properties of the probability density function of demand $f_\theta(\cdot)$ and $f(\cdot)$. Where $f(\cdot)$ is the $\theta$ invariant version of the demand function, and $f_\theta(\cdot)$ is the probabilistic demand function with respect to the realized value of demand, such that Eq. \eqref{eq:f_demand_relation} holds. $x$ is the difference between realized demand and the expected demand $\mu$, and $z$ is the difference between the order amount $b$ expected demand and the expected demand $\mu$.
    
    \begin{align}
        f(x) &= f_\theta(x + \Gamma_\theta(p)) \label{eq:f_demand_relation} \\
        \Gamma_\theta(p) &= \mathbbm{E}[d_\theta(p)] \\
        z &= b - \Gamma_\theta(p)
    \end{align}

    For the \textit{price-setting Newsvendor}, the optimal price is the is the solution to the equation denoting the expected profit, given retail price $a$, $\mathbbm{E}[\mathcal{G}_B(p, b| a)]$, is expressed in Eq. \eqref{eq:newsvendor-max-int}.

    \begin{align}
        p^* &= \underset{ (p, b) \in \mathcal{P} \cross \mathcal{B} } {\mathrm{argmax}} \ \mathbbm{E}[\mathcal{G}_B(p, b| a)]
    \end{align}

    From \citet{mills:1959uncertainty} and \citet{petruzzi:1999newsv} we know that under perfect information, the optimal price for the price setting Newsvendor $p^*$ is always the risk-free price $p_0$ subtracted by a term proportional to $\Psi_\theta(b)$ for expected inventory loss. It follows that $\mathbbm{E}[\mathcal{G}_B(p, b| a)]$ is concave in $p$ given $b$, and conversely $b$ is also concave given $p$. Therefore, there exists a unique solution to $p^*$ given $b$, as illustrated in Eq. \eqref{eq:mills-opt-price}. 
    
    \begin{align}
        p^* &= p_0 - \frac{\Psi_\theta(b)}{2 \theta_1} \\
        \text{where,} \quad p_0 &= \frac{\theta_0 + \theta_1a}{2}
    \end{align}

    Given this relation we can then maximize over a subspace of $\mathcal{P} \cross \mathcal{B}$ more efficiently.

    \begin{align}
        p^* = p(b^*) &= \underset{ b \in  \mathcal{B} } {\mathrm{argmax}} \ \mathbbm{E}[\mathcal{G}_B(p(b), b| a)]
    \end{align}
    
    We wish to show that $\Psi(z) = \Psi_\theta(b)$ remains the same despite variation in $\theta$.

    \begin{align}
        \Psi_\theta(b) &= \Psi(z), \quad \forall \theta \\
        \text{where,} \quad \Psi_\theta(b) &= \int_b^{\infty} (x - b)f_{\theta, p}(x) dx \label{eq:psi_integ_b} \\
        \quad \Psi(z) &= \int_z^{\infty} (x - z)f(x - \Gamma_\theta(p)) dx \label{eq:psi_integ_z}
    \end{align}

    We show the equivalance of $\Psi_\theta(z)$ and $\Psi(z)$ via equivalence of the integral definitions from Eq. \eqref{eq:psi_integ_b} and Eq. \eqref{eq:psi_integ_z} respectively, utilizing the relation from Eq. \eqref{eq:f_demand_relation}. For convience, let $\mu_p = \Gamma_\theta(p)$.
    
    \begin{align}
        \int_z^{\infty} (x - z)f(x) dx &= \int_z^{\infty} (x - z)f_{\theta, p}(x + \mu_p) dx \\
        &= \int_{b - \mu_p}^{\infty} \Big( x - (b - \mu_p) \Big) f_{\theta, p}(x + \mu_p, p) dx \\
        &= \int_{b}^{\infty}\Big((x - \mu_p) - (b - \mu_p) \Big) f_{\theta, p}(x + \mu_p - \mu_p) dx \\
        &= \int_{b}^{\infty} (x - b)f_{\theta, p}(x) dx
    \end{align}

    Because $\theta$ can only vary by updating samples with respect to $t$, it follows that $\Psi(z)$ is also invariant with respect to $t$.

\end{proof}

\subsection{Theorem \ref{thm:fgamma-decreasing}} \label{prf:fgamma-decreasing} \label{thm:fgamma-decreasing}

\textbf{Theorem \ref{thm:fgamma-decreasing}. Strictly decreasing $\gamma^*$:} Given the condition $a < p^* \leq p_0$, $\gamma^*$ is strictly decreasing with respect to $a$. 

\begin{proof}
    Sketch of idea, as $a$ increases, the boundary between $p_0 - \bar{p}$ decreases. 
    
    This boundary $\bar{p} - p_0 = p_0 - a$ due to symmetry. 
    
    \begin{align}
        p_0 &= \frac{\bar{p}+a}{2} \\
        \bar{p} - p_0 &= \frac{2\bar{p}}{2} - \frac{\bar{p}+a}{2} = \frac{\bar{p} - a}{2} = \frac{2p_0 - a - a}{2} \\
        & = p_0 - a
    \end{align}
    
    Therefore given $p^* \in (a, p_0]$, we know that $|(a, p_0]|$ is decreasing w.r.t increasing $a$, holding $\bar{p}$ constant.
    
    Given $p^* \geq p_0$, and $\gamma = (p - a)/p$,
    
    \begin{align}
        p^* \leq p_0 \xrightarrow[]{} (p^* - a)/p^* \leq (p_0 - a)/p_0
    \end{align}

    The behaviour of $\gamma$, where $\gamma^* = (p^* - a)/p^*$ will also be decreasing with increasing $a$.
    
\end{proof}

\subsection{Proof of Theorem \ref{thm:unique-se}} \label{prf:unique-se}

\textbf{Theorem \ref{thm:unique-se}. Unique Stackelberg Equilibrium:} A unique pure strategy Stackelberg Equilibrium exists in the full-information Newsvendor pricing game.

\begin{proof}
    Let us define an arbitrary probability density function $f_\theta(x) >0 $ which operates on the of $x \in [0, \bar{\bar{p}}]$, and is strictly positive. The cumulative probability function function is therefore,

    \begin{align}
        F_\theta(x) = \int_0^x f_\theta(x) dx = \gamma \label{eq:arbitrary-cdf}
    \end{align}
    
    Therefore, from Eq. \eqref{eq:arbitrary-cdf}, its cumulative distribution, defined as the integral of such a probability distribution function $F_\theta(x)$ is strictly monotonically increasing from $x \in [0, \bar{\bar{p}}]$. Therefore the inverse of $F_\theta(x)$, defined as $F_\theta^{-1}$, is an injection from $\gamma \xrightarrow[]{} x$. We also that the critical fractile $\gamma^* = (p^* - a)/p^*$ is strictly decreasing with respect to $a$. Therefore, from the follower's perspective, since $F_\theta^{-1}(\gamma^*)$ is an injection from $\gamma^* \xrightarrow[]{} \mathfrak{B}(a)$, there is always a best response function defined by $\mathfrak{B}(a)$ given $a$ and is an injection. As $\mathfrak{B}(a)$ is unique given $a$, therefore the Eq. \eqref{eq:leader-max} has a unique maximum.
    
    \begin{align}
        \mathcal{G}_A^*(a) \leq  \underset{a \in \mathcal{A} }{\mathrm{max}} \ a F_{\theta}^{-1} \Big( \frac{p^* - a}{p^*} \Big) \label{eq:leader-max}
    \end{align}
    
\end{proof}

\subsection{Proof of Lemma \ref{thm:opt-p-star-decrease}} \label{prf:opt-p-star-decrease}
    \textbf{Lemma \ref{thm:opt-p-star-decrease}. Convergence  of $\bar{p}^*$ to $p^*$:}  Suppose the solution to Eq. \eqref{eq:mills-opt-price} is computable. For any symmetric demand distribution, the optimistic estimate of the optimal market price from Eq. \eqref{eq:mills-opt-price}, denoted as $\bar{p}^*$ which is the solution to $p^*$ under optimistic parameters $\bar{\theta}$, is decreasing with respect to $t$ such that $\bar{p}^*(\theta_{t+1}) \leq \bar{p}^*(\theta_t)$, where the  difference $\bar{p}^*(\theta_t) - p^*$ is bounded by $\mathcal{O}(\log(T)/T)$.

\begin{proof}
    Suppose the solution to Eq. \eqref{eq:mills-opt-price} is computable.
    We examine the equation for the optimal price function, from  Eq. \eqref{eq:mills-opt-price}, and we make the assumption that the term $\Psi(z)$ is computable,
    
    \begin{align}
        p_\theta^*(z) &= p_0(\theta) - \frac{\Psi(z)}{2 \theta_1} \label{eq:optimistic-p*-theta}
    \end{align}

    Due to Lemma \ref{thm:invariance-of-z} we know that $\Psi(z)$ is invariant of $\theta$, for symmetric probability distributions with constant variance. Thus using this advantage, we know from Eq. \eqref{eq:bigO-theta-del-decrease} that the rate of convergence to $\theta^*$ is $\mathcal{O}(\log(t)/t)$. We then look at the behaviour of $\theta_1$, from Eq. \eqref{eq:theta_constr}, we can obtain an optimistic upper bound on $\theta_1$ denoted as $\bar{\theta}_1$ as $t \xrightarrow[]{} \infty$, when given the confidence ball $\mathcal{C}^t$, by setting $\hat{\theta}_0 = \theta_0^*$, this allows for the slack, denoted as, $\kappa \sqrt{\log(t)/t}$ to be completely maximized for maximizing $\bar{\theta}_1^t$. In other words, we assume that the estimate of $\hat{\theta}_0^t$ is exact. Subsequently, we have a 1-dimensional optimization problem.
    
    \begin{align}
        \Delta_{\theta_1} = |\hat{\theta}_1 - \theta_1^*| &\leq \kappa \sqrt{\log(t)/t}
    \end{align}

    Let $\Delta_{\theta_1} = \hat{\theta}_1 - \theta_1^*$, we maximize $\Delta_{\theta_1}$ by setting $\Delta_{\theta_1} = \kappa \sqrt{\log(t)/t})$, and this maximum value is shrinking at a rate of $\mathcal{O}(\log(t)/t$. We examine the asymptotic behaviour of $p^*_\theta(z)$ as $t \xrightarrow[]{} \infty$.

    \begin{align}
        \lim_{t \to \infty} p_\theta^*(z) = \lim_{t \to \infty} p_0(\theta) - \frac{\Psi(z)}{2 (\theta_1^* \pm \Delta_{\theta_1}) } \label{eq:optimistic-p*-theta}
    \end{align}

    Examining Eq. \eqref{eq:optimistic-p*-theta} we know that ultimately, $\bar{p}_0(t) \xrightarrow[]{} p_0$, and $\Delta_{\theta_1} \xrightarrow[]{} 0$ as $t \xrightarrow[]{} \infty$. And thus, $\bar{p}^* \xrightarrow[]{} p^*$ asymptotically as $t \xrightarrow[]{} \infty$.
    






    
\end{proof}


\section{Online Learning Theory}

\subsection{Optimism Under Uncertainty (OFUL) Algorithm} \label{sec:alg-abbasi}


\begin{algorithm}[h!] 
\caption{OFUL Algorithm (from \citet{abbasi:2011ofu}) }\label{algo:oful-abbasi}
\begin{algorithmic}[1]
    \For {$t \in 1 ... T$}:
    \State $(X_t, \hat{\theta}_t) = \underset{ (x,\theta) \in \mathcal{X} \times \mathcal{C}^t_t} {\mathrm{argmax}} \langle x, \theta \rangle$
    \State Agent plays $X^t$ and observes reward $Y^t$
    \State Update $\mathcal{C}^t_t$
    \EndFor
\end{algorithmic}
\end{algorithm}

\subsection{Optimization Under Uncertainty (OFUL)} \label{sec:oful-dicuss}

The OFUL algorithm presented in Appendix \ref{sec:alg-abbasi} is a linear bandit optimization function. It works by constructing a confidence ball $\mathcal{C}^t_t$ which provides a bound on the variation of estimated parameters $\theta$. This confidence bound shrinks as more observations are obtained. When running this algorithm, the error of estimated parameters $\norm{\theta^*_t - \hat{\theta}_t}_{\bar{V}_t}$ is bounded by the \textit{self-normalizing norm}, $\norm{\cdot}_{\bar{V}_t}$, where by definition,

\begin{align}
    \bar{V}_t = \lambda I + \sum_{i=1}^t X_iX_i^T, \quad \norm{S}_{\bar{V}_t} = \sqrt{S^T \bar{V}_t S} \label{eq:vt-defn}
\end{align}

Where $I$ is the identity matrix, for some $\lambda > 0$. With these definitions in Eq. \eqref{eq:vt-defn}, the OFUL algorithm, with respect to the shrinking confidence ball $\mathcal{C}^t_t$ admits a bound of the confidence of the parameters as,

\begin{align}
    \norm{\theta^*_t - \hat{\theta}_t}_{\bar{V}_t} \leq R \sqrt{d \log \Big( \frac{1 + tL^2/\lambda}{\delta} \Big)} + \lambda^{1/2}J \in \mathcal{O}(\sqrt{\log{t}})
\end{align}


Where the random variables $\eta_t$, and filtration $\{F_t\}_{t=0}^\infty$ are $R$-sub Gaussian, for some $R > 0$, such that,

\begin{align}
    \mathbbm{E}[e^{\lambda \eta_t} | F_{t-1}] \leq \exp \Big( \frac{\lambda^2 R^2}{2} \Big), \quad \forall \lambda \in \mathbbm{R}
\end{align}

Where $d$ is the dimension of the parameters. $L$ is a condition for the random variables where $\norm{X_t}_2 \leq L, \forall t \geq 1$. And $J$ is a constraint on the parameters such that $\norm{\theta^*}_2 \leq J$.

\subsection{Proof of Lemma \ref{thm:theta-2-norm}} \label{prf:theta-2-norm}

\textbf{OFUL - Bound on 2-Norm:} With probability $1-\delta$, and $t \geq e$, the constraint $||\theta - \theta^*||_{\bar{V}_t} \leq \mathcal{O}(\sqrt{\log(t)})$ from \citet{abbasi:2011ofu} implies the finite $2$-norm bound $||\theta - \theta^*||_2 \leq \mathcal{O}(\sqrt{\log(t)/ t})$, and also thereby the existence of $\kappa > 0$ such that $||\theta - \theta^*||_2 \leq \kappa \sqrt{\log(t)/ t}$ (Proof in Appendix \ref{prf:theta-2-norm}.)    

\begin{proof} 
    We assume $\bar{V}_t$ is full rank, and that sufficient sufficient exploration has occurred. Let $\mathbf{\Delta}_\theta = \theta^* - \hat{\theta}_t$. From \citet{abbasi:2011ofu}, we define $\bar{V}_t$ as, 

    \begin{align}
        \bar{V}_t = V + \sum_{t=1}^T X^TX \label{eq:v_bar_t_def}
    \end{align}
    
    Let $V$ be a positive definite real-valued matrix defined on $\mathbb{R}^d$, represented as a diagonal matrix with entries $\lambda$, where $\lambda \in \mathbb{R}^+$. Moreover, $X$ is the vector corresponding to the arm selected (specifically in our case, this corresponds to the pricing action of the follower). Furthermore, we define the \textit{self-normalizing norm} of $\mathbf{\Delta}_\theta$ as,

    \begin{align}
       \norm{\mathbf{\Delta}_\theta}_{\bar{V}_t} = \sqrt{\mathbf{\Delta}_\theta^T \bar{V}_t \mathbf{\Delta}_\theta} 
    \end{align}

    Fundamentally, we wish to demonstrate that given $\kappa$ there exists $\kappa_2$ such that,

    \begin{align}
        ||\theta - \theta^*||_{\bar{V}_t} &\leq \kappa \sqrt{\log(t)} \implies ||\theta - \theta^*||_2 \leq \kappa_2 \sqrt{\log(t)/t}\\
        \forall t &\in \{ \mathbbm{Z}| 1 \leq t \leq T \}, \quad \forall \kappa_1 \in \mathbb{R}^+, \quad \forall \kappa_2 \in \mathbb{R}^+
    \end{align}

    We note that $||\theta - \theta^*||_{\bar{V}_t} \leq \kappa \sqrt{\log(t)}$ with non-decreasing eivengalues of $\bar{V}_t$ with respect to $t$, this generally denotes a volume decreasing ellipsoid \cite{lattimore:2020-bandit}. Suppose there exists a scalar $\kappa_2 \in (0, \infty)$, where we impose that the minimum eigenvalue of $\bar{V}_t$ scales as $\kappa_2 t$ for all $t \in [1, T]$. This would result in the inequality,
    
    \begin{align}
        \kappa_2 \sqrt{t} ||\mathbf{\Delta}_\theta||_2 \leq \sqrt{\mathbf{\Delta}_\theta^T \bar{V}_t  \mathbf{\Delta}_\theta} = \norm{\mathbf{\Delta}_\theta}_{\bar{V}_t} \leq \sqrt{\log(t)} \label{eq:eta_ineq}
    \end{align}

    
    The inequality in Eq. \eqref{eq:eta_ineq} holds independent of $t$, as $||\mathbf{\Delta}_\theta||_2$ is constant and $||\mathbf{\Delta}_\theta||_{\bar{V}_t }$ is non-decreasing with respect to time. Therefore, the existence of a scalar constant $\kappa_2$ implies that,
    
    \begin{align}
        \kappa_2\sqrt{t} \norm{\theta_t^* - \hat{\theta}_t}_2  &\leq \norm{\theta_t^* - \hat{\theta}_t}_{\bar{V}_t} \label{eq:kappa-2-norm-bound}
    \end{align}

    Next suppose we select some $\kappa_v$, such that $\kappa_v$ obeys the constraint in Eq. \ref{eq:kv_bound}, we show the existence of $\kappa_v > 0$ next.

    \begin{align}
         \norm{\theta_t^* - \hat{\theta}_t}_{\bar{V}_t} \leq R \sqrt{d \log \Big( \frac{1 + tL^2/\lambda}{\delta} \Big)} + \lambda^{1/2}S \leq  \kappa_v \sqrt{\log(t)} \label{eq:kv_bound}
    \end{align}


    \textbf{Existence of $\kappa_v$:} We wish to show the existence of $\kappa_v$, such that,

    \begin{align}
         R \sqrt{d \log \Big( \frac{1 + tL^2/\lambda}{\delta} \Big)} + \lambda^{1/2}S \leq  \kappa_v \sqrt{\log(t)}
    \end{align}

    Let us define constants,

    \begin{align}
         c_1 \sqrt{\log (c_2 + c_3 t) } + \lambda^{1/2} S &\leq  \kappa_v \sqrt{\log(t)} \\
         \text{where} \quad c_1 = R\sqrt{d}, \, c_2 = 1/\delta, \, c_3 &= \frac{L^2}{\lambda \delta}
    \end{align}
    
    We first look at the relation where,

    \begin{align}
         c_1 \sqrt{\log (c_2 + c_3 t) } &\leq  \kappa_v \sqrt{\log(t)} \label{eq:kv_multiply}
    \end{align}

    We wish to find an value for $\kappa_v$ such that Eq. \eqref{eq:kv_multiply} holds. Let $c' = c_2c_3$, so for $t \geq 1$, we have,

    \begin{align}
         c_1 \sqrt{\log (c_2 + c_3 t)} &\leq c_1 \sqrt{\log (c' t) } \\
         &=  c_1 \sqrt{\log (c') + \log (t) } 
    \end{align}

    Suppose $\kappa_v$ can be arbitrarily large, and we want to infer the existence of a $\kappa_v$ such that Eq. \eqref{eq:kv_multiply} holds.

    \begin{align}
         c_1 \sqrt{\log (c') + \log (t) } &\leq \kappa_v \sqrt{\log(t)} \\ 
         c_1^2 (\log (c') + \log (t)) &\leq \kappa_v^2 \log(t) \\
         c_1 \sqrt{\log (c' - t) + 1} &\leq \kappa_v \label{eq:ineq-kv-1}
    \end{align}

    Therefore we select $\kappa_v$ as,

    \begin{align}
         \kappa_v &= c_1 \sqrt{\log (c') + 1} \\
         &= R\sqrt{d} \sqrt{\log (\frac{L^2}{\lambda \delta^2}) + 1}
    \end{align}

    Such that Eq. \eqref{eq:ineq-kv-1} holds, as long as $t \geq 1$. The second condition is that,

    \begin{align}
         \kappa_v \sqrt{\log(t)} \geq \lambda^{1/2}S \label{eq:ineq-kv-2}
    \end{align}

    For this we can argue that so long as $\sqrt{\log(t)} \geq 1 \xrightarrow[]{} t \geq e$, can choose $\kappa_v = \lambda^{1/2}S$, such that Eq. \eqref{eq:ineq-kv-2} holds. Therefore, we can select $\kappa_v$ as,

    \begin{align}
        \kappa_v = \max \Big\{ \lambda^{1/2}S,\ R\sqrt{d} \sqrt{\log (\frac{L^2}{\lambda \delta^2}) + 1} \Big\} \label{eq:kv-condition}
    \end{align}
    
    Which holds when $t \geq e$, allowing for an upper bound for $\kappa_v$ in Eq. \eqref{eq:kv-condition}. We combine our bound with Eq. \eqref{eq:kappa-2-norm-bound}, where

    \begin{align}
        \kappa_2 \sqrt{t}  \norm{\theta_t^* - \hat{\theta}_t}_2  &\leq \norm{\theta_t^* - \hat{\theta}_t}_{\bar{V}_t} \leq \kappa_v \sqrt{\log(t)} \\
         \norm{\theta_t^* - \hat{\theta}_t}_2 &\leq \kappa \sqrt{\log(t)/t}
    \end{align}
    
    Where $\kappa = \kappa_v / \kappa_2$.

    

    
    

    

    
    
\end{proof}




\subsection{Proof of Theorem \ref{thm:bound-del-ht}} \label{prf:bound-del-ht}

\textbf{Bounding the Optimistic Follower Action:} Let $\Delta_\mathcal{H}(t)$ represent the deviation in pricing action of the follower (retailer) at time $t$. Then $\Delta_\mathcal{H}(t) \leq  \kappa_H \sqrt{ \log(t)/t } $ for some $\kappa_H > 0$ and $t \geq e$.

\begin{proof}
    We connect the bound on estimated parameters $\norm{\hat{\theta} - \theta^*}_2 \leq \mathcal{O}(\sqrt{\log(t)/t})$ with bounds on the optimal expected follower rewards $\mathcal{G}(p_a)^*_a$. Under the optimism principle, we therefore solve the optimization problem to maximize $\mathcal{H}(\theta) = p_0(\theta)$ defined in Eq. \eqref{eq:h-max-theta-app} and \eqref{eq:theta_constr-app}. In this optimization problem, we optimize for the optimistic estimate of $p_0$ by maximizing $\theta_0/\theta_1$. 
    
    By Lemma \ref{thm:theta-2-norm}, we have an exact bound on the inequality Eq. \eqref{eq:theta_constr-app}. Therefore we can solve for the maximization problem for the value $\mathcal{H}(\theta)$ to give the optimistic estimate of the upper bound on follower reward. With probability $1- \delta$,
    
    
    \begin{subequations}
        \begin{alignat}{2}
            &\!\max_{\theta}   &\qquad \qquad \mathcal{H}(\theta) &= \ \frac{\hat{\theta}_0}{\hat{\theta}_1} \label{eq:h-max-theta-app}  \\
            &\text{subject to} &\qquad       \sqrt{ (\hat{\theta}_0 - \theta_0^*)^2 + (\hat{\theta}_1 - \theta_1^*)^2 } &\leq \kappa \sqrt{\log(t)} / \sqrt{t}  \label{eq:theta_constr-app}
        \end{alignat}
    \end{subequations}
    
    Where $\theta_0$ and $\theta_1$ can describe a hyperplane, (a line in the 2D case), and the confidence ball constraint in Eq. \eqref{eq:theta_constr-app} can denote the parameter constraints. The optimization problem can be visualized in 2D space by maximizing the distance a from point to line, and solving for the parameters which maximizes the slope $\theta_0/\theta_1$, as illustrated in Fig. \ref{fig:conf-ball}.
    
    \begin{align}
        \mathcal{H}^* \hat{\theta}_1 &= \hat{\theta}_0  \\
        \frac{| \mathcal{H}^* \theta_1 - \theta_0|}{\sqrt{ {\mathcal{H}^*}^2+1}} &= \kappa\sqrt{\log(t)}/\sqrt{t}  
    \end{align}
    
    Solving for Eq. \eqref{eq:point-to-line} for $\mathcal{H}^*$ gives us a closed form expression for $\mathcal{H}^*(\theta)$ in terms of $\theta$ as expressed in Eq. \eqref{eq:hstar-def}.
    
    \begin{align}
        \mathcal{H}^*(\theta) &= \frac{  \hat{\theta}_0 \hat{\theta}_1 \pm \sqrt{ (\kappa\sqrt{\log(t)}/\sqrt{t})^2 (\hat{\theta}_0^2 + \hat{\theta}_1^2) - (\kappa\sqrt{\log(t)}/\sqrt{t})^4 }  }{ \hat{\theta}_1^2 - (\kappa\sqrt{\log(t)}/\sqrt{t})^2} \label{eq:hstar-def}\\
        &= \frac{  \hat{\theta}_0 \hat{\theta}_1 \pm \sqrt{ (\hat{\theta}_0^2 + \hat{\theta}_1^2) \kappa^2\log(t)/t  - \kappa^4 \log^2(t)/t^2 }  }{\hat{\theta}_1^2 -  \kappa^2\log(t)/t}  \\
        &< \frac{  \hat{\theta}_0 \hat{\theta}_1 + \sqrt{ (\hat{\theta}_0^2 + \hat{\theta}_1^2) \kappa^2\log(t)/t} }{\hat{\theta}_1^2 -  \kappa^2\log(t)/t} \label{eq:hstar-def-bound}
    \end{align}


    
    When we obtain extreme points for $\mathcal{H}(\theta)$ we obtain the bounds for $\mathcal{G}_A(p_a)^* \in [\underline{\mathcal{G}}_A(p_a, a), \ \bar{\mathcal{G}}_A(p_a, a)]$. The extremum of $\mathcal{H}(\theta)$ always occurs at the tangent of the confidence ball constraints, where the slope of the line denoted by $\theta_0/\theta_1$ is maximized or minimized. We note $\mathcal{H}(\theta)$ is convex in both $\theta_0$ and $\theta_1$.
    
    As $\mathcal{H}^*(\theta)$ denotes the optimistic estimate of $p_0$, we see from the bounded inequality expressed in Eq. \eqref{eq:hstar-def-bound} that $\mathcal{H}^*(\theta)$ is a decreasing function as the estimate of $\theta$ improves. Therefore, the critical fractile $\gamma$ is also decreasing, with respect to $t$, where $\bar{p}_0 = \mathcal{H}^*(\theta)/2$, with reference to Eq. \eqref{eq:cf-p-relation}.

    \begin{align}
        \gamma = \frac{p_0 - a}{p_0} &=  1 - \frac{ 2 a }{ \mathcal{H}(\theta) + a} \label{eq:cf-p-relation}
    \end{align}
    
    The learning algorithms converges to a gap, we define this gap as, at large values of $t$,
    
    \begin{align}
        \Delta_\mathcal{H}(t) &= \frac{  \hat{\theta}_0 \hat{\theta}_1 \pm \sqrt{ (\hat{\theta}_0^2 + \hat{\theta}_1^2) \kappa^2\log(t)/t  - \kappa^4 \log^2(t)/t^2 }  }{\hat{\theta}_1^2 - \kappa^2\log(t)/t} - \frac{\hat{\theta}_0}{\hat{\theta}_1}
    \end{align}
    
    As the $t \xrightarrow[]{} T$ for sufficiently large values of $T$, $\Delta_\mathcal{H}(t)$ approaches 0, as $\log(t)/t$ approaches 0.
        
    \begin{align}
        \Delta_\mathcal{H}(t) &= \frac{ \hat{\theta}_1 \Big( \hat{\theta}_0 \hat{\theta}_1 + \sqrt{ (\hat{\theta}_0^2 + \hat{\theta}_1^2 ) \kappa^2\log(t)/t } \Big) - \hat{\theta}_0 (\hat{\theta}_1^2 - \kappa^2\log(t)/t) }{ \hat{\theta}_1 (\hat{\theta}_1^2 - \kappa^2\log(t)/t)} \\
        &= \frac{ \hat{\theta}_1 \hat{\theta}_0 \hat{\theta}_1 + \hat{\theta}_1 \sqrt{ (\hat{\theta}_0^2 + \hat{\theta}_1^2 ) \kappa^2\log(t)/t } - \hat{\theta}_0 \hat{\theta}_1^2  + \hat{\theta}_0 \kappa^2\log(t)/t }{ \hat{\theta}_1 (\hat{\theta}_1^2 - \kappa^2\log(t)/t)} \\
        &= \frac{ \hat{\theta}_1^2 \hat{\theta}_0 - \hat{\theta}_0 \hat{\theta}_1^2 + \hat{\theta}_1 \sqrt{ (\hat{\theta}_0^2 + \hat{\theta}_1^2 ) \kappa^2\log(t)/t } + \hat{\theta}_0 \kappa^2\log(t)/t }{ \hat{\theta}_1 (\hat{\theta}_1^2 - \kappa^2\log(t)/t)} \\
        &= \frac{ \hat{\theta}_1 \sqrt{ (\hat{\theta}_0^2 + \hat{\theta}_1^2 ) \kappa^2\log(t)/t } + \hat{\theta}_0 \kappa^2\log(t)/t }{ \hat{\theta}_1 \hat{\theta}_1^2 - \hat{\theta}_1 \kappa^2\log(t)/t} \\
        &< \frac{ \hat{\theta}_1 \sqrt{ (\hat{\theta}_0^2 + \hat{\theta}_1^2 ) \kappa^2\log(t)/t } + \hat{\theta}_0 \kappa^2\log(t)/t }{ \hat{\theta}_1 \hat{\theta}_1^2} 
    \end{align}

    We have decreasing behaviour, at sufficiently large values of $t$, such that $\sqrt{\log(t)/t} \gg \log(t)/t$, which is not difficult to obtain, as the function $\hat{\theta}_1 \hat{\theta}_1^2 - \hat{\theta}_1 \kappa^2\log(t)/t \leq \hat{\theta}_1 \hat{\theta}_1^2$, is always decreasing for $t > e$, and upper-bounded by $\hat{\theta}_1 \hat{\theta}_1^2$ (See Appendix \ref{sec:show-t-h-bound}). The upper-bound on $\Delta_\mathcal{H}(t)$ is denoted in Eq. \eqref{eq:upper-bound-del-h}.
    
    \begin{align}
       \Delta_\mathcal{H}(t) &\leq \frac{ \hat{\theta}_1 \sqrt{ (\hat{\theta}_0^2 + \hat{\theta}_1^2 ) \kappa^2\log(t)/t } + \hat{\theta}_0 \kappa^2\log(t)/t }{ \hat{\theta}_1 \hat{\theta}_1^2} \label{eq:upper-bound-del-h} 
    \end{align}

    Essentially Eq. \eqref{eq:upper-bound-del-h} is a linear combination of $\sqrt{\log(t)/t}$ and $\log(t)/t$, where $\sqrt{\log(t)/t} \geq \log(t)/t$ for $t \geq 1 $. We denote two constants $v_1$ and $v_2$ such that  $v_1\sqrt{\log(t)/t} + v_2\log(t)/t$. Therefore,

    \begin{align}
        v_1\sqrt{\log(t)/t} + v_2\log(t)/t \leq (v_1 + v_2)\sqrt{\log(t)/t} \label{eq:v12-ineq}
    \end{align}

    Where,

    \begin{align}
        v_1 = \frac{\hat{\theta}_1 \sqrt{ (\hat{\theta}_0^2 + \hat{\theta}_1^2 ) \kappa^2}}{\hat{\theta}_1 \hat{\theta}_1^2}, \quad v_2 = \frac{\hat{\theta}_0 \kappa^2}{\hat{\theta}_1 \hat{\theta}_1^2} 
    \end{align}
    
    We can see that from the inequality Eq. \eqref{eq:v12-ineq}, that for some $\kappa_H = v_1 + v_2$, then,
    
    \begin{align}
       \Delta_\mathcal{H}(t)  &\leq  \kappa_H \sqrt{ \log(t)/t } \label{eq:del_h_raw_logt_t}
    \end{align}
    
    Therefore, we can deduce that the decreasing order of magnitude of $\Delta_\mathcal{H}(t)$ is $\mathcal{O}(\sqrt{\log(t)/t})$ for sufficiently large values of $t \geq e$, where the bound in Eq. \eqref{eq:upper-bound-del-h}. 

    \begin{align}
        \Delta_\mathcal{H}(t) &\in \mathcal{O}(\sqrt{\log(t)/t}), \quad \forall t \geq e \label{eq:bigO-theta-del-decrease}
    \end{align}
    
    $\Delta_\mathcal{H}(t)$ describes the behaviour of the convergence of $\mathcal{H}^*(\theta)$ towards the true $p_0$ with respect to $t$, and this decreases with $\mathcal{O}(\sqrt{\log(t)/t})$. 

\end{proof}

\subsection{Proof of Lemma \ref{thm:cum-del-HT}} \label{prf:cum-del-HT}

\textbf{Bounding the Cumulation of $\Delta_\mathcal{H}(t)$}: Suppose we wish to measure the cumulative value of $\Delta_\mathcal{H}(t)$ from $1$ to $T$, denoted as $\Delta_\mathcal{H}(T)$. Then  $\Delta_\mathcal{H}(T) \leq 2 \kappa_H  \sqrt{T \log(T)} $.

\begin{proof}
    To quantify $\sum^T_{t=\bar{t}} \Delta_\mathcal{H}(t)$ we take the integral of $\Delta_\mathcal{H}(t)$ noting that the integral of $\Delta_\mathcal{H}(t)$ from $1$ to $T$ will always be greater than the discrete sum so longs as $\Delta_\mathcal{H}(t) > 0$.

    \begin{align}
         \Delta_\mathcal{H}(T) &< \sum^T_{t=\bar{t}} \Delta_\mathcal{H}(t)  \\
        &< \kappa_H \int_1^T \sqrt{\frac{\log(t)}{t}} dt \\
        &= \kappa_H \Big( 2 \sqrt{t \log(t)} - \sqrt{2 \pi} \erf( \sqrt{\log(t) / 2}  )  \Big|_1^T \Big) \\
        &= \kappa_H \Big( 2 \sqrt{T \log(T)} - \sqrt{2 \pi} \erf( \sqrt{\log(T) / 2}  - 2 \sqrt{1 \log(1)} - \sqrt{2 \pi} \erf( \sqrt{\log(1) / 2} )  \Big) \\
        &= \kappa_H \Big( 2 \sqrt{T \log(T)} - \sqrt{2 \pi} \erf( \sqrt{\log(T) / 2})  \Big) \\
        &\leq 2 \kappa_H  \sqrt{T \log(T)} \label{eq:h-decrease-rate}
    \end{align}
\end{proof}

From Lemma \ref{thm:cum-del-HT} we can see that the optimistic pricing deviation of the follower is,

\begin{align}
    \Delta_\mathcal{H}(T) &\in \mathcal{O}(\sqrt{T \log(T)})
\end{align}

\subsection{Proof of Lemma \ref{thm:demand-optimistic-reduction}} \label{prf:demand-optimistic-reduction}

\begin{lemma} \label{thm:demand-optimistic-reduction}
    \textbf{Demand estimate reduction:} Given any price $p$, the maximum optimistic demand estimate is $\bar{d}_\theta(p) \in \mathcal{O}(\sqrt{\log(t)/t})$, and $|\theta_0^t - \theta_0^*| \in \mathcal{O}(\sqrt{\log(t)/t})$. (Proof in Appendix \ref{prf:demand-optimistic-reduction}.)
\end{lemma}

\begin{proof}
    Suppose we are interested in only the confidence on $\theta_0$, then Eq. \eqref{eq:theta_constr} simplifies to Eq. \eqref{eq:theta0-opt-bound}.

    \begin{align}
        (\hat{\theta}_0 - \theta_0^*)^2 &\leq \kappa^2 \log(t) / t \label{eq:theta0-opt-bound} \\
        \hat{\theta}_0 - \theta_0^* &\leq \kappa \sqrt{\log(t) / t} \label{eq:theta0-opt-bound-sq}
    \end{align}

    Suppose in Eq. \eqref{eq:theta0-opt-bound-sq} we are setting $\theta^*_1 = \hat{\theta}_1$, thereby maximizing the possible demand. Effectively this is the maximum possible demand (when $p=0$) and shrinks at a rate of $\kappa \sqrt{\log(t) / t}$.

\end{proof}

\subsection{Notes on the Bound on $\Delta_\mathcal{H}(t)$} \label{sec:show-t-h-bound}

We want to show that $\hat{\theta}_1 \hat{\theta}_1^2 - \hat{\theta}_1 \kappa^2\log(t)/t \leq \hat{\theta}_1 \hat{\theta}_1^2$ for $t \geq e$.

\begin{align*}
    \frac{d}{dt}\left(\frac{\kappa^2\log t}{t}\right) &= \frac{v\frac{du}{dt} - u\frac{dv}{dt}}{v^2} \\
    &= \frac{t \cdot \frac{d}{dt}(\kappa^2\log t) - \kappa^2\log t \cdot \frac{d}{dt}(t)}{t^2} \\
    &= \frac{t \cdot \kappa^2\frac{1}{t} - \kappa^2\log t \cdot 1}{t^2} \\
    &= \frac{\kappa^2 - \kappa^2\log t}{t^2} \\
    &= \frac{\kappa^2(1 - \log t)}{t^2}
\end{align*}

Therefore,

\begin{align}
    \frac{d}{dt}\left(\frac{\kappa^2\log t}{t}\right) = 0 \xrightarrow[]{} t = e
\end{align}


\section{Regret Bounds}

\subsection{Proof for Lemma \ref{thm:shrink_x_eps}} 

\begin{lemma} \label{thm:shrink_x_eps} \textbf{Approximation of Smooth Concave Function:} Suppose any smooth differentiable concave function $f(x): \mathbb{R} \to \mathbb{R}$, and its approximation, $\hat{f}(x): \mathbb{R} \to \mathbb{R}$, which is also smooth and differentiable (but not guaranteed to be concave), if $f(x)$ exhibits an decreasing upper-bound on the approximation error $\lVert \hat{f}(x) - f(x) \rVert \leq \epsilon$, then $\epsilon \to 0$ consequently implies $\norm{\hat{x}^* - x^*} \to 0$ almost surely, where $x^* = \mathrm{argmax} \ f(x)$, $\hat{x}^* = \mathrm{argmax} \ \hat{f}(x)$, and $|| \cdot ||$ denotes the supremum norm.
 
\end{lemma}

\begin{proof}
    
    Suppose some controllable function $\epsilon(N)$ is monotonically decreasing with respect to increasing $N$ bounding the supremum norm of the approximation error of $f(x)$, then the following statements hold true for any $x \in \mathbb{R}$.

    \begin{align}
        \norm{f(x^*) - \hat{f}(x^*)} \leq \epsilon(N), \quad \norm{f(\hat{x}^*) - \hat{f}(\hat{x}^*)}  \leq \epsilon(N)
    \end{align}

    Summing these conditions, we have the new relation,

    \begin{align}
        \norm{f(x^*) - \hat{f}(x^*)} + \norm{\hat{f}(\hat{x}^*) - f(\hat{x}^*)} \leq 2\epsilon(N)
    \end{align}

    By triangle inequality,

    \begin{align}
        \norm{f(x^*) - \hat{f}(x^*) + \hat{f}(\hat{x}^*) - f(\hat{x}^*)} \leq \norm{f(x^*) - \hat{f}(x^*)} + \norm{\hat{f}(\hat{x}^*) - f(\hat{x}^*)} \leq 2\epsilon(N) \label{eq:tri_ineq_f_bound}
    \end{align}

    Note: the order of $\hat{f}(\cdot) - f(\cdot)$ versus $f(\cdot) -  \hat{f}(\cdot)$ can be swapped and will yield the same result. The right left hand side of Eq. \eqref{eq:tri_ineq_f_bound}, allows us to rearrange the terms to obtain the equivalent expression rearranging the terms inide the norm on the left hand side,
    
    \begin{align}
       \norm{f(x^*) - f(\hat{x}^*) + \hat{f}(\hat{x}^*) - \hat{f}(x^*) } \leq 2\epsilon(N)
    \end{align}
    


    By definition, for any smooth function, 

    \begin{align}
        f(x^*) - f(\hat{x}^*) \geq 0, \quad \hat{f}(\hat{x}^*) - \hat{f}(x^*) \geq 0
    \end{align}

    Thus we have equality, 

    \begin{align}
        \lVert \underbrace{\hat{f}(\hat{x}^*) - \hat{f}(x^*)}_{ \geq 0} +  \underbrace{f(x^*) - f(\hat{x}^*)}_{ \geq 0} \rVert = \lVert \hat{f}(\hat{x}^*) - \hat{f}(x^*) \rVert + \norm{f(x^*) -  f(\hat{x}^*)}  
    \end{align}

    Thus we have the new inequality,

    \begin{align}
        0 \leq \lVert \hat{f}(\hat{x}^*) - \hat{f}(x^*) \rVert + \norm{f(x^*) -  f(\hat{x}^*)}  \leq 2 \epsilon(N)
    \end{align}

    As $2\epsilon(N) \to 0$, therefore this implies both,

    \begin{align}
        \lim_{{\epsilon \to 0}} \norm{\hat{f}(\hat{x}^*) - \hat{f}(x^*)} = 0 \\
        \lim_{{\epsilon \to 0}} \norm{f(x^*) -  f(\hat{x}^*)} = 0 \label{eq:concave_distance_fx}
    \end{align}

    \textbf{Concavity of $f(x)$:} This is where it is important to note the concavity of $f(x)$ particularly in Eq. \eqref{eq:concave_distance_fx}. Due to the concavity of $f(x)$, as the function approximation error approaches 0, so does $\norm{\hat{x}^* - x^*}$. 

    \begin{align}
        \lim_{{\epsilon \to 0}} \norm{f(x^*) -  f(\hat{x}^*)} = 0 \implies \lim_{{\epsilon \to 0}} \norm{\hat{x}^* - x^*} = 0
    \end{align}

    \end{proof}

\subsection{Notes on Optimistic Best Response} \label{sec:notes-opt-br}

The upper-bound of the optimistic best response, $\bar{b}_a$, can be determined by estimating the optimistic riskless price $\bar{p}_0$, while both sides of Eq. \eqref{eq:max-theta} provide optimistic estimates of expected profit, $\mathbb{E}[\mathcal{G}_B(\theta|a)]$, assuming the optimal price is at the riskless price, $p_0 = \mathcal{H}(\theta)$. These expressions summarize follower pricing and ordering under a greedy strategy, accounting for optimistic parameter estimation. In principle, although the follower may never converge to the optimal $p^*$ given the limitations of the approximation tools, aiming for convergence towards $p_0$ serves as a practical approach. $p_0$ represents not only a simplifying mathematical assumption, but also aligns with the idea of optimization under optimism - as $p^ \leq p_0$. By pricing as the riskless price, the follower aims to price optimistically, serving as a reasonable economic strategy for dynamic pricing.



\begin{align}
    \mathbb{E}[\bar{\mathcal{G}}_B(\theta|a)] &= \underset{\theta \in \mathcal{C}^t} {\mathrm{max}} \ \Big( p^*(\theta) - a \Big)  F_{\theta}^{-1} \Big(\frac{ p^*(\theta) - a}{  p^*(\theta) } \Big) \\
    &\leq  \underset{\theta \in \mathcal{C}^t} {\mathrm{max}} \ \Big( \frac{\mathcal{H}(\theta) + a}{2} \Big)  F_{\theta}^{-1} \Big(\frac{ \mathcal{H}(\theta) - a}{  \mathcal{H}(\theta) } \Big)  \\ 
    &= \underset{\theta \in \mathcal{C}^t} {\mathrm{max}} \ \Big( \frac{\mathcal{H}(\theta) + a}{2} \Big)  F_{\theta}^{-1} \Big(1 - \frac{2a}{\mathcal{H}(\theta) + a} \Big)  \\
    \bar{\theta} &=   \underset{\theta \in \mathcal{C}^t} {\mathrm{argmax}} \ \Big(\frac{\mathcal{H}(\theta) + a}{2} \Big)  F_{\theta}^{-1} \Big(1 - \frac{2a}{\mathcal{H}(\theta) + a} \Big)  \\
    \bar{b}_a &= F^{-1}_{\theta_a} \Big( \frac{p_0 - a}{p_0} \Big), \quad  p_0 = a + \frac{\bar{p} - a}{2} = \frac{\mathcal{H}(\theta) + a}{2}, \quad \mathfrak{B}(a) \in [\underline{b}_a, \bar{b}_a], \ \underline{b}_a \geq 0  \
\end{align}

\subsection{Proof of Theorem \ref{thm:follower-regret}} \label{prf:follower-regret}

\textbf{Follower Regret in Online Learning for NPG:} \textbf{Follower Regret in Online Learning for NPG:} Given pure leader strategy $\pi_A$, the worst-case regret of the follower, $R_B^T(\pi_A)$, as defined in Eq. \eqref{eq:follower-regret} when adopting the LNPG algorithm, is bounded by $R_B^T(\pi_A) \leq \aleph_B + \epsilon_B T - \theta_1 \kappa_H^2 \log^2(T)$, for sufficiently large values of $T$, where $\aleph_B, \epsilon_B, \theta_1, \kappa_H$ are positive real constants in $\mathbb{R}^+$.

\begin{proof}
    Let $\mathcal{U}_B(p|a)$ represent the riskless profit for the follower at follower action (retail price) $p$, give leader action (supplier price) $a$.
    \begin{align}
        \mathcal{U}_B(p|a) &= \theta_0(p-a) - \theta_1(p-a)^2 = -\theta_1 \Big( (p-a) - \frac{\theta_0}{2\theta_1} \Big)^2 + \frac{\theta_0^2}{4 \theta_1} 
    \end{align}

    We argue that given any pure strategy $\pi_A$, the contextual regret $R_B^T(\pi_A)$, is bounded by $R_B^T(\pi_A) \leq \aleph_B + \epsilon_B T - \theta_1 \kappa_H^2 \log^2(T)$, for well-defined positive real value constants $\aleph_B, \epsilon_B, \theta_1, \kappa_H$. For the sake of the argument, we assume that given $a$, $p(z)^*$ which is the solution to Eq. \eqref{eq:pz_optim} is known. Where $z$ represents the quantity above the expected demand from $\Gamma_\theta(p)$ (refer to Appendix \ref{prf:optimal-order-b}). For the mechanics of this proof, this allows us to consider $\mathbbm{E}[\mathcal{G}_B(p^*|a)]$ a constant value.


    \textbf{Riskless Pricing Difference $\Delta_\mathcal{U}(t)$:} First we begin by analyzing the convergence properties of the upper-bounding function $\mathcal{U}_B(p|a)$.  We show that $\Delta_\mathcal{U}(t) \leq \theta_1 \kappa_H^2 \log(t)/t$. The difference in exact profit between pricing at the riskless price $\mathcal{U}_B(p_0)$ and some optimistic riskless price $\mathcal{U}_B(\bar{p}_0)$, denoted as $\Delta_\mathcal{U}(t) = \mathcal{U}_B(p_0) - \mathcal{U}_B(\bar{p}_0)$ is such that $\Delta_\mathcal{U}(t) \leq \theta_1 \kappa_H^2 \log(t)/t$.

    For convenience, in equations where no ambiguity exists, we may omit from explicitly including the variable $a$ in the subsequent equations, denoting them as $\mathcal{G}_B(p|a) \equiv \mathcal{G}_B(p)$, for $\mathcal{U}_B(p|a), \equiv \mathcal{U}_B(p)$ and so forth.
    
    \begin{align}
        \mathcal{U}_B(p) &= (\theta_0 - \theta_1p)(p-a) \\
        &= -\theta_1\left(p^2 - \frac{\theta_1a + \theta_0}{\theta_1}p\right) + \theta_0a \\
        &= -\theta_1\left(p - \frac{\theta_1a + \theta_0}{2\theta_1}\right)^2 + \frac{\theta_1a^2}{4} + \frac{3\theta_0a}{2} + \frac{\theta_0^2 }{4\theta_1} \\
        &= -\theta_1\left(p - \frac{\theta_1a + \theta_0}{2\theta_1}\right)^2 + g(a, \theta_0, \theta_1), \quad \text{where} \quad g(a, \theta_0, \theta_1) \equiv \frac{\theta_1a^2}{4} + \frac{3\theta_0a}{2} + \frac{\theta_0^2 }{4\theta_1} \\
        &= -\theta_1\left(p - \frac{\mathcal{H}(\theta) + a}{2}\right)^2 + g(a, \theta_0, \theta_1)
    \end{align}
    
    The difference in the riskless price, note $p_0 = \frac{\mathcal{H}(\theta) + a}{2}$, combining Theorem \ref{thm:bound-del-ht},
    
    \begin{align}
        \Delta_\mathcal{U}(t) &= \mathcal{U}_B(p_0) - \mathcal{U}_B(\bar{p}_0) \\
        &= g(a, \theta_0, \theta_1) -\theta_1 \Big(p_0 - \frac{\mathcal{H}(\theta) + a}{2} \Big)^2  - g (a, \theta_0, \theta_1 ) + \theta_1 \Big(p_0 + \Delta_\mathcal{H}(t) - \frac{\mathcal{H}(\theta) + a}{2} \Big)^2 \\
        &= \theta_1(\Delta_\mathcal{H}(t) )^2  \\
        &\leq \theta_1 \kappa_H^2 \log(t)/t \label{eq:logt-over-t}
    \end{align}
    
    We integrate the term $\log(t)/t$ to obtain obtain a closed form expression for Eq. \eqref{eq:logt-over-t}.
    
    \begin{align}
        \int_{1}^{T} \frac{\log(t)}{t} \, dt =  \frac{1}{2} \Big( \log^2(T) - \log^2(1) \Big) 
    \end{align}
    
    \textbf{Regret Analysis:} We introduce $a$ back again into the notation, as the expected profit, $\mathbbm{E}[\mathcal{G}_B(p|a)]$ and upper-bounding functions $\mathcal{U}_B(p)$, are conditioned on the leader action $a$. We present the definition of regret $R_B^T(\pi_A)$ conditioned upon the policy of $\pi_A$, which is a pure strategy consisting of a sequence of action. As we assume a rational and greedy follower; he does not attempt to alter the strategy of the leader.

    \begin{align}
        R_B^T(\pi_A) &=  \sum^T_{t = 1} \mathbbm{E}[\mathcal{G}_B(p^* | a)] - \sum^T_{t = 1} \mathbbm{E}[\mathcal{G}_B(p^t| a)] \\
        &= \sum^T_{t = 1} \  \Big( \mathbbm{E}[\mathcal{G}_B(p^* | a)] - \mathbbm{E}[\mathcal{G}_B(\bar{p}_0 | a)] \Big)
    \end{align}
    
     We also know that if we use the riskless price estimate.

    \begin{align}
         \mathbbm{E}[\mathcal{G}_B(\bar{p}_0|a)] &\leq \mathbbm{E}[\mathcal{G}_B(p^*|a)] \leq \mathcal{U}_B(p_0|a) \\
         \text{where,} \quad \mathbbm{E}[\mathcal{G}_B(p|a)] &\leq \mathcal{U}_B(p|a), \quad \forall p, \ \forall a 
    \end{align}

    Where $\bar{p}_0$ is the upper estimate of the riskless price $p_0$. And in the general case,
    
    \begin{align}
         \mathbbm{E}[\mathcal{G}_B(\tilde{p}|a)] &\leq \mathbbm{E}[\mathcal{G}_B(p^*|a)] \leq \mathcal{U}_B(p_0|a) \\
         \text{where,} \quad &\tilde{p} \in \{x \in \mathbb{R}^+ : x \neq p^*\} \\
         \text{where,} \quad \mathbbm{E}[\mathcal{G}_B(p|a)] &\leq \mathcal{G}_B(p^*|a), \quad \forall p, \ \forall a 
    \end{align}

    For any price $p$, we know that the profits under no inventory risk is always higher than when inventory risk exists, $\mathbbm{E}[\mathcal{G}_B(p|a)] \leq \mathcal{U}_B(p|a)$. We make note of the relation $\mathcal{G}_B(p|a) \leq \mathcal{U}_B(p|a), \ \forall p \in \mathbb{R}^+, \ \forall a \in \mathbb{R}^+$. Thus we define $\mathcal{Z}_B(p|a)$ as,
    
    \begin{align}
        \mathcal{Z}_B(p|a) = \mathcal{U}_B(p|a) - \mathbbm{E}[\mathcal{G}_B(p|a)] 
    \end{align}

    As a consequence of this approximation error, $\mathcal{Z}_B(p|a)$, we note there must be a constant term which upper bounds the regret of the follower.

    \begin{align}
         \epsilon_0(a) = \mathbbm{E}[\mathcal{G}_B(p^*|a)] - \mathbbm{E}[\mathcal{G}_B(p_0|a)]
    \end{align}

    Our approach involves analyzing two cases. 
    
    \textbf{Case 1:} $\mathbbm{E}[\mathcal{G}_B(p^*)] \leq \mathcal{U}_B(\bar{p}_0)$, 
    
    \begin{align}
        R_B^T(\pi_A) &\leq \sum^T_{t = 1} \ \mathcal{U}_B(\bar{p}_0) - \mathbbm{E}[\mathcal{G}_B(p_0)]   \\
        &= \sum^T_{t = 1} \Big( \mathcal{U}_B(p_0) - \Delta_\mathcal{U}(t) \Big)  - \Big( \mathcal{U}_B(p_0) - \mathcal{Z}_B(p_0)  \Big) \\
        &= \sum^T_{t = 1} \mathcal{Z}_B(p_0) - \Delta_\mathcal{U}(t)
    \end{align}

    \textbf{Case 2:} $\mathbbm{E}[\mathcal{G}_B(p^*)] > \mathcal{U}_B(\bar{p}_0)$,

    \begin{align}
        R_B^T(\pi_A) &\leq \sum^T_{t = 1} \  \ \mathcal{U}_B(p_0) - \mathbbm{E}[\mathcal{G}_B(\bar{p}_0)] \\
        &\leq \sum^T_{t = 1} \  \ \mathcal{U}_B(p_0) - (\mathcal{U}_B(\bar{p}_0) - \mathcal{Z}_B(\bar{p}_0) ) \\
        &= \sum^T_{t = 1} \  \ \mathcal{U}_B(p_0) - (\mathcal{U}_B(p_0) - \Delta_\mathcal{U}(t) - \mathcal{Z}_B(\bar{p}_0)) \\
        &= \sum^T_{t = 1} \ \mathcal{Z}_B(\bar{p}_0) + \Delta_\mathcal{U}(t) \label{eq:case1_follower_regret_bound_pbar}
    \end{align}

    We define $\Delta_\mathcal{Z}(t)$ as,

    \begin{align}
        \Delta_\mathcal{Z}(t) &= \mathcal{Z}_B(\bar{p}_0) - \mathcal{Z}_B(p_0) \\
        &= \mathcal{U}_B(\bar{p}_0) - \mathbb{E}[\mathcal{G}_B(\bar{p}_0)] - \mathcal{U}_B(p_0) + \mathbb{E}[\mathcal{G}_B(p_0)] \\
        &= \underbrace{\mathcal{U}_B(\bar{p}_0) - \mathcal{U}_B(p_0)}_{\Delta_\mathcal{U}(t)} + \underbrace{\mathbb{E}[\mathcal{G}_B(p_0)] - \mathbb{E}[\mathcal{G}_B(\bar{p}_0)]}_{\Delta_\mathcal{G}(t)} 
    \end{align}

    We take the definition of $\Delta_\mathcal{Z}(t)$ and insert it back to the definition in Eq. \eqref{eq:case1_follower_regret_bound_pbar}.
    
    \begin{align}
        R_B^T(\pi_A) &\leq \sum^T_{t = 1} \ \mathcal{Z}_B(\bar{p}_0) + \Delta_\mathcal{U}(t) \\
        &= \sum^T_{t = 1} \ \mathcal{Z}_B(p_0) + \Delta_\mathcal{Z}(t) + \Delta_\mathcal{U}(t) \\
        &= \sum^T_{t = 1} \ \mathcal{Z}_B(p_0) + 2 \Delta_\mathcal{U}(t) + \Delta_\mathcal{G}(t)
    \end{align}

    We note the peculiarity that $\Delta_\mathcal{G}(t) \in \mathbb{R}$ and $\Delta_\mathcal{U}(t) \geq 0$. Thus in the worst case,

    \begin{align}
        R_B^T(\pi_A) &\leq \sum^T_{t = 1} \ \mathcal{Z}_B(p_0) + 2 \Delta_\mathcal{U}(t) + | \Delta_\mathcal{G}(t)|
    \end{align}
    
    We upper bound $\mathcal{Z}_B(\bar{p}_0 | a)$ by $\tilde{\epsilon}(a)$, where,  

    \begin{align}
        \mathbbm{E}[\mathcal{G}_B(p^*|a)]  - \mathbbm{E}[\mathcal{G}_B(p_0|a)] &\leq \mathcal{U}_B(p_0|a) - \mathbbm{E}[\mathcal{G}_B(p_0|a)] \\
        \tilde{\epsilon}_0 &=\mathcal{U}_B(p_0|a) - \mathbbm{E}[\mathcal{G}_B(p_0|a)] \label{eq:approximation_error_defn_b}
    \end{align}

    We note that $\mathcal{U}_B(p|a)$ can be expressed in linear form, with parameters $\theta$, whereas $\mathbbm{E}[\mathcal{G}_B(p|a)]$ requires the solution to Eq. \eqref{eq:pz_optim}, which may not always have a closed form solution. 

    \begin{align}
        \tilde{\epsilon}(a) &= \underset{p \in \mathbb{R}^+ } {\mathrm{max}} \ \mathcal{Z}_B(p|a) 
    \end{align}
        
    We can see that there are two contributions to the contextual regret. $\tilde{\epsilon}(a)$ is the worst-case error term of the approximation of $\mathbb{E}[\mathcal{G}_B(\cdot)]$ via $\mathcal{U}_B(\cdot)$. $\Delta_\mathcal{U}(t)$ represents the learning error of learning to achieve riskless pricing under imperfect information. Economically, $\tilde{\epsilon}(a)$ represents the upper-bound on the difference between expected profit with and without potential inventory risk. Thus using this approximation, we learn to price at $p_0$, and converge to a constant error term $\epsilon_0(a)$. 
    
    
    
    In Case 2, where $\mathbbm{E}[\mathcal{G}_B(p^*|a)] \geq \mathcal{U}_B(\bar{p}_0|a)$, the follower learns to price at $\bar{p}_0$ due to misspecification. Nevertheless, in terms of asymptotic behaviour, as $t \xrightarrow[]{} \infty$, $\bar{p}_0 \xrightarrow[]{} p_0$ and $\Delta_\mathcal{U}(t) \xrightarrow[]{} 0$ almost surely. Since by definition $\mathcal{U}_B(p_0) \geq \mathbbm{E}[\mathcal{G}_B(p)], \ \forall p \in \mathbb{R}^+$, we know that the system will approach Case 1 as $t \xrightarrow[]{} \infty$. In this case, the learning error $\Delta_\mathcal{U}(t)$ is vanishing, while the riskless inventory approximation error $\epsilon_0(a)$ is persistent, and bounded by $\epsilon_0(a) \leq \tilde{\epsilon}(a)$. We can now we can impose a worst case bound on expected regret, assuming $t$ is sufficiently large, and $\tilde{\epsilon}(a) \geq \Delta_\mathcal{U}(t), \, \forall a \in \mathbb{R}^+$ (which is true in Case 1). We consider the behaviour of $T$ as it approaches infinity. When considering relatively small values of $T$, the regret may exhibit a worst-case regret spanning the entire range of the $\mathcal{U}(\cdot)$, given the absence of constraints on the magnitude of $|\Delta_\mathcal{G}(t)|$. Therefore, it becomes imperative to establish the assumption that the temporal horizon $T$ is sufficiently large.
    
    
    \begin{align}
        R_B^T(\pi_A) &\leq \sum^T_{t = 1} \ \Big( \tilde{\epsilon}(a) - \Delta_\mathcal{U}(t) \Big) \label{eq:inf_reg_anal_b}
    \end{align}

    Therefore, when examining the asymptotic behaviour for Case 2, when letting $\epsilon_B$ be the maximum value of $\tilde{\epsilon}(a)$ for all $a \in \mathbb{R}^+$, we observe that,

    \begin{align}
        R_B^T(\pi_A) &\leq  \epsilon_B T - \theta_1 \kappa_H^2 \log^2(T) \\
        \text{where} \quad \epsilon_B &= \underset{a \in \mathbb{R}^+ } {\mathrm{max}} \ \mathcal{Z}_B(\tilde{\epsilon}(a)|a) \label{eq:epsilon_b_interms_z}
    \end{align}
    

    Where Eq. \eqref{eq:epsilon_b_interms_z} serves as an alternative expression for Eq. \eqref{eq:br_approx_def_eps_b}. As recounted in Eq. \eqref{eq:approximation_error_defn_b} and Eq. \eqref{eq:inf_reg_anal_b}, the system will approach Case 1 almost surely as $t \xrightarrow[]{} \infty$. We note that the constant worst-case linear scaling via $\epsilon_B$ due to the approximate best response of the follower is external to the learning algorithm, and thus is dependent on the error of the approximation of the solution to Eq. \eqref{eq:pz_optim}, via $\mathcal{Z}_B(p_0)$.
    
    \textbf{General Case:} We combine the characteristics of Case 1 and Case 2 and express the instantaneous regret behaviour as,
    
    \begin{align} 
      r^t_B(\pi_A) = 
        \begin{cases} 
          \mathcal{Z}_B(p_0) + 2 \Delta_\mathcal{U}(t) + \Delta_\mathcal{G}(t), \quad & 1 \leq t \leq T'  \\ 
          \mathcal{Z}_B(p_0) - \Delta_\mathcal{U}(t) \quad & T' < t
        \end{cases} \label{eq:inst_reg_b_cases}
    \end{align} 

    To clarify, $T'$ denotes the change point where the system transitions from Case 2, where $\mathbbm{E}[\mathcal{G}_B(p^*|a)] > \mathcal{U}_B(\bar{p}_0|a)$, to Case 1, where $\mathbbm{E}[\mathcal{G}_B(p^*|a)] \leq \mathcal{U}_B(\bar{p}_0|a)$. 
    
    \begin{align} 
      R^T_B(\pi_A) &= \sum_{t=1}^T r^t_B(\pi_A) \\
      &\leq \int_{t=1}^T \, \sum_{t=1}^T r^t_B(\pi_A) \\
      &= \int_{t=1}^{T'} \, \Big( \mathcal{Z}_B(p_0) + 2 \Delta_\mathcal{U}(t) + |\Delta_\mathcal{G}(t)| \Big) dt + \int_{t=T'}^{T} \, \Big( \mathcal{Z}_B(p_0) - \Delta_\mathcal{U}(t) \Big) dt
    \end{align} 

    Suppose an upper-bounding term bounds $|\Delta_\mathcal{G}(t)|$, such that $|\Delta_\mathcal{G}(t)| \leq \bar{\Delta}_\mathcal{G}$. We can then rearrange,

    \begin{align} 
      R^T_B(\pi_A) &\leq \mathcal{Z}_B(p_0)(T'-1) + \bar{\Delta}_\mathcal{G}(T'-1) + 2 \int_{1}^{T'} \, \Delta_\mathcal{U}(t) \, dt + \mathcal{Z}_B(p_0)(T-T') - \int_{T'}^{T} \,  \Delta_\mathcal{U}(t) \, dt \\
      &= C_\mathcal{Z} + T \mathcal{Z}_B(p_0) + 2 \theta_1 \kappa_H^2 \log^2(t) \Big|_{1}^{T'} - \theta_1 \kappa_H^2 \log^2(t) \Big|_{T'}^{T}
    \end{align}
    

    Although $\Delta_\mathcal{G}$ is fundamentally the value of interest, it potentially exhibits arbitrary growth until $T'$. Given the unconstrained nature of this growth, we are particularly concerned about Case 2. Therefore, in a hypothetical scenario where we possess knowledge of $\mathcal{G}_B(p^*|a)$ and $\mathcal{U}(p_0|a)$ (without direct access to $p^*$ and $p_0$), we could theoretically determine $t$, contingent upon the convergence rate of $\theta$ towards $\theta^*$, as outlined in Lemma \ref{thm:theta-2-norm}.
    
    \begin{align}
        C_\mathcal{Z} = \bar{\Delta}_\mathcal{G}(T'-1) - \mathcal{Z}_B(p_0)
    \end{align}

    From which, it holds that,
    \begin{align}
      R^T_B(\pi_A) &\leq C_\mathcal{Z} + T \mathcal{Z}_B(p_0) + \theta_1 \kappa_H^2 \Big( 2 \log^2(T') - \log^2(T) -\log^2(T') \Big) \\
      &= C_\mathcal{Z} + T \mathcal{Z}_B(p_0) + \theta_1 \kappa_H^2 \Big( \log^2(T') - \log^2(T) \Big) \\
      &= \aleph_B + \epsilon_B T  - \theta_1 \kappa_H^2 \log^2(T)
    \end{align} 

    Where $\aleph_B$ absorbs the finite $\theta_1 \kappa_H^2 \log^2(T')$ term, $\aleph_B = C_\mathcal{Z} + \theta_1 \kappa_H^2 \log^2(T')$, and $\epsilon_B = \mathcal{Z}_B(p_0)$ by definition.


    
    
    

\end{proof}

\subsection{Proof for Theorem \ref{thm:leader-regret}} \label{prf:leader-regret}
    \textbf{Stackelberg Regret:} Given the pure strategy best response of the follower defined $\mathfrak{B}(a) \in [\underline{b}_a, \bar{b}_a]$ from Eq. \eqref{eq:ba_lower_bound_true} and \eqref{eq:ba_upper_bound_true} respectively, the worst-case regret, $R_A^T(\pi_A)$, from the leader's perspective, as defined in Eq. \eqref{eq:leader-regret-def} when adopting the LNPG algorithm, is bounded by $\mathcal{O}(\sqrt{T \log(T)})$.

\begin{proof}
    From the perspective of the leader, we are given the bounds on the follower's best response to any action $a$, denoted as $\mathfrak{B}(a) \in [\underline{b}_a, \bar{b}_a]$ from Eq. \eqref{eq:ba_lower_bound_true} and \eqref{eq:ba_upper_bound_true} respectively. And we can express this as a function of $\mathcal{H}^*(\theta)$ (which is the riskless price). Let's take an upper bound on the actions $b$, 

    \begin{align}
        \bar{b}_a &= F^{-1}_{\bar{\theta}_a} \Big( 1 - \frac{ 2a}{ \mathcal{H}^*(\theta) + a} \Big)
    \end{align}
    
    In our algorithm the leader has access to the same parameter estimates of $\theta$ as the follower, and the follower is running OFUL, with logarithmic regret bounds. Given the same confidence ball $\mathcal{C}^t$ which leads to an optimistic parameter estimate $\bar{\theta}$, the leader can also optimistically estimate her rewards, $\bar{\mathcal{G}}_A(\theta)$, by anticipating the optimistic best response of the follower.
    
    \begin{align}
        \bar{\mathcal{G}}_A(\theta) &= \underset{a \in \mathcal{A}} {\mathrm{max}} \ a \bar{b}_a = \underset{a \in \mathcal{A}} {\mathrm{max}} \ a F^{-1}_{\bar{\theta}} \Big( 1 - \frac{ 2a}{ \mathcal{H}^*(\theta) + a} \Big)\label{eq:ga_profit_upper}
    \end{align}
    
    Essentially the leader plays the arm that maximizes $\bar{\mathcal{G}}_A(\theta)$, which is the optimization problem for the leader. We know that the solution to $\mathcal{H}^*(\theta)$ is expressed in Eq. \eqref{eq:hstar-def}. We express leader regret by assuming $\underline{b}$ to be the lower bound of the optimistic best response of the follower.
    

    
    \begin{align}
        R_A^T(\pi_A) &= \sum^T_{t = 1} \mathcal{G}_A(a^*, \mathfrak{B}(a) ) - \sum^T_{t = 1} \mathcal{G}_A(a^t, b^t) \\
        &= \sum^T_{t = 1} a^* \mathfrak{B}(a^*)  - \sum^T_{t = 1}  a^t \mathfrak{B}(a^t) \\
        &= \sum^T_{t = 1} \ \underset{a \in \mathcal{A}} { \mathrm{max}} \ a F^{-1}_{\theta^*} \Big( 1 - \frac{2a}{\mathcal{H}^*(\theta^*) + a } \Big)  - \sum^T_{t = 1} \bar{a} F^{-1}_{\bar{\theta}} \Big( 1 - \frac{ 2 \bar{a} }{\mathcal{H}^*(\theta^*) + \bar{a}} \Big) \\ 
        &\leq \sum^T_{t = 1} \ \underset{a \in \mathcal{A}} { \mathrm{max}} \ a F^{-1}_{\theta^*} \Big( 1 - \frac{2a}{\mathcal{H}^*(\theta^*) + a } \Big) - \sum^T_{t = 1} \bar{a} F^{-1}_{\theta^*} \Big( 1 - \frac{ 2 \bar{a} }{\mathcal{H}^*(\theta^*) + \bar{a}} \Big) \label{leader-regret-bound-expression} \\ 
        a^* &= \underset{a \in \mathcal{A}} { \mathrm{argmax}} \ a \, \mathfrak{B}(a), \quad \bar{a} = \underset{a \in \mathcal{A}} { \mathrm{argmax}} \ a F^{-1}_{\bar{\theta}} \Big( 1 - \frac{2a}{\mathcal{H}^*(\theta^*) + \Delta_{\mathcal{H}} + a} \Big) \label{leader-action-upperbound}
    \end{align}
    
    There exists a confidence ball for both agents $\mathcal{C}^t$, one of which could be the subset of another (i.e. one agent has more information than the other), and a unique maximizing point for $\mathcal{G}_A(a^t)$, for $b(\theta^t)$, which depends on the estimate of $\theta \in \mathcal{C}^t$. From Theorem \ref{thm:optimal-order-b}, we know that $b_a$ is monotonically decreasing w.r.t. $p$, and from Lemma \ref{thm:cum-del-HT}, we know that the optimistic estimate of $\bar{p}_0$ is decreasing. There exists a unique maximizing solution when $\theta^*$ is known. However when we solve for the approximation, there is a suboptimality introduced by $\bar{a}$.

    We would like to quantify and bound the behaviour of $\bar{a}(t)$ from Eq. \eqref{leader-regret-bound-expression}, and thus we can then quantify the behaviour of leader regret, in Eq. \eqref{leader-regret-bound-expression}. First we note an assumption from Section \ref{sec:assumptions}, where the demand distribution is symmetric $\mathcal{N}(\mu^t, \sigma)$, with constant variance $\sigma$, where $\mu = \Gamma_\theta(p)$, a property of additive demand. We characterize the change in $\gamma$ with respect to order amount (or follower response) $b$. We know the relationship, when considering optimistic demand parameters $\bar{\theta}$, that $b = F^{-1}_{\bar{\theta}}(\gamma)$. Suppose we hold $a$ constant, the quantifiable difference of $\gamma(a^t)$ to the optimal $\gamma(a^*)$ is defined in Eq. \eqref{eq:gamma-diff}. 
    \begin{align}
        \gamma &= 1 - \frac{2a}{\mathcal{H}^*(\theta^*) + \Delta_{\mathcal{H}} + a} \\
        \Delta_\gamma(t) &= \gamma(a^t) - \gamma(a^*) \label{eq:gamma-diff} \\
        &= 1 - \frac{2a}{\mathcal{H}^*(\theta^*) + \Delta_{\mathcal{H}} + a } - 1 + \frac{2a^*}{\mathcal{H}^*(\theta^*) + a^*}  \\
        &= \frac{2a^*}{\mathcal{H}^*(\theta^*) + a^*} - \frac{2a}{\mathcal{H}^*(\theta^*) + \Delta_{\mathcal{H}} + a}
    \end{align}
    

    With knowledge that $a^t$ with respect to $a^*$, there could be multiple scenarios, either  $a^t > a^*$, or $a^t \leq a^*$. We can argue is that we start $a^t \geq a^*$ with probability $1 - \delta$ as we are optimistic, and seek to set prices always higher than the optimal $a^*$, therefore the series $a^t$ must be decreasing, and we known by design $a^t \geq 0$, therefore $a^t - a^* \geq 0$.
    
    \begin{align}
        \Delta_\gamma(t) &\leq 2a^t \Big( \frac{1}{\mathcal{H}^*(\theta^*) + a^t} - \frac{1}{\mathcal{H}^*(\theta^*) + \Delta_{\mathcal{H}} + a^t} \Big)  \\
        &= \frac{2a^t}{\mathcal{H}^*(\theta^*) + a^t} \Big( 1 - \frac{\mathcal{H}^*(\theta^*) + a^t}{\mathcal{H}^*(\theta^*) + \Delta_{\mathcal{H}} + a^t} \Big) \\
        &= \frac{2a^t}{\mathcal{H}^*(\theta^*) + a^t} \Big( \frac{\Delta_{\mathcal{H}}}{\mathcal{H}^*(\theta^*) + \Delta_{\mathcal{H}} + a^t} \Big), \quad a^* \leq a^t \label{eq:delta_h_bound_no_T} 
    \end{align}

    
    We know from Lemma \ref{thm:demand-optimistic-reduction} that we have a bound, $\Delta_{\mathcal{H}} \leq 2\kappa_H \sqrt{\log(t)/t}$. So if we substitute this bound into Eq. \eqref{eq:delta_h_bound_no_T} which maintains the inequality, we have,

    \begin{align}
        \Delta_\gamma(t) &\leq \frac{2a}{\mathcal{H}^*(\theta^*) + a} \Big( \frac{2\kappa_H \sqrt{\log(t)/t} }{\mathcal{H}^*(\theta^*) + a + 2\kappa_H \sqrt{\log(t)/t}} \Big) \\
        &= \frac{2a}{\mathcal{H}^*(\theta^*) + a} \Big( \frac{2\kappa_H \log(t)}{2\kappa_H \log(t) + (\mathcal{H}^*(\theta^*) + a) \sqrt{t\log(t)}} \Big) \\
        &< \frac{2a}{\mathcal{H}^*(\theta^*) + a} \Big( \frac{2\kappa_H \log(t)}{(\mathcal{H}^*(\theta^*) + a) \sqrt{t\log(t)}} \Big) \\
        &= \frac{4 a \kappa_H }{(\mathcal{H}^*(\theta^*) + a)^2} \Big( \sqrt{\frac{\log(t)}{t}} \Big) \\
        &\leq \frac{4 \kappa_H }{\mathcal{H}^*(\theta^*)^2} \Big( \sqrt{\frac{\log(t)}{t}} \Big) a
        \label{eq:delta_gamma}
    \end{align}
    
    Suppose $F_{\bar{\theta}}(b)$ is continuous and differentiable in the domain of $b \in [0, \bar{b}]$, and there exists a finite minimum $\mathscr{M} > 0$ lower bounds the derivative of $F_{\bar{\theta}}(b)$ with respect to $b$, 

    \begin{align}
        \mathscr{M} \leq \frac{\partial \gamma}{\partial b} =  \frac{\partial F_{\bar{\theta}}(b) }{\partial b} \label{eq:Mb_bound} 
    \end{align}
    
    We can use the existence of $\mathscr{M}$ to bound the change in $\gamma$ with respect to the change in $b$, denoted as $\Delta_\gamma(t)$ and $\Delta_b(t)$ respectively. 
    
    \begin{align}
        \Delta_b(t) \leq \frac{\Delta_\gamma(t)}{\mathscr{M}}
    \end{align}


    We construct a lower bound on $\mathcal{G}_A(\bar{a}^t, \mathfrak{B}(\bar{a}^t))$, where $\bar{\mathfrak{B}}(a)$ denotes the optimistic best response as a function of leader action $a$. For notation purposes, unless otherwise specified, $\mathcal{G}_A(a)$ denotes that the leader is acting with $a$, and the follower is responding with $\mathfrak{B}(a)$, thus we will use the shorthand definition $\mathcal{G}_A(a)$,

    \begin{align}
          \mathcal{G}_A(a) \equiv \mathcal{G}_A(a, \mathfrak{B}(a)) = a \, \mathfrak{B}(a), \forall a \in \mathcal{A}
    \end{align}
    
    
    Subsequently, we can bound the instantaneous expected leader reward by examining the behaviour of $\Delta_a$ and $\Delta_b$, where if $\Delta_a > 0 \xrightarrow[]{} \Delta_b < 0$.  Any suboptimal action $a = a^* + \Delta_a$ will result in a potentially lower expected reward for the leader, thus $\mathcal{G}_A(a) \leq \mathcal{G}_A(a^*)$. We can quantify this reduction in reward as,
    

    \begin{align}
        \mathcal{G}_A(a^* + \Delta_a)  = (a^* + \Delta_a)(\mathfrak{B}(a) + \Delta_b(a, t)) \leq  \mathcal{G}_A(a^*) %
    \end{align}
    
    Let us characterize $\mathcal{G}_A(\bar{a}^t)$, denoted in Eq. \eqref{eq:GAB_opt}, as the leader's reward, under a best responding follower under perfect information, given an estimate of $a$ computed using an optimistic estimate of $\theta$ with respect to $t$, denoted as $\bar{a}^t$. We define a lower bound to $\mathcal{G}_A(\bar{a}^t)$, in Eq. \eqref{eq:GAB_lower_bound}, to establish that $\mathcal{G}_A(\bar{a}^t) \geq \underline{\mathcal{G}}_A(\bar{a}^t)$.
    
    \begin{align}
          \mathcal{G}_A(\bar{a}^t) &= \bar{a}^t \, \mathfrak{B}(\bar{a}^t) = (a^* + \Delta_a(t)) \mathfrak{B}(a^* + \Delta_a(t)) \label{eq:GAB_opt} \\
          &\geq a^* \mathfrak{B}(a^* + \Delta_a(t)) \\
          &= a^* \mathfrak{B}(\underset{a \in \mathcal{A} } { \mathrm{argmax}} \  a \bar{\mathfrak{B}}(a) ) = a^* \mathfrak{B} \Big( \underset{a \in \mathcal{A} } { \mathrm{argmax}} \  a  F^{-1}_{\bar{\theta}} (\bar{\gamma}(a)) \Big) \\
          &\geq a^* \mathfrak{B} \Big( \underset{a \in \mathcal{A} } { \mathrm{argmax}} \  a  F^{-1}_{\theta} (\bar{\gamma}(a)) \Big)  \equiv \underline{\mathcal{G}}_A(\bar{a}^t) \label{eq:GAB_lower_bound}
    \end{align}
    
    We denote upper bound $\bar{a}$ as the optimistic estimate of $a^*$, and quantify the relation from Eq. \eqref{leader-action-upperbound}, in a closed form expression. The value in Eq. \eqref{eq:GAB_lower_bound} forms a lower bound for the best response under optimism $\mathcal{G}_A(\bar{a}^t)$, and therefore $\underline{\mathcal{G}}(\bar{a}^t)$ can be used to form an upper bound on the leader regret. To be more specific, $\underline{\mathcal{G}}(\bar{a}^t)$ is the hypothetical value, which lower bounds the leader's reward obtained under optimistic misspecification, $\mathcal{G}_A(\bar{a}^t)$.
    

    

    We give the expressions for $a^*$ , the optimal leader action, $\bar{a}$ the optimistic leader action, and $\hat{\bar{a}}$ an estimate of the optimistic leader action, in Eq. \eqref{eq:opt_a} and Eq. \eqref{eq:a_bar_convex_maximizer} respectively.
    
    \begin{align}
        a^* &= \underset{a \in \mathcal{A}} { \mathrm{argmax}} \ a \, \mathfrak{B}(a) \label{eq:opt_a} \\
        \bar{a} &= \underset{a \in \mathcal{A} } { \mathrm{argmax}} \  a \, F^{-1}_{\bar{\theta}} \Big( 1 - \frac{2a}{\mathcal{H}^*(\theta^*) + \Delta_\mathcal{H} + a} \Big) \label{eq:a_bar_convex_maximizer} 
    \end{align}


    Eq. \eqref{eq:a_bar_convex_maximizer} is strictly concave and thus a unique solution to $\bar{a}$ exists. The major obstacle currently is that there is no closed form solution to $F^{-1}_{\bar{\theta}}$, and therefore we need to find an appropriate approximation function such that we can perform the maximization over the concave space with controllable approximation error. Let $\mathcal{P}_N(\cdot)$ denote an approximation of $F^{-1}_{\bar{\theta}}(\cdot)$, where $\mathcal{P}_N(\cdot)$ can be expressed by a closed form solution. We estimate the $\bar{a}$ by setting,

    \begin{align}
        F^{-1}_{\bar{\theta}}\Big( 1 - \frac{2a}{\mathcal{H}^*(\theta^*) + \Delta_\mathcal{H}}\Big) \approx \mathcal{P}_N \Big( 1 - \frac{2a}{\mathcal{H}^*(\theta^*) + \Delta_\mathcal{H}} \Big) \label{eq:afn_approx}
    \end{align}
    
    Where Eq. \eqref{eq:afn_approx} represents an approximation of the inverse CDF $F^{-1}_{\bar{\theta}}(\cdot)$. This approximation will give us a maximum approximation error denoted as,

    \begin{align}
        \widetilde{\Delta}_a = \underset{a \in \mathcal{A} } { \mathrm{argmax}} \ \Big[  a \, \mathcal{P}_N \Big( 1 - \frac{2a}{\mathcal{H}^*(\theta^*) + \Delta_\mathcal{H} + a} \Big) \Big] - \bar{a}
    \end{align}

    Where $\widetilde{\Delta}_a$ is the error in approximation of $\bar{a}$ due to the approximation of the inverse CDF $F_\theta^{-1}(\cdot)$ with $\mathcal{P}_N(\cdot)$.

    \textbf{Bounded Action Space:} We stipulate explicitly that, as a consequence of a bounded action space for the follower $[\underline{b}, \bar{b}]$, the interval of the the critical fractile is also bounded in $[\underline{\gamma}, \bar{\gamma}]$. This implies that the image of $F^{-1}_\theta(\underline{\gamma})$ is lower bounded by 0.
    
    \begin{align}
        F^{-1}_\theta(\underline{\gamma}) = 0, \quad \gamma \in [\underline{\gamma}, \bar{\gamma}] \label{eq:F_underline_gamma0}
    \end{align}

    Since in the economic sense we cannot order negative items, we implement a lower bound on the action space of the follower, and conduct our analysis in the range $[\underline{\gamma}, \bar{\gamma}]$, implying a finite upper bound on $a$, such that $a \in [\underline{a}, \bar{a}]$ for the leader. Within this range lies the optimal solution for $\underline{G}_A(\cdot)$ . 
    

    \textbf{Weierstrass Approximation Theorem:}  For any continuous function \( f(\gamma): [\underline{\gamma}, \bar{\gamma}] \rightarrow \mathbb{R} \), for any $\epsilon(N) \in \mathbb{R}^+$ arbitrarily small, there exists a sequence of polynomials \( \{ P_N \} \) of the $N^{th}$ order such that the supremum norm, denoted as \( \lVert \cdot \rVert \), is bounded by $\epsilon(N)$, as stated in Eq. \eqref{eq:weierstrass_thm}. \cite{weierstrass:1885}

    \begin{align}
        \lVert f(\gamma) - P_N \rVert \leq \epsilon(N) \label{eq:weierstrass_thm}
    \end{align}
    
    As $F^{-1}_{\bar{\theta}}(\cdot)$ is a bounded and smooth function, we can leverage the \textit{Weierstrass approximation theorem} to always obtain an $N^{th}$ order polynomial approximation such that if $N$ is large enough, $\lVert f(\gamma) - P_N \rVert < \epsilon(N)$ for any $\epsilon$ arbitrarily small. Applying a polynomial approximation and the Weirstrass approximation theorem brings us two main crucial benefits. First we can take advantage of having an arbitrarily small error $\epsilon(N)$ that is not respective of the sample size $t$, rather a separate controllable computational term $N$. Secondly, by approximating  $F^{-1}_{\bar{\theta}}(\cdot)$ with $P_N$ we consequently have a closed form expression for an approximation of $F^{-1}_{\bar{\theta}}(\cdot)$ with controllable error. This allows us to perform optimization over a polynomial rather than a some arbirtary concave function derived from a cumulative distribution function.

    \begin{align}
        R_A(T) &\leq \sum^T_{t = 1} a^*\mathfrak{B}(a^*) - \mathcal{G}_A(\bar{a}^t, \mathfrak{B}(\bar{a}^t) ) \label{eq:leader-regret-defn} \\
         &\leq \sum^T_{t = 1} a^*\mathfrak{B}(a^*) - \underline{\mathcal{G}}_A(\bar{a}^t, \mathfrak{B}(\bar{a}^t))
    \end{align}



    We characterize $\underline{\mathcal{G}}_A(\bar{a}^t, \mathfrak{B}(\bar{a}^t))$ as the hypothetical leader profit, derived within a best response framework, where the best response is an additive composite of the term $\Delta_b(\theta)$ under perfect information (where $\Delta_b(\theta) = 0$), and the term $\Delta_b(\gamma)$ under imperfect information (where $\Delta_b(\gamma) \geq 0$). In the context of Fig. \ref{fig:cum-dist-shift}, an optimistic assessment under uncertainty influences two components with respect to the inverse CDF. First, it affects the overall optimism in demand, inducing a systematic shift along the vertical axis in the negative direction. We denote this as the change as $\Delta_b(\theta)$. Second, it results in a descent along the CDF curve itself, signified by the diminishing value of $\gamma(a)$. We denote this as the change as $\Delta_b(\gamma)$. Both aspects additively contribute to the total change in $\mathfrak{B}(a)$. This formulation enables the construction of a lower bound for the best response mechanism, which is instrumental in subsequent analysis.
    
    \textbf{Solution via $N^{th}$ Order Polynomial Approximation:} In theory an infinite Taylor series expansion $P_N(\gamma)$ or $N$ polynomial terms, centred around $\underline{\gamma}$ can represent $F^{-1}_\theta(\gamma)$. 
    
    \begin{align}
        F^{-1}_\theta(\gamma) &= F^{-1}_\theta(\underline{\gamma}) + \frac{1}{1!}  \frac{d F^{-1}_\theta(\underline{\gamma})}{d\gamma} (\gamma - \underline{\gamma}) + \frac{1}{2!}  \frac{d^2 F^{-1}_\theta(\underline{\gamma})}{d\gamma^2} (\gamma - \underline{\gamma})^2 + \frac{1}{3!}  \frac{d^3 F^{-1}_\theta(\underline{\gamma})}{d\gamma^3} (\gamma - \underline{\gamma})^3 + \cdots \\
        &= \lim_{N \to \infty} F^{-1}_\theta(\underline{\gamma}) + \sum_{n=1}^{N} \frac{1}{n!}  \frac{d^n F^{-1}_\theta(\underline{\gamma})}{d\gamma^n} (\gamma - \underline{\gamma})^n
    \end{align}

    As there is no closed form solution for $F_\theta^{-1}(\cdot)$, we can apply a finite Taylor series expansion $\mathcal{P}_N(\gamma)$ to approximate $F_\theta^{-1}(\cdot)$. Let $\mathcal{P}_N(\gamma)$ denote the $N^{th}$ order Taylor expansion centred around $\underline{\gamma}$. Consequently, $F^{-1}_\theta(\gamma)$ is equivalent to the sum of $\mathcal{P}_N(\gamma)$ and the remainder term $R_N(\gamma)$, which decreases with increasing order of the polynomial series, $N$.

    \begin{align}
        \mathcal{P}_N(\gamma) &= F^{-1}_\theta(\underline{\gamma}) + \sum_{n=1}^{N} \frac{1}{n!}  \frac{d^n F^{-1}_\theta(\underline{\gamma})}{d\gamma^n} (\gamma - \underline{\gamma})^n \\
        F^{-1}_\theta(\gamma) &= \mathcal{P}_N(\gamma) + R_N(\gamma)
    \end{align}
    
    Let $\mathcal{V}(a, t)$ represent a polynomial approximation of $\underline{\mathcal{G}}(\bar{a}^t, \mathfrak{B}(\bar{a}^t))$, given $\mathcal{P}_N(\gamma)$. 

    \begin{align}
        \mathcal{V}(a, t) &= a \mathcal{P}_N \Big( \gamma(a) - \underline{\gamma} \Big) \\
        &= a \sum_{i = 0}^{N}  \beta_i \Big(\gamma(a) - \underline{\gamma} \Big)^n \\
        &= a \sum_{i = 0}^N  \beta_i \Big(1 - \frac{2a}{\mathcal{H}^*(\theta^*) + \Delta_\mathcal{H} + a} - \underline{\gamma} \Big)^n \label{eq:n-poly-approx} 
    \end{align}

    The structural form of Eq. \eqref{eq:n-poly-approx} allows us to exploit the Weierstrass approximation theorem to place an arbitrarily small error bound on the approximation of $\underline{\mathcal{G}}(\bar{a}^t, \mathfrak{B}(\bar{a}^t))$ via $\mathcal{V}(a, t)$, which can be expressed as the sum of finite polynomials, with some controllable additive error decreasing with $N$, which we denote as $\epsilon(N)$.
    
    \begin{align}
        \underline{\mathcal{G}}(\bar{a}^t, \mathfrak{B}(\bar{a}^t)) &= \mathcal{V}(a, t) + \epsilon(N)
    \end{align}

    We also proved that as the approximation error of $F^{-1}_\theta(\gamma)$, via $\mathcal{P}_N(\gamma)$ with the approaches $0$, then the distance between the estimated optima value and the true optimal value, denoted as $\widetilde{\Delta}_a$ also approaches 0. Thus by Lemma \ref{thm:shrink_x_eps}, $\epsilon(N) \xrightarrow[]{} 0 \implies \widetilde{\Delta}_a \xrightarrow[]{} 0$. Where $\widetilde{\Delta}_a$ represents the maximum discrepancy between the value of $a$ that maximizes $\mathcal{V}(a, t)$ and the value of $a$ that maximizes $\underline{\mathcal{G}}(\bar{a}^t, \mathfrak{B}(\bar{a}^t))$.
    
    \begin{align}
        \bar{a} &= \underset{a \in \mathcal{A} } { \mathrm{argmax}} \  \mathcal{V}(a, t) + \widetilde{\Delta}_a \label{eq:bound-bar-a-orig} 
    \end{align}

    But the Weierstrass approximation theorem alone is insufficient to complete our derivation to solve for $\bar{a}$, as we have to take the derivative of $\mathcal{P}_N(\gamma)$, and set $\frac{\partial \mathcal{P}_N(\gamma)}{\partial \gamma} = 0$. This has no simple closed form solution that generalizes to $N$ for $\mathcal{P}_N(\gamma)$. The challenge arises when differentiating $\mathcal{P}_N(\gamma)$ with respect to $\gamma$ and equating $\frac{\partial \mathcal{P}_N(\gamma)}{\partial \gamma}$ to $0$ to solve for $\bar{a}$. Identifying a concise closed-form solution for $\bar{a}$ that generalizes to $N$ for $\mathcal{P}_N(\gamma)$ remains a formidable challenge, which motivates the use of the \textit{$\gamma(a)$ trick}.
    

    \textbf{The $\gamma(a)$ trick:} The selection of $\underline{\gamma}$ as the center for the Taylor series expansion was deliberate means of simplification, stemming from the the fact that we impose of $F^{-1}_\theta(\underline{\gamma}) = 0$ in Eq. \eqref{eq:F_underline_gamma0}. To address the complexity of root-finding problem, a strategic approach is employed. We introduce $\mathcal{P}^{(k)}_N(\gamma)$ as the Taylor expansion of $F^{-1}\theta(\gamma)$ centered at $\underline{\gamma}$, while omitting the initial $k$ terms. To be clear,

    \begin{align}
        \mathcal{P}^{(k)}_N(\gamma) &= \sum_{i = k+1}^{N}  \beta_i \Big(\gamma(a) - \underline{\gamma} \Big)^n \\
        &=  \sum_{i = k+1}^{N}  \beta_i \Big(1 - \frac{2a}{\mathcal{H}^*(\theta^*) + \Delta_\mathcal{H} + a} - \underline{\gamma} \Big)^n  
    \end{align}

    Such that, at $k=2$. Also note that $\beta_0 = 0$, from the construct of our Taylor expansion as it was centeres at $\underline{\gamma}$.
    
    \begin{align}
        \mathcal{P}_N(\gamma) &= \beta_0 + \beta_1 \Big(\gamma(a) - \underline{\gamma} \Big) + \mathcal{P}^{(k=2)}_N(\gamma) 
    \end{align}

    We define $\Lambda(\gamma)$ as,
    
    \begin{align}
        \Lambda(\gamma) &\equiv \beta_0 + \beta_1 \Big(\gamma(a) - \underline{\gamma} \Big) \\
        &= \beta_0 + \beta_1 \Big(1 - \frac{2a}{\mathcal{H}^*(\theta^*) + \Delta_\mathcal{H} + a} - \underline{\gamma} \Big)
    \end{align}

    Subsequently, a linear affine transform on our function of interest $F^{-1}_\theta(\gamma)$ can be performed, where we denote this affine linear transform as $\phi(F^{-1}_\theta(\gamma))$. 
    
    \begin{align}
        \phi(F^{-1}_\theta(\gamma)) \equiv F^{-1}_\theta(\gamma) - \Lambda(\gamma)
    \end{align}
    

    We leverage the Weierstrass approximation theorem, as expressed in Eq. \eqref{eq:weierstrass_thm}, for the relation of $\phi(F^{-1}_\theta(\gamma))$ and $\mathcal{P}^{(k)}_N(\gamma)$, resulting in,
    This allows us to express the 
    
    $\Lambda(\gamma)$ is defined as a linear function, allowing us to subsequently substitute it into the Weierstrass approximation theorem, as expressed in Eq. \eqref{eq:weierstrass_thm}, resulting in,
    
    \begin{align}
        \lVert F^{-1}_\theta(\gamma) - \mathcal{P}_N(\gamma) \rVert = \lVert \phi(F^{-1}_\theta(\gamma)) - \mathcal{P}^{(k)}_N(\gamma) \rVert \leq \epsilon(N) \label{eq:weierstrass_thm_adj}
    \end{align}

    Now we take a look at the behaviour of $\mathcal{P}^{(k)}_N(\gamma(a))$ specifically. We wish to find the value of $a$ that would maximize $\mathcal{P}^{(k)}_N(\cdot)$, and by taking partial derivatives and setting the result to 0. Clearly, we can apply the product rule, followed by the chain rule for this.
    
    \begin{align}
        \frac{\partial}{\partial a} \mathcal{V}(a, t) &= \sum_{i=2}^N \beta_i \left(\gamma(a) - \underline{\gamma}\right)^n + \frac{\partial}{\partial a} (\gamma(a) - \underline{\gamma}) \sum_{i = 2}^N  \beta_i n  \Big( 1 - \frac{2a}{\mathcal{H}^*(\theta^*) + \Delta_\mathcal{H} + a} - \underline{\gamma} \Big)^{n-1} \\
        &= \sum_{i = 2}^N  \beta_i \Big(1 - \frac{2a}{\mathcal{H}^*(\theta^*) + \Delta_\mathcal{H} + a} - \underline{\gamma} \Big)^n \nonumber  \\
        &\quad \quad - \Big( \frac{2a(\mathcal{H}^*(\theta^*) + \Delta_\mathcal{H})}{(\mathcal{H}^*(\theta^*) + \Delta_\mathcal{H} +a)^2} \Big) \sum_{i = 2}^N  \beta_i n  \Big( 1 - \frac{2a}{\mathcal{H}^*(\theta^*) + \Delta_\mathcal{H} + a} - \underline{\gamma} \Big)^{n-1} = 0 \label{eq:a_deriv_k2_equality}
    \end{align}


    Here is where the $\gamma(a)$ trick becomes pivotal. We can now solve Eq. \eqref{eq:n-poly-approx} by solving, 

    \begin{align}
        \gamma(a) - \underline{\gamma} = 1 - \frac{2a}{\mathcal{H}^*(\theta^*) + \Delta_\mathcal{H} + a} - \underline{\gamma} &= 0
    \end{align}

    Which gives us the solution, 

    \begin{align}
        a &= \frac{\mathcal{H}^*(\theta^*) + \Delta_\mathcal{H} - \underline{\gamma}(\mathcal{H}^*(\theta^*) + \Delta_\mathcal{H})}{\underline{\gamma} + 1} \\
        &= \tilde{\gamma} \, (\mathcal{H}^*(\theta^*) + \Delta_\mathcal{H}), \quad \tilde{\gamma} = \frac{1 -\underline{\gamma}}{1 + \underline{\gamma}}
        \label{eq:a_opt_2_poly}
    \end{align}

    And by setting $a$ as equal to the value in Eq. \eqref{eq:a_opt_2_poly}, we effectively solve the partial derivative equation from Eq. \eqref{eq:a_deriv_k2_equality}. The result from Eq. \eqref{eq:a_opt_2_poly} allows us to simply set $a = \tilde{\gamma}(\mathcal{H}^*(\theta^*) + \Delta_\mathcal{H})$ to solve for $a$ in Eq. \eqref{eq:bound-bar-a-orig}. Essentially, by maximizing our approximation $\mathcal{V}(a, t)$, w.r.t. $a$, we force the condition such that $\Lambda(\gamma(a)) = 0$. Therefore, the bound $\lVert F^{-1}_\theta(\gamma) - \mathcal{P}_N(\gamma) \rVert$ expressed in Eq. \eqref{eq:weierstrass_thm_adj} holds.

    \begin{align}
        \lVert F^{-1}_\theta(\gamma) - \mathcal{P}_N(\gamma) \rVert = \lVert \phi(F^{-1}_\theta(\gamma)) - \mathcal{P}^{k=2}_N(\gamma) \rVert =\lVert F^{-1}_\theta(\gamma) - \mathcal{P}^{(k=2)}_N(\gamma) \rVert
    \end{align}

    The beauty is that we have an error term that can be arbitrarily small as a result $\epsilon(N)$. So then we can apply the lower bound from from Eq. \eqref{eq:GAB_lower_bound} and $\mathcal{V}(a, t)$ approximation of $\mathcal{G}_A(\bar{a}^t, \mathfrak{B}(\bar{a}^t))$, so obtain a bound on $\Delta_a(t)$. Given $\Delta_a(t) = \bar{a} - a^*$, we now bound $\mathcal{G}_A(\bar{a})$ given the bound on $\Delta_a$ and $\Delta_\gamma$,

    

    \begin{align}
        \mathcal{G}_A(\bar{a}^t, \mathfrak{B}(\bar{a}^t)) &= \Big(a^* + \Delta_a(t)\Big) F^{-1}_{\theta} \Big( \gamma^* - \Delta_\gamma(t) \Big) \\
        &\geq a^* F^{-1}_{\theta} \Big( \gamma^* - \Delta_\gamma(t) \Big) \\
        &= a^* \Big(\mathfrak{B}(a) - \Delta_b(t)\Big) \\
        &\geq a^* \Big( \mathfrak{B}(a) - \frac{1}{\mathscr{M}} \frac{4 \kappa_H }{\mathcal{H}^*(\theta^*)^2} (\sqrt{\log(t)/t}) \, \bar{a}  \Big) \label{eq:lower_bound_G_ab}
    \end{align}

     We aim to express the change in $b(a)$ with a constant scale with respect to a change in $\gamma(a)$. Thus we apply the relations from Eq. \eqref{eq:delta_gamma} and Eq. \eqref{eq:Mb_bound} in combination with Eq. \eqref{eq:lower_bound_G_ab} to obtain a bound on $\mathcal{G}_A(\bar{a}^t, \mathfrak{B}(\bar{a}^t))$ as expressed in Eq. \eqref{eq:bound-on-ga}. 

    \begin{align}
        \mathcal{G}_A(\bar{a}^t, \mathfrak{B}(\bar{a}^t))  &\geq a^* \Big( \mathfrak{B}(a) - \frac{1}{\mathscr{M}} \frac{4 \kappa_H }{\mathcal{H}^*(\theta^*)^2} (\sqrt{\log(t)/t}) \bar{a}  \Big) \\ 
        &\geq a^* \Big( \mathfrak{B}(a) - \frac{1}{\mathscr{M}} \frac{4 \kappa_H }{\mathcal{H}^*(\theta^*)^2} \sqrt{\log(t)/t} \Big( \tilde{\gamma} \, (\mathcal{H}^*(\theta^*) + \Delta_\mathcal{H}) - \widetilde{\Delta}_a \Big) \Big) \\
        &= a^* \Big( \mathfrak{B}(a) - \frac{1}{\mathscr{M}} \frac{4 \kappa_H }{\mathcal{H}^*(\theta^*)^2} \sqrt{\log(t)/t} \Big(  \tilde{\gamma} \, \Big( \mathcal{H}^*(\theta^*) + \kappa_H  \sqrt{\log(t)/t}  \Big) - \widetilde{\Delta}_a \Big) \Big) \label{eq:bound-on-ga-expr-c} \\
        &= a^*b^* + C_1\sqrt{\log(t)/t} + C_2 \log(t)/t + C_3 \widetilde{\Delta}_a \label{eq:bound-on-ga}
    \end{align}

    Effectively we solve the Eq. \eqref{eq:bound-on-ga-expr-c}, for three constants, defined as, given $\Delta_a(t) = \bar{a} - a^*$, we now bound $\mathcal{G}_A(\bar{a})$ given the bound on $\Delta_a$ and $\Delta_\gamma$, where, 
    
    \begin{align}
        C_1 &= - \frac{4 a^* \kappa_H ( \tilde{\gamma} \mathcal{H}^*(\theta^*) - \widetilde{\Delta}_a )}{\mathscr{M} \mathcal{H}^*(\theta^*)^2 } \\
        C_2 &= - \frac{4 a^* \kappa_H^2 \tilde{\gamma} }{\mathscr{M} \mathcal{H}^*(\theta^*)^2 } \\
        C_3 &= \frac{4 a^* \kappa_H \widetilde{\Delta}_a }{\mathscr{M} \mathcal{H}^*(\theta^*)^2 }
    \end{align}

    The novelty in our method is that we can force $\widetilde{\Delta}_a$ to be arbitrarily small as a consequence of Lemma \ref{thm:shrink_x_eps}, by increasing the order of the polynomial approximation such that, $\epsilon(N) < \min \{ |C_1|, |C_2| \}$. Given this we can then eliminate the term introduced by $\tilde{\omega}$ for big $\mathcal{O}$ analysis. Substituting Eq. \eqref{eq:bound-on-ga} into Eq. \eqref{eq:leader-regret-defn}, we obtain a bound on leader regret.

    \begin{align}
        R_A(T) &\leq \sum^T_{t = 1} a^*\mathfrak{B}(a^*) - a^* F^{-1}_{\theta} \Big( \gamma^* - \Delta_\gamma(t) \Big) \\
        &= \sum^T_{t = 1} C_1\sqrt{\log(t)/t} + C_2 \log(t)/t + C_3 \widetilde{\Delta}_a
    \end{align}

    By the same process as Lemma \ref{thm:cum-del-HT} we take the integral of the discrete upper bound and obtain a Big-O bound on leader regret. The constant $T$ term originates from estimation error of the inverse CDF via a linear function, we denote this error, $\tilde{\omega}$, to be sufficiently small given a suitable range $[\underline{b}, \bar{b}]$, and is external to the learning algorithm.
    
    \begin{align}
        R_A(T) \in \mathcal{O}(\sqrt{T \log(T)}) \label{eq:leader-regret}
    \end{align}

    \end{proof}
    
    
    

\subsection{Proof for Corollary \ref{cor:alt_br}} \label{prf:alt_br}

\textbf{Relaxation of the Follower's Objective:} Suppose we extend our analysis to another best-response proxy $\tilde{\mathfrak{B}}(a)$, which gives a closer solution to the \textit{price-setting Newsvendor} (PSN). Let $R_A^T(\pi_A)$ be the Stackelberg regret of the original Algorithm \ref{alg:se-newsv} (as also expressed in Eq. \eqref{eq:leader-regret}). As $a^t \to a^*$, $R_A^T(\pi_A)$ now becomes $\tilde{R}_A^T(\pi_A)$ with new target $\tilde{a}^*$ (where $\tilde{a}^*$ maximizes $a \tilde{\mathfrak{B}}(a)$),

\begin{align}
    R_A^T(\pi_A) = \sum^T_{t = 1} a^* \mathfrak{B}(a^*)  - \sum^T_{t = 1}  a^t \mathfrak{B}(a^t), \xrightarrow{\text{becomes}} \tilde{R}_A^T(\pi_A) = \sum^T_{t = 1} \tilde{a}^* \tilde{\mathfrak{B}}(\tilde{a}^*) - \sum^T_{t = 1}  a^t \tilde{\mathfrak{B}}(a^t)
\end{align}

 The leader is still learning the via the original best-response proxy $\mathfrak{B}(\cdot)$, converging to solution $a^*$, the regret for the new proxy $\tilde{R}_A^T(\pi_A)$ can be expressed as:

 \begin{align}
 \tilde{R}_A^T(\pi_A) = \sum^T_{t = 1} \Big( a^* \mathfrak{B}(a^*) - \Delta_G \Big) - \sum^T_{t = 1}  \Big( a^t \mathfrak{B}(a^t) - \Delta_G(t) \Big)
 \end{align}

 Where:

 \begin{align}
    \Delta_G = a^* \mathfrak{B}(a^*) - \tilde{a}^* \tilde{\mathfrak{B}}(\tilde{a}^*), \quad \Delta_G(t) = a^t \mathfrak{B}(a^t) - a^t \tilde{\mathfrak{B}}(a^t)
 \end{align}
 
 We can express the new regret with the original expression as:

\begin{align}
    \tilde{R}_A^T(\pi_A) &= \underbrace{\sum^T_{t = 1} a^* \mathfrak{B}(a^*)  - \sum^T_{t = 1}  a^t \mathfrak{B}(a^t)}_{R^T(\pi_A)} + \sum^T_{t = 1} \Delta_G(t) - \sum^T_{t = 1} \Delta_G  =  R_A^T(\pi_A) + \sum^T_{t = 1} \Big(\Delta_G(t) - \Delta_G \Big) \\
    &= R_A^T(\pi_A) + \sum^T_{t = 1} \Big( a^t \mathfrak{B}(a^t) - a^t \tilde{\mathfrak{B}}(a^t) - a^* \mathfrak{B}(a^*) + \tilde{a}^* \tilde{\mathfrak{B}}(\tilde{a}^*) \Big)\\
    &= R_A^T(\pi_A) + \underbrace{\sum^T_{t = 1} \Big( a^t \mathfrak{B}(a^t)  - a^* \mathfrak{B}(a^*) \Big)}_{c(T) \leq 0} + \underbrace{\sum^T_{t = 1} \Big( \tilde{a}^* \tilde{\mathfrak{B}}(\tilde{a}^*) - a^t \tilde{\mathfrak{B}}(a^t)}_{f(T) \geq 0} \Big) 
\end{align}

  \textbf{Notes on Asymptotic Performance:} In short, by rearranging the expression for $\tilde{R}_A^T(\pi_A)$, for sufficiently large sample size $T$, we can express: 

 \begin{align}
     \tilde{R}_A^T(\pi_A) \leq R_A^T(\pi_A) + c(T) + f(T) 
 \end{align}

 We note, that as as $T \to \infty$, $|c(T)| \leq C \sqrt{T \log(T)} $, for some constant $C$, by extension of Thm. 4. Furthermore, $\tilde{a}^* \tilde{\mathfrak{B}}(\tilde{a}^*) - a^t \tilde{\mathfrak{B}}(a^t) \to \hat{\epsilon}$, where $\hat{\epsilon} = \tilde{a}^* \tilde{\mathfrak{B}}(\tilde{a}^*) - a^* \tilde{\mathfrak{B}}(a^*)$, by divergence of the proxy objective. Additionally, $c(T) \leq 0$ is by definition since $a^*$ is the optimal solution to $a^* \mathfrak{B}(a^*)$ under the original risk-free pricing proxy, and $a^t \mathfrak{B}(a^t) \leq a^* \mathfrak{B}(a^*)$. Likewise $\hat{\epsilon} \geq 0$, as $\tilde{a}^*$ optimizes for $\tilde{a}^* \tilde{\mathfrak{B}}(\tilde{a}^*)$ and $\tilde{a}^* \tilde{\mathfrak{B}}(\tilde{a}^*) \geq a^t \tilde{\mathfrak{B}}(a^t)$. Therefore, for sufficiently large values of $T$:

 \begin{align}
     \tilde{R}_A^T(\pi_A) \in \mathcal{O}(\sqrt{T \log(T)} + \hat{\epsilon}T)
 \end{align}

 Because the PSN has no closed-form solution, the follower must always approximate their best response with some proxy $\tilde{\mathfrak{B}}(a)$, which will induce a suboptimality term $\hat{\epsilon}$ for the leader's Stackelberg regret. Depending on the fixed value of $\hat{\epsilon}$, w.r.t. the terms expressed in Eq. (C.159) of Appendix C.6 in the paper, $\hat{\epsilon}$ could be negligible. Evidently, the Stackelberg regret always relies on assumptions regarding the follower's approximation to the PSN problem, where $\hat{\epsilon}$ is computable given the $\tilde{\mathfrak{B}}(\cdot)$ and $\mathfrak{B}(\cdot)$. In other cases, $\hat{\epsilon}$ could be computed explicitly and added to the regret term, but it entirely depends on the approximation method. Fundamentally, we can see that this Stackelberg game always relies on assumptions regarding the follower's approximation to the PSN problem, and the current regret bounds can transfer to a new proxy, given the approximation error.

\subsection{Non-Linear Demand Function under Exponential Link Function} \label{sec:non_lin_demand_exp}

The current framework can accommodate certain non-linear demand functions by employing Generalized Linear Models (GLMs). We apply a link function $\Phi(\cdot)$ with a well-defined functional form in order to establish a linear relationship: $\Phi(\Gamma(p)) = \theta_0 - \theta_1 p$. Consider the link function $\Phi(x) = \log(x)$, which implies $\Gamma(p) = \exp(\theta_0 - \theta_1 p)$. This choice simplifies the analysis for Theorem \ref{thm:bound-del-ht}. Specifically, the maximization operator from Eq. \eqref{eq:h-max-theta} becomes $\mathcal{H}(\theta) = \hat{\theta}_0$, where $\hat{\theta}_1$ can be disregarded as it does not influence the optimum of the new expected revenue function $p \exp(\theta_0 - \theta_1 p)$. This can be verified through calculus or graphical methods by evaluating the solution of the optimum of $p \exp(\theta_0 - \theta_1 p)$.

Thus, our updated objective is to maximize $\theta_0$ alone, subject to the constraints given by:

\begin{align}
    \text{Maximize: } & \theta_0, \\
    \text{Subject to: } & \sqrt{(\hat{\theta}_0 - \theta_0^*)^2} \leq \frac{\kappa \sqrt{\log(t)}}{\sqrt{t}}.
\end{align}

This adjustment simplifies the analysis from Appendix Eq. \eqref{eq:hstar-def} to Eq. \eqref{eq:del_h_raw_logt_t}, focusing solely on the estimated variable $\hat{\theta}_0$. It is straightforward to demonstrate that the bound in Appendix Eq. \eqref{eq:del_h_raw_logt_t} remains valid under this new objective function. Therefore, the non-linear link function $\Phi(x) = \log(x)$ preserves the bounds on Stackelberg Regret established in Theorem 4. It follows that, the rest of the theoretical results regarding Stackelberg regret continue to hold. For any hypothetical non-linear demand function, the applicability of our theoretical results will depend on the specific functional form of the GLM link function. However, we posit that for any demand function of the logarithmic or exponential class, the theoretical contributions of this paper remain valid.

\section{Equilibrium Proofs}

\subsection{Proof of Theorem \ref{thm:se-convergence}}\label{prf:se-convergence}
    
    \textbf{Convergence to Equilibrium (Algorithm \ref{alg:se-newsv}):} The LNPG algorithm (Algorithm \ref{alg:se-newsv}) converges to an approximate Stackelberg equilibrium, as defined in Section \ref{sec:stack-eq-defn}.
\begin{proof}

    \textbf{Follower Best Response:} By Lemma \ref{thm:opt-p-star-decrease}, so long as  $\bar{p}_0$ converges to $p_0$, then so does $ \bar{p}^*(t) \xrightarrow[]{} p^*$. We also know from Theorem \ref{thm:bound-del-ht} that $\bar{p}_0$ will converge to $p_0$ in the asymptotic infinite time horizon. Where,
    
    \begin{align}
        p_0 &\geq p^* \xrightarrow[]{} F_{\theta}^{-1} \Big(1 - \frac{2a}{ \mathcal{H}^*(\theta^*) + a} \Big) \geq F_{\theta}^{-1} \Big(1 - \frac{a}{ p^* } \Big) 
    \end{align}
    
    We can solve $p^* - p_0$ from Eq. \eqref{eq:mills-opt-price}, and we can also compute a solution the the CDF, $F_{\theta}^{-1} \Big(1 - \frac{a}{ p^* } \Big)$. We define,
    \begin{align}
        \mathbbm{E}[\mathcal{G}_B(p|a)] &= \mathbbm{E}[\mathcal{M}_\theta(p|a)]  - \mathbbm{E}[a F_{\theta}^{-1} (1 - a/ p)] \\
        \text{where,} \quad \mathbbm{E}[\mathcal{M}_\theta(p|a)] &\equiv \mathbbm{E}[p \min (d_\theta(p), F_{\theta}^{-1} (1 - a/ p)]
    \end{align}

    $\mathbbm{E}[\mathcal{M}_\theta(p|a)]$ is a controllable variable which consequently maximizes $\mathbbm{E}[\mathcal{G}_B(p|a)]$. As $F_{\theta}^{-1}(\cdot)$ is a deterministic function, we note that $\mathbbm{E}[a F_{\theta}^{-1} (1 - a/ p)] = a F_{\theta}^{-1} (1 - a/ p)$ Therefore,

    \begin{align}
        \mathbbm{E}[\mathcal{G}_B(p^*|a)] &= \mathbbm{E}[\mathcal{M}_\theta(p^*|a)] - a F_{\theta}^{-1} (1 - a/ p^*) \\
        \mathbbm{E}[\mathcal{G}_B(p_0|a)] &=  \mathbbm{E}[\mathcal{M}_\theta(p_0|a)] - a F_{\theta}^{-1} (1 - a/p_0) \\
        &= \mathbbm{E} \Big[ \mathcal{M}_\theta \Big( \frac{\mathcal{H}^*(\theta^*) + a}{2}\Big) \Big] - a F_{\theta}^{-1} \Big(1 - \frac{2a}{\mathcal{H}^*(\theta^*) + a} \Big)
    \end{align}
    
    From Theorem \ref{thm:follower-regret}, the follower will always have some suboptimality upper bounded by $\epsilon_B > 0$, given an upper-bounding approximation function for $\mathcal{G}_B(\cdot)$, where $\epsilon_B$ denotes the gap to the optimal reward under the best response $p^*$ versus $p_0$ given $a$.

    \begin{align}
        \epsilon_B(a) = \mathbbm{E}[\mathcal{G}_B(p^*|a)] - \mathbbm{E}[\mathcal{G}_B(p_0|a)]
    \end{align}

    The optimistic estimate of $\bar{p}_0$ approaches $p_0$, with a worst case error of $\mathcal{O}(\sqrt{\log(t)/t})$ by Theorem \ref{thm:bound-del-ht}. Specifically we can express, $\mathbbm{E}[\mathcal{G}_B(\bar{p}_0|a)]$ as,

    \begin{align}
         \mathbbm{E}[\mathcal{G}_B(\bar{p}_0|a)] &= \mathbbm{E}\Big[ \mathcal{M}_\theta \Big( \frac{\mathcal{H}^*(\theta^*) + \Delta_{\mathcal{H}}(t) + a}{2}\Big) \Big] - a F_{\bar{\theta}}^{-1} \Big(1 - \frac{2a}{\mathcal{H}^*(\theta^*) + \Delta_{\mathcal{H}}(t) + a} \Big)
    \end{align}

    Let $\Delta_B(t)$ denote the suboptimality gap over time.

    \begin{align}
         \Delta_B(t|a) &= \mathbbm{E}[\mathcal{G}_B(p^*|a)] - \mathbbm{E}[\mathcal{G}_B(\bar{p}_0|a)] \\
         &= \mathbbm{E}[\mathcal{G}_B(p^*|a)] - \mathbbm{E}[\mathcal{G}_B(p_0 + \Delta_{\mathcal{H}}(t) |a)] \\
         &= \mathbbm{E}[\mathcal{G}_B(p^*|a)] - \mathbbm{E}[\mathcal{M}_\theta(\bar{p}_0)] + a F_{\theta}^{-1} (1 - a/\bar{p}_0) \\
        &= \mathbbm{E}[\mathcal{G}_B(p^*|a)] - \mathbbm{E}\Big[ \mathcal{M}_\theta \Big( \frac{\mathcal{H}^*(\theta^*) + \Delta_{\mathcal{H}}(t) + a}{2}\Big) \Big] + a F_{\bar{\theta}}^{-1} \Big(1 - \frac{2a}{\mathcal{H}^*(\theta^*) + \Delta_{\mathcal{H}}(t) + a} \Big)   
    \end{align}

    As $t \xrightarrow[]{} \infty$, $\bar{\theta} \xrightarrow[]{} \theta^*$, and $\Delta_{\mathcal{H}}(t) \xrightarrow[]{} 0$. Therefore,

    \begin{align}
        \underset{a \in \mathcal{A}}{\mathrm{max}} \ \Big(\lim_{t \to \infty} \ \Delta_B(t|a) \Big) = \epsilon_B
    \end{align}
    
    \textbf{Follower's Perspective:}  As the follower will always be subject to a suboptimality gap of at most $\epsilon_B$, the learned policy constitutes an $\epsilon$-approximate Stackelberg equilibrium, with $\epsilon_B$ denoting the maximum difference in the follower's expected reward due to the approximate best response (as defined in Eq. \eqref{eq:br_approx_def}).
    
    \textbf{Leader's Perspective:} In this problem setting, the follower does not attempt to alter the strategy of the leader via its own strategy, and only the optimal order amount computed by $b = F_{\theta}^{-1}(\gamma)$ affects the reward of the leader. Assuming the follower is abiding by an optimistic riskless pricing $\bar{p}_0$, we suppose that there exists some uncertainty with respect to the level of optimism for the follower response, in which $\mathfrak{B}(a) \in \mathfrak{B}_\varepsilon(a) \equiv [\underline{b_a}, \bar{b}_a]$. 

    
    
    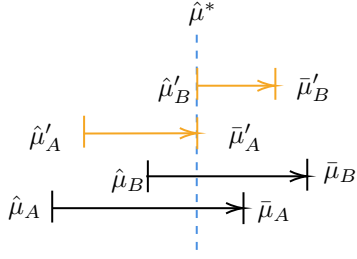
\begin{figure}[H] 
        \centering
        \tikzset{every picture/.style={line width=0.75pt}} 
            \begin{tikzpicture}[x=0.60pt,y=0.60pt,yscale=-1,xscale=1]
                \draw    (81.33,157.33) -- (199.33,157.33) ;
                \draw [shift={(201.33,157.33)}, rotate = 180] [color={rgb, 255:red, 0; green, 0; blue, 0 }  ][line width=0.75]    (10.93,-3.29) .. controls (6.95,-1.4) and (3.31,-0.3) .. (0,0) .. controls (3.31,0.3) and (6.95,1.4) .. (10.93,3.29)   ;
                \draw [dashed, color={rgb, 255:red, 74; green, 144; blue, 226 }  ,draw opacity=1 ]   (172,48.5) -- (172,187.5) ;
                \draw    (201.33,147.33) -- (201.33,167.33) ;
                \draw    (140.33,137.33) -- (239.33,137.33) ;
                \draw [shift={(241.33,137.33)}, rotate = 180] [color={rgb, 255:red, 0; green, 0; blue, 0 }  ][line width=0.75]    (10.93,-3.29) .. controls (6.95,-1.4) and (3.31,-0.3) .. (0,0) .. controls (3.31,0.3) and (6.95,1.4) .. (10.93,3.29)   ;
                \draw    (241.33,126.33) -- (241.33,146.33) ;
                \draw    (81.33,146.33) -- (81.33,166.33) ;
                \draw    (141.33,127.33) -- (141.33,147.33) ;
                \draw [color={rgb, 255:red, 245; green, 166; blue, 35 }  ,draw opacity=1 ]   (101,110.5) -- (170.33,110.34) ;
                \draw [shift={(172.33,110.33)}, rotate = 179.87] [color={rgb, 255:red, 245; green, 166; blue, 35 }  ,draw opacity=1 ][line width=0.75]    (10.93,-3.29) .. controls (6.95,-1.4) and (3.31,-0.3) .. (0,0) .. controls (3.31,0.3) and (6.95,1.4) .. (10.93,3.29)   ;
                \draw [color={rgb, 255:red, 245; green, 166; blue, 35 }  ,draw opacity=1 ]   (172.33,100.33) -- (172.33,120.33) ;
                \draw [color={rgb, 255:red, 245; green, 166; blue, 35 }  ,draw opacity=1 ]   (101.33,99.33) -- (101.33,119.33) ;
                \draw [color={rgb, 255:red, 245; green, 166; blue, 35 }  ,draw opacity=1 ]   (172,80.5) -- (219,80.5) ;
                \draw [shift={(221,80.5)}, rotate = 180] [color={rgb, 255:red, 245; green, 166; blue, 35 }  ,draw opacity=1 ][line width=0.75]    (10.93,-3.29) .. controls (6.95,-1.4) and (3.31,-0.3) .. (0,0) .. controls (3.31,0.3) and (6.95,1.4) .. (10.93,3.29)   ;
                \draw [color={rgb, 255:red, 245; green, 166; blue, 35 }  ,draw opacity=1 ]   (172.33,69.33) -- (172.33,89.33) ;
                \draw [color={rgb, 255:red, 245; green, 166; blue, 35 }  ,draw opacity=1 ]   (221.33,69.33) -- (221.33,89.33) ;
                
                \draw (208.25,151.58) node [anchor=north west][inner sep=0.75pt]   [align=left] {$\bar{\mu}_A$};
                \draw (51.8,147.51) node [anchor=north west][inner sep=0.75pt]   [align=left] {$\hat{\mu}_A$};
                \draw (116.44,130.43) node [anchor=north west][inner sep=0.75pt]   [align=left] {$\hat{\mu}_B$};
                \draw (249.69,127.16) node [anchor=north west][inner sep=0.75pt]   [align=left] {$\bar{\mu}_B$};
                \draw (165,25.96) node [anchor=north west][inner sep=0.75pt]   [align=left] {$\hat{\mu}^*$};
                \draw (65,102.51) node [anchor=north west][inner sep=0.75pt]   [align=left] {$\hat{\mu}_A'$};
                \draw (145.8,72.51) node [anchor=north west][inner sep=0.75pt]   [align=left] {$\hat{\mu}_B'$};
                \draw (190,102.16) node [anchor=north west][inner sep=0.75pt]   [align=left] {$\bar{\mu}_A'$};
                \draw (232.69,71.16) node [anchor=north west][inner sep=0.75pt]   [align=left] {$\bar{\mu}_B'$};
            \end{tikzpicture}
          \caption{Visualizing the optimistic $\bar{\mu} = \bar{\theta}_0$. $\mu^*$ represents the true $\theta_0^*$. Suppose A and B have different estimates of $\theta_0$, where $\mu^*$ lines between $[\hat{\mu}, \bar{\mu}]$ with probability $1 - \delta$. The orange lines for $[\hat{\mu}_A', \bar{\mu}_A']$ and $[\hat{\mu}_B', \bar{\mu}_B']$ denote the worst case scenario for $\Delta_{\mu} = \bar{\mu}_B' - \hat{\mu}_A'$ due to misspecification.} 
                    
          \label{fig:optimism-illus}
    \end{figure} 

    We define $\theta$ as a 2-element vector $[\theta_0, \theta_1]^T$, where $\hat{\theta} \equiv [\hat{\theta}_0, \hat{\theta}_1]^\intercal$, and $\theta^* \equiv [\theta_0^*, \theta_1^*]^\intercal$. We let $\hat{\theta}_A$ and $\hat{\theta}_B$ denote the estimates of $\theta^*$ for players A and B respectively. Leveraging Lemma \ref{thm:theta-2-norm}, we bound the confidence ranges for the estimates of $\theta^*$, as expressed in Eq. \eqref{eq:u_a_defn} to \eqref{eq:u_b_defn}. Thus we construct the following uncertainty ranges denoted as $\mathbf{u}_A$ and $\mathbf{u}_B$,

    \begin{align}
        \mathbf{u}_A &=  \left\{ \theta : \norm{\theta^* - \hat{\theta}}_2 \leq \kappa \sqrt{\log(t_A)/t_A} \right\} \label{eq:u_a_defn} \\
        \mathbf{u}_B &=  \left\{ \theta : \norm{\theta^* - \hat{\theta}}_2 \leq \kappa \sqrt{\log(t_B)/t_B} \right\} \label{eq:u_b_defn} 
    \end{align}

     We see that, due to potential differences in sample sizes $t_A$ and $t_B$, the confidence regions for $\mathbf{u}_A$ and $\mathbf{u}_B$, can exhibit differences in Lebesgue measures, highlighting the effect of varying sample sizes on parameter estimation in the respective domains (for example A may have a different estimate of $\theta^*$ than B, due to her being more confident of her market demand estimate as a consequence of being in the market for a longer time and obtaining more observations than B). Consequently, this enforces, $t_A > t_B$. For the purpose of this analysis, without loss of generality, we suppose A is always more confident than B (the opposite can also be argued without altering the mechanics of the proof).
    
    It follows that, with probability $1- \delta^2$, $\theta^*$ must lie within the union of the two uncertainty ranges, which overlap each other.

    \begin{align}
        \theta^* \in \mathbf{u}_A \cup \mathbf{u}_B \label{eq:u_union} \\
        \mathbf{u}_A \cap \mathbf{u}_B \neq \emptyset \label{eq:u_intersect}
    \end{align}
    
    Eq. \eqref{eq:u_union} and Eq. \eqref{eq:u_intersect} represent the union and intersection of the two confidence regions, and due to Eq. \eqref{eq:u_intersect}, the two confidence regions overlap. While the existence of $\theta^*$ outside of the union of the two confidence regions is less than probability $\delta$. The conditions imply the existence of a non-empty subset common to both sets, ensuring the non-disjointness of the confidence regions. 

    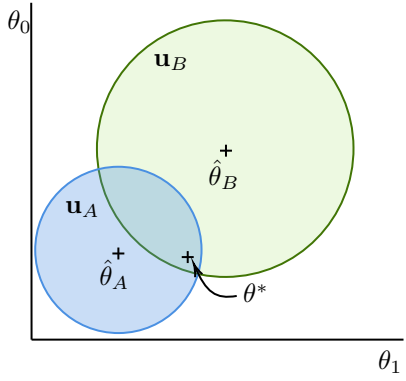
\begin{figure}[H] 
        \centering
        \tikzset{every picture/.style={line width=0.75pt}} 
        \begin{tikzpicture}[x=0.45pt,y=0.45pt,yscale=-1,xscale=1]
            
            \draw    (131.6,63.28) -- (131.88,345.05) ;
            \draw    (131.88,345.05) -- (443,345.05) ;
            \draw  [color={rgb, 255:red, 65; green, 117; blue, 5 }  ,draw opacity=1 ][fill={rgb, 255:red, 184; green, 233; blue, 134 }  ,fill opacity=0.25 ] (186.56,186.27) .. controls (186.03,127.01) and (233.63,78.53) .. (292.89,77.99) .. controls (352.15,77.46) and (400.63,125.06) .. (401.17,184.32) .. controls (401.7,243.58) and (354.1,292.06) .. (294.84,292.6) .. controls (235.58,293.13) and (187.1,245.53) .. (186.56,186.27) -- cycle ;
            \draw  [color={rgb, 255:red, 74; green, 144; blue, 226 }  ,draw opacity=1 ][fill={rgb, 255:red, 74; green, 144; blue, 226 }  ,fill opacity=0.3 ] (134.95,270.63) .. controls (134.61,232.22) and (165.46,200.8) .. (203.87,200.45) .. controls (242.28,200.11) and (273.7,230.96) .. (274.05,269.37) .. controls (274.39,307.78) and (243.54,339.2) .. (205.13,339.55) .. controls (166.72,339.89) and (135.3,309.04) .. (134.95,270.63) -- cycle ;
            \draw    (204.5,268) -- (204.5,278) ;
            \draw    (199.5,273) -- (209.5,273) ;
            \draw    (294.5,182) -- (294.5,192) ;
            \draw    (289.5,187) -- (299.5,187) ;
            \draw    (262.5,271) -- (262.5,281) ;
            \draw    (257.5,276) -- (267.5,276) ;
            \draw    (303.5,308.5) .. controls (285.95,311.43) and (277.91,305.79) .. (268.25,283.27) ;
            \draw [shift={(267.5,281.5)}, rotate = 67.38] [color={rgb, 255:red, 0; green, 0; blue, 0 }  ][line width=0.75]    (10.93,-3.29) .. controls (6.95,-1.4) and (3.31,-0.3) .. (0,0) .. controls (3.31,0.3) and (6.95,1.4) .. (10.93,3.29)   ;
            
            \draw (109.18,66.91) node [anchor=north west][inner sep=0.75pt]   [align=left] {$\theta_0$};
            \draw (419.32,351.63) node [anchor=north west][inner sep=0.75pt]   [align=left] {$\theta_1$};
            \draw (186.18,275.91) node [anchor=north west][inner sep=0.75pt]   [align=left] {$\hat{\theta}_A$};
            \draw (277.5,192) node [anchor=north west][inner sep=0.75pt]   [align=left] {$\hat{\theta}_B$};
            \draw (306.5,296) node [anchor=north west][inner sep=0.75pt]   [align=left] {$\theta^*$};
            \draw (231.5,104) node [anchor=north west][inner sep=0.75pt]   [align=left] {$\mathbf{u}_B$};
            \draw (157.5,227) node [anchor=north west][inner sep=0.75pt]   [align=left] {$\mathbf{u}_A$};
            
            \end{tikzpicture}
            \caption{Visualizing the confidence intervals $\mathbf{u}_A$ and $\mathbf{u}_B$ as two ellipsoid regions, where $\theta^*$ must fall within the intersection of the two regions. We can see that as the volume of both ellipsoids shrink with time $t$, the uncertainty in $\theta^*$ is also decreasing.} 
                            
                  \label{fig:optimism-illus-v2}
            \end{figure} 
    

    As both $t_A$ and $t_B$ approach infinity, any optimistic estimate of $\theta^*$, denoted as $\bar{\theta}^*$, for either A or B, would approach $\hat{\theta}_A$ and $\hat{\theta}_B$ respectively. We see that from Eq. \eqref{eq:u_a_defn} and Eq. \eqref{eq:u_b_defn}, both $\hat{\theta}_A$ and $\hat{\theta}_B$ approach the true parameters $\theta^*$. By the squeezing of both of the overlapping confidence regions as $\min\{t_A, t_B \} \to \infty$, any action the leader takes leads to a tighter possibility of the best responses of the follower under optimism. Thus, with probability $1-\delta$, we have convergence ta unique approximate best response as the best response range $\widetilde{\mathfrak{B}}(a) \in \mathfrak{B}_\varepsilon(a) \equiv [\underline{b_a}, \bar{b}_a]$ is shrinking, and $\bar{b}_a - \underline{b_a} \xrightarrow[]{} 0$ as $\min\{t_A, t_B \} \to \infty$. Consequently, the system converges to a unique best response with increasing sample size, $\mathfrak{B}_\varepsilon(a) \xrightarrow[]{} \{\widetilde{\mathfrak{B}}(a) \}$. The convergence of the algorithm to $\epsilon_B$ completes the definition for the $\epsilon$-approximate Stackelberg equilibrium, as defined in Section \ref{sec:stack-eq-defn},

    \end{proof}


\section{Experimental Results} \label{sec:experiment-results}

We present additional experiments, experimenting with varying environmental and LNPG algorithm parameters. Empirical analysis across multiple settings consistently demonstrates LNPG's convergence to stable equilibrium and its superiority over the baseline algorithm (UCB), particularly evident when minimizing Stackelberg regret.


\begin{figure}[H]
\minipage{0.5\textwidth}
  \includegraphics[width=\linewidth]{figures/reward_323-r-v2.png}
\endminipage\hfill
\minipage{0.5\textwidth}
  \includegraphics[width=\linewidth]{figures/regret_323-r-v2.png}
\endminipage\hfill
\caption*{Parameters: $\kappa=1.0, \theta_0 = 73, \theta_1 = 7, \sigma=3.2$.}
\minipage{0.5\textwidth}
  \includegraphics[width=\linewidth]{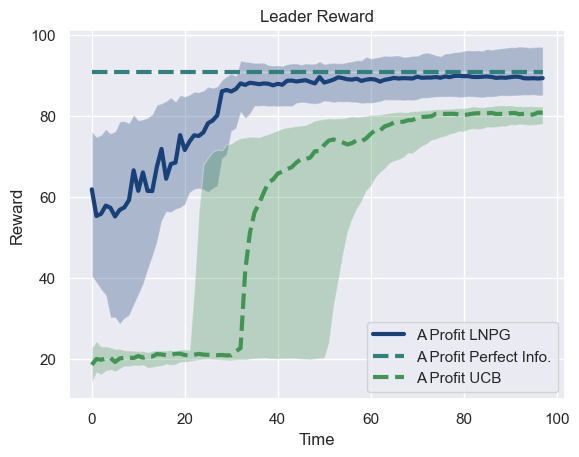}
\endminipage\hfill
\minipage{0.5\textwidth}
  \includegraphics[width=\linewidth]{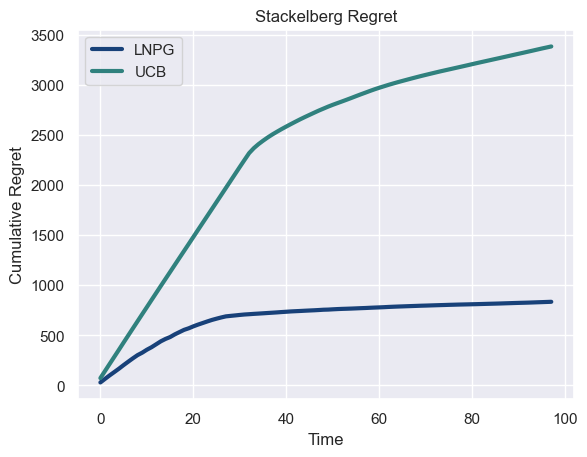}
\endminipage\hfill
\caption*{Parameters: $\kappa=1.0, \theta_0 = 80, \theta_1 = 11, \sigma=3.2$.}
\minipage{0.5\textwidth}
  \includegraphics[width=\linewidth]{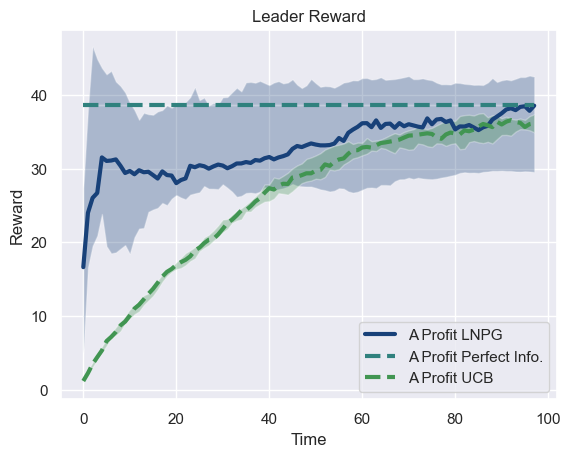}
\endminipage\hfill
\minipage{0.5\textwidth}
  \includegraphics[width=\linewidth]{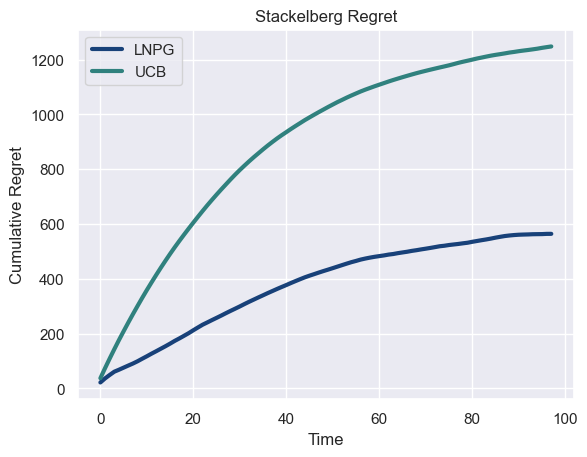}
\endminipage\hfill
\caption*{Parameters: $\kappa=0.05, \theta_0 = 40, \theta_1 = 5, \sigma=5.8$.}
\caption{Regret performance of various experiments. Mean values are computed over 20,000 trials. All shaded areas, denoting confidence intervals, are within a quarter quantile. } \label{fig:loss_plot}
\end{figure}


\begin{table}[H]
    \centering
    \begin{tabular}{||c c||}
     \hline
     Parameter & Value \\ [0.5ex] 
     \hline\hline
     Number of Arms (UCB) & 50 \\
     Pricing Range & $[0, 50]$ \\
     Number of Trials & 20,000 \\
     \hline
    \end{tabular}
    \caption{Global hyperparameters. }
    \label{table:hyperparam}
\end{table}

\end{document}